\documentclass[11pt]{article}
\usepackage{slantsc}
\usepackage{lmodern}
\usepackage[hypertexnames=false,pagebackref,colorlinks=true,citecolor=blue]{hyperref}
\usepackage{amsmath}
\usepackage{amsthm}
\usepackage{thmtools}
\usepackage{amssymb}
\usepackage{graphicx}
\usepackage{multicol}
\usepackage{multirow}
\usepackage{color}
\usepackage[dvips,letterpaper,margin=1in,bottom=1in]{geometry}
\usepackage[capitalize]{cleveref}

\usepackage[utf8]{inputenc}
\usepackage[english]{babel}
\usepackage{mathtools}
\usepackage{quantikz}
\usepackage{comment}
\usepackage{booktabs}

\usepackage{interval}
\intervalconfig{soft open fences}

\newtheorem{theorem}{Theorem}[section]

\newtheorem{lemma}[theorem]{Lemma}
\newtheorem{corollary}[theorem]{Corollary}
\newtheorem{proposition}[theorem]{Proposition}

\newtheorem{definition}[theorem]{Definition}

\renewcommand{\braket}[2]{\left< #1 \vphantom{#2} \middle| #2 \vphantom{#1} \right>}
\newcommand{\ketbra}[2]{\ensuremath{\ket{#1}\!\bra{#2}}}

\DeclarePairedDelimiter\rbra{\lparen}{\rparen}
\DeclarePairedDelimiter\sbra{\lbrack}{\rbrack}
\DeclarePairedDelimiter\cbra{\{}{\}}
\DeclarePairedDelimiter\abs{\lvert}{\rvert}
\DeclarePairedDelimiter\Abs{\lVert}{\rVert}

\DeclarePairedDelimiter\floor{\lfloor}{\rfloor}
\let\ket\relax
\DeclarePairedDelimiter\ket{\lvert}{\rangle}
\let\bra\relax
\DeclarePairedDelimiter\bra{\langle}{\rvert}

\newcommand{\tr} {\mathrm{Tr}}
\newcommand{\poly} {\operatorname{poly}}
\newcommand{\diag} {\operatorname{diag}}
\newcommand{\polylog} {\operatorname{polylog}}

\usepackage{enumitem}

\newcommand{\ignore}[1]{}

\DeclareMathOperator*{\Prob}{\mathrm{Pr}}
\DeclareMathOperator*{\E}{\mathbb{E}}
\DeclareMathOperator*{\Ex}{\mathbb{E}}

\newcommand{\R}{\mathbb R}
\newcommand{\C}{\mathbb C}
\newcommand{\N}{\mathbb N}
\newcommand{\Z}{\mathbb Z}

\newcommand{\eps}{\varepsilon}

\newcommand{\calA}{\mathcal{A}}

\newcommand{\calE}{\mathcal{E}}
\newcommand{\calF}{\mathcal{F}}

\newcommand{\calI}{\mathcal{I}}
\newcommand{\calJ}{\mathcal{J}}

\newcommand{\calO}{\mathcal{O}}
\newcommand{\calP}{\mathcal{P}}

\newcommand{\calU}{\mathcal{U}}

\DeclarePairedDelimiterX\diverg[2]{(}{)}{#1 \,\|\, #2}

\usepackage{algcompatible}
\algnewcommand{\algorithmicparameter}{ \textbf{Parameter:}}
\algnewcommand{\PARAMETER}[1]{%
  \item[\algorithmicparameter] #1
}
\usepackage{algorithm}
\makeatletter
\edef\ftype@algorithm{\the\c@float@type}
\makeatother

\usepackage{footnotehyper} 
\makesavenoteenv{table}  

\newcommand{\footremember}[2]{%
    \footnote{#2}
    \newcounter{#1}
    \setcounter{#1}{\value{footnote}}%
}

\usepackage{tikz}

\newcommand{\heavy}{\textrm{heavy}}
\newcommand{\tail}{\textrm{tail}}
\newcommand{\FAIL}{\mathsf{FAIL}}
\newcommand{\epsest}{\eps_{\mathrm{est}}}
\newcommand{\delest}{\delta_{\mathrm{est}}}

\newcommand{\epsFP}{\eps_{\mathrm{FP}}}
\newcommand{\delFP}{\delta_{\mathrm{FP}}}

\newcommand{\muth}{\mu_{\mathrm{th}}}

\newcommand{\restrictedmu}{\mu|_{\calI_\heavy}}
\newcommand{\barrestrictedmu}{\bar{\mu}|_{\calI_\heavy}}
\newcommand{\barrestrictednu}{\bar{\nu}|_{\calI_\heavy}}

\newcommand{\SIJ}{\sigma_i\sigma_j\sigma_i}
\newcommand{\SIJs}[1][s]{(\SIJ)^{#1}}
\newcommand{\pSIJ}{p(\SIJ)}

\newcommand{\KBp}[1]{\ketbra{\psi_{#1}}{\psi_{#1}}}

\usepackage{thmtools,thm-restate} 
\usepackage{regexpatch}
\makeatletter
\xpatchcmd\thmt@restatable{%
\csname #2\@xa\endcsname\ifx\@nx#1\@nx\else[{#1}]\fi
}{%
\ifthmt@thisistheone
\csname #2\@xa\endcsname\ifx\@nx#1\@nx\else[{#1}]\fi
\else
\csname #2\@xa\endcsname[{restated}]
\fi}{}{}
\makeatother

\title{Towards Optimal Quantum Estimators for State Frame Potential}

\author{Jinge Bao\footremember{1}{University of Edinburgh. \href{mailto:jingebao1011@gmail.com}{\nolinkurl{jingebao1011@gmail.com}}}
\and
Wang Fang\footremember{2}{United Arab Emirates University and University of Edinburgh. \href{mailto:njuwfang@gmail.com}{\nolinkurl{njuwfang@gmail.com}}}
\and
Yoshifumi Nakata\footremember{3}{Institute of Science Tokyo and Kyoto University. \href{mailto:nakata@comp.isct.ac.jp}
{\nolinkurl{nakata@comp.isct.ac.jp}}}
\and
Qisheng Wang\footremember{4}{Shanghai Jiao Tong University. \href{mailto:QishengWang1994@gmail.com}{\nolinkurl{QishengWang1994@gmail.com}}}
}

\AddToHook{cmd/maketitle/before}{%
  \AddToHookNext{shipout/foreground}{%
    \begin{tikzpicture}[remember picture,overlay]
      \node[anchor=east]
        at ([xshift=-2cm,yshift=-2cm]current page.north east) {YITP-26-89};
    \end{tikzpicture}%
  }%
}

\date{}

\begin{document}

\maketitle

\begin{abstract}
    The state frame potential is a standard diagnostic of how closely a quantum state ensemble approximates Haar randomness.
    In this work, we study the problem of estimating the state frame potential of order $t$ to within additive error $\varepsilon$ under three progressively weaker access models: (i) query access to a multi-state-preparation oracle, (ii) general sample access, and (iii) single-copy sample access.
    In the query model, we establish a near-optimal query complexity of $\widetilde{\Theta}(\sqrt{t}/\varepsilon)$, yielding a quadratic improvement in the dependence on $t$ over the previous best result of \hyperlink{cite.NTKD25}{Nakata, Takeuchi, Kliesch, and Darmawan (\textit{PRX Quantum} 2025)}.
    In the general sample model, we establish the optimal sample complexity $\Theta(t/\varepsilon^2)$.
    In the single-copy sample model, we present a \emph{store-and-estimate} approach whose sample complexity depends on the R\'enyi entropy of the ensemble weights.
    As an application, we use the single-copy algorithm to assess the randomness of projected state ensembles, where the entropy term becomes the observational R\'{e}nyi entropy associated with measuring one subsystem.
\end{abstract}

\clearpage

\tableofcontents

\clearpage

\section{Introduction}\label{sec:introduction}

Quantum randomness serves as a fundamental resource in quantum information processing.
Two canonical examples are Haar-random unitaries and Haar-random states, drawn, respectively, from the Haar measure on the unitary group and the unitarily invariant measure on the unit sphere of a Hilbert space.
Applications of quantum randomness range from algorithms~\cite{Sen06,BH13,BIS+18,AAB+19,BFNV19}, quantum cryptography~\cite{HLSW04,AS04,Kre21,MY22,AQY22}, and sensing~\cite{KRT17,KL17,KZG16,OAG+16}, to quantum communication~\cite{Dev05,DW04,ADHW09,DBWR14,SDTR13,HOW05,HOW07,NHMW17,NWY21,WN23}.
For theoretical physics, quantum randomness provides new insights to thermalization~\cite{PSW06,LPSW09,DRHRW16,KYI20,IH22,BDP23}, the holographic principle~\cite{HP07,NWK23,NMK25} and scrambling~\cite{SS08,LSH+13,MSS16,RY17,MSE+24,PCMCH24} and quantum machine learning~\cite{SJA19,HSCC22}.
For experiments, it has been used to study complex quantum many-body dynamics~\cite{CSM+23,MCS+23} and randomized benchmarking~\cite{EAZ05,KLR+08,MGE11,MGE12,KBC+14,SBM+16,GKL+21,OWE19,HRO+22,HKR22,SCC+24}.

However, a generic $n$-qubit pure state requires exponentially many parameters to describe and, in general, exponential resources to prepare. This motivates efficiently implementable substitutes for Haar randomness, including pseudorandom states~\cite{JLS18}, approximate unitary designs~\cite{MPSY24,SHH25,CSBH25}, and pseudorandom unitaries~\cite{MPSY24,MH25,SHH25}.
Here, we focus on moment-based randomness for state ensembles.
An ensemble is called a \emph{state $t$-design} if its order-$t$ moment operator coincides with that of the Haar measure~\cite{AE07,Mel24}. It therefore reproduces Haar-random statistics through order $t$.
State designs and their approximate variants also arise naturally in physical systems.
For instance, projectively measuring one subsystem of a bipartite pure state and conditioning the complementary subsystem on the measurement outcome produces a \emph{projected state ensemble}~\cite{HC22,CMH+23}.
Projected ensembles have been studied in connection with deep thermalization and emergent state designs~\cite{IH23,Varikuti:2024xeq,MSE+24}, as well as information scrambling~\cite{IH23,MCR26}. They have also been used for random-state preparation and fidelity benchmarking~\cite{CSM+23}, and for randomness conversion and shadow tomography~\cite{MHS+25}.
An ensemble need not be exactly Haar-random to be useful, motivating the problem of quantifying its moment-level randomness.

A widely used diagnostic of randomness in state ensembles is the \emph{state frame potential}, which is defined in terms of pairwise state overlaps.
Given an ensemble $\calE = \cbra{\rbra{\mu_i,\ket{\psi_i}}}$, in which $\ket{\psi_i}$ is sampled with probability $\mu_i$, the state frame potential of order $t$ is
\begin{align}
    \calF_t\rbra{\calE} = \E_{\ket{\psi},\ket{\phi}\sim\mu}\sbra*{\abs*{\braket{\psi}{\phi}}^{2t}},
\end{align} 
where $\ket{\psi}$ and $\ket{\phi}$ are sampled independently according to $\mu$.
The order-$t$ state frame potential is bounded below by the corresponding Haar value, with equality if and only if $\calE$ is a state $t$-design~\cite{NTKD25}.
State-design diagnostics closely related to the frame potential have been used to benchmark analog quantum simulators~\cite{CSM+23,SCC+24} and to study scrambling and deep thermalization~\cite{IH23,Varikuti:2024xeq,MSE+24,MCR26}; the frame potential itself has been used to quantify randomness in projected ensembles~\cite{MHS+25}.

Despite the widespread use of the frame potential, algorithms for its estimation and the associated computational complexity have only recently been studied~\cite{NTKD25}.
In this work, we consider three access models for the estimation task and determine the (near-)optimal cost of estimating $\calF_t\rbra{\calE}$ in each model.
The three access models are: (i) the \emph{query model}, where one is given a multi-state-preparation oracle for the states~\cite{NTKD25} and, for a weighted ensemble, coherent sampling-oracle access to its weights~\cite{Bel19}; (ii) the \emph{general sample model}, where one can obtain multiple copies of $\ket{\psi}$ with $\ket{\psi}\sim\mu$; and (iii) the \emph{single-copy sample model}, where each sample provides a single labeled copy $\ket{\psi}\sim\mu$.

These models correspond to different levels of experimental control.
The single-copy sample model applies when copies of the same state cannot be produced on demand, as is often the case for projected state ensembles obtained by measuring part of an entangled state.
The general sample model applies when each sampled state can be reproduced multiple times, for example, by running a randomized circuit several times with the same random seed.
Coherently implementing such circuit access realizes a multi-state-preparation oracle and hence the query model.
Previous work~\cite{NTKD25} considered only the query model and gave an $O\rbra{t/\eps}$-query algorithm without a matching lower bound.

\section{Results}\label{sec:main_results}

In this work, we give algorithms and corresponding lower bounds under all three access models:
\begin{itemize}
    \item[(i)] \textit{Query model}. We give an algorithm with query complexity $O\rbra{\sqrt{t\log\rbra{1/\eps}}/\eps}$, which improves the dependence on $t$ quadratically over~\cite{NTKD25}, and an almost matching lower bound $\Omega\rbra{\sqrt{t}/\eps}$. 
    \item[(ii)] \textit{General sample model}. We give an algorithm using $O\rbra{t/\eps^2}$ samples and a matching lower bound $\Omega\rbra{t/\eps^2}$, which settles the optimal sample complexity in this model.
    \item[(iii)] \textit{Single-copy sample model}. We give a \emph{store-and-estimate} algorithm whose sample complexity is controlled by the R\'enyi-$\alpha$ entropy $\mathrm{H}_\alpha\rbra{\mu}$ of the ensemble weights. This entropy dependence has a direct operational meaning: ensembles concentrated on a small effective support can be estimated with a number of samples independent of the Hilbert-space dimension. The universal lower bound $\Omega\rbra{t/\eps^2}$ applies in this model as well.
\end{itemize}
\cref{tab:summary} compares our results with prior work.

As an application, we use the single-copy estimation algorithm to assess the randomness of projected state ensembles.
In this setting, the entropy term in the sample-complexity upper bound becomes the observational R\'{e}nyi entropy~\cite{Safranek:2019nwk,Safranek:2019wml,Safranek:2020tgg,Sinha:2023rwr} associated with measuring one subsystem, making the bound directly interpretable in experimentally relevant settings.

\begin{table}[htbp]
    \caption{Resource bounds for estimating the state frame potential $\calF_t\rbra{\calE}$ to additive error $\eps$ with constant success probability. The query and general-sample upper bounds match the corresponding lower bounds up to logarithmic factors. The single-copy upper bound additionally depends on the entropy of the ensemble weights.}
    \label{tab:summary}
    \centering
    \renewcommand{\arraystretch}{1.2}
    \begin{tabular}{l c c c}
    \toprule
    \begin{tabular}{c} Access model \end{tabular}
    & \begin{tabular}{c} Prior work \end{tabular}
    & \begin{tabular}{c} Upper bound \end{tabular}
    & \begin{tabular}{c} Lower bound \end{tabular} \\
    \midrule
    \begin{tabular}{c} Query \end{tabular}
    & \begin{tabular}{c} $O\rbra{t/\eps}$ \\ \cite{NTKD25} \end{tabular}
    & \begin{tabular}{c} $O\rbra{\sqrt{t\log\rbra{1/\eps}}/\eps}$ \\ (\cref{thm:query-main}) \end{tabular}
    & \begin{tabular}{c} $\Omega\rbra{\sqrt{t}/\eps}$ \\ (\cref{thm:query-main}) \end{tabular} \\
    \midrule
    \begin{tabular}{c} General sample \end{tabular}
    & \begin{tabular}{c} --- \end{tabular}
    & \begin{tabular}{c} $O\rbra{t/\eps^2}$ \\ (\cref{thm:sample-main}) \end{tabular}
    & \begin{tabular}{c} $\Omega\rbra{t/\eps^2}$ \\ (\cref{thm:sample-main}) \end{tabular} \\
    \midrule
    \begin{tabular}{c} Single-copy sample \end{tabular}
    & \begin{tabular}{c} --- \end{tabular}
    & \begin{tabular}{c} $\widetilde{O}\rbra*{\frac{t}{\eps^2} + \frac{2^{\mathrm{H}_\alpha\rbra{\mu}}}{\eps^{2-1/\rbra{\alpha-1}}}}$ \\ (\cref{Thm:one_shot_state_frame_potential_informal}) \end{tabular}
    & \begin{tabular}{c} $\Omega\rbra{t/\eps^2}$ \\ (\cref{thm:sample-main}) \end{tabular}\\
    \bottomrule
    \end{tabular}
\end{table}

In the remainder of this section, we present our main results and techniques for estimating the state frame potential in three access models (\cref{sec:query_model_main_result,sec:general_sample_model_main_result,sec:single-copy_model_main_result}). We then apply the single-copy result to projected state ensembles in \cref{sec:projected_ensembles_main_result}.

\subsection{Query model}\label{sec:query_model_main_result}

The query model is natural when the ensemble of interest is specified by a known preparation procedure: the states $\ket{\psi_i}$ may be produced by compiled circuits, digitally simulated Hamiltonian evolutions, or any other procedure whose preparation unitary can be queried coherently. In this regime, the relevant cost is the number of invocations of the state-preparation procedure.

To describe our quantum algorithm for estimating the state frame potential, we assume the following unitary oracle. The multi-state-preparation oracle for $K$ states $\ket{\psi_0}, \ket{\psi_1}, \dots, \ket{\psi_{K-1}}$ is defined by
\begin{align}\label{eq:multi-state_preparation_oracle}
    \calO_S = \sum_{i=0}^{K-1} \ketbra{i}{i} \otimes \calO_i,
\end{align}
where $\mathcal{O}_i\ket{0}=\ket{\psi_i}$ for each $0\leq i<K$.
The multi-state-preparation oracle $\calO_S$ can be viewed as a multiplexer of the individual state-preparation oracles $\mathcal{O}_i$. This type of access is also assumed in~\cite{NTKD25}.
This standard state-preparation assumption is also used in~\cite{BJ98,ATS07,AS05,Bel19,vACGN23,Wan24}.
Our algorithm also works with weaker forms of access, such as purified quantum query access~\cite{GL20}, which is commonly used in property testing of mixed states~\cite{Wat02,BASTS10,GL20,GHS21,SH21,WZC+23,RASW23,GP22,WZ24,Liu25,CWZ25}.
In the pure-state setting, this weaker oracle may prepare $\mathcal{O}_i \ket{0} = \ket{\psi_i} \ket{\phi_i}$, where $\ket{\phi_i}$ is an arbitrary pure state (see also~\cite[Equation (5)]{FW25}).

Following~\cite[Definition 7]{NTKD25}, we say that the ensemble has a computationally efficient description if $K = 2^{\poly\rbra{n}}$ and there exists a polynomial-time classical partial function that, on input $i$, outputs a classical description of a $\poly\rbra{n}$-size circuit implementing $\calO_i$.

To deal with the problem of estimating the weighted state frame potential, we also assume the following sampling oracle for discrete probability distributions.
The sampling oracle for a distribution $\mu$ on the sample space $\cbra{0, 1, 2, \dots, K-1}$ is defined by
\begin{align}\label{eq:sampling_oracle}
    \calO_\mu \ket{0} = \sum_{i=0}^{K-1} \sqrt{\mu_i} \ket{i}.
\end{align}
Nonuniform weights arise naturally for weighted complex-projective designs~\cite{Zhu22}. Here, coherent access to those weights is provided by the sampling oracle $\calO_\mu$.
This type of oracle is commonly used to generate a quantum sample from a probability distribution~\cite{BJ98,ATS07,AS05,Bel19}.
For complexity-theoretic purposes, we assume that $\calO_\mu$ has a $\poly\rbra{n}$-size circuit description when $K = 2^{\poly\rbra{n}}$.
When the ensemble is uniform and $K$ is a power of two, so that $\mu_i = 1/K$ as in~\cite{NTKD25}, $\calO_\mu$ can be implemented simply as $H^{\otimes \log_2\rbra{K}}$. Other cardinalities can be handled by padding the index register.

The query complexity of a quantum algorithm for estimating the state frame potential is measured by the number of queries to \mbox{(controlled-)$\calO_S$}, \mbox{(controlled-)$\calO_\mu$}, and their inverses. 
For convenience, a query to a unitary oracle $U$ means a query to controlled-$U$ or controlled-$U^\dagger$.

Our query algorithm is based on a reformulation: the overlap quantity $\abs{\braket{\psi}{\phi}}^{2t}$ equals the trace of a subnormalized density operator built from $\rho=\ketbra{\psi}{\psi}$ and $\sigma=\ketbra{\phi}{\phi}$, namely $\rbra{\rho\sigma\rho}^{s}\rho\rbra{\rho\sigma\rho}^{s}$ for even $t=2s$ and an analogous expression for odd $t$. Trace estimation then reduces to preparing this subnormalized operator, which we accomplish by composing block-encoding primitives with quantum singular value transformation~\cite{GSLW19}.
The quadratic speedup in $t$ arises from a polynomial-approximation fact that $x^s$ admits a degree-$O\rbra{\sqrt{s\log\rbra{1/\delta}}}$ approximation on $\interval{-1}{1}$~\cite{SV14}.
The resulting query-complexity bounds are stated below. The algorithm and its analysis are deferred to~\cref{sec:query_model}.

\begin{theorem}[Near-optimal quantum algorithm for estimating state frame potential with query access] \label{thm:query-main}
    Let $\calE = \cbra{\rbra{\mu_i, \ket{\psi_i}}}_{i=0}^{K-1}$ be a discrete ensemble on $\calP$ and $t \in \N$. In the query model, the state frame potential $\calF_t\rbra{\calE}$ can be estimated to within additive error $\eps$ using $O\rbra{\sqrt{t \log\rbra{1/\varepsilon}}/\varepsilon}$ queries to $\mathcal{O}_S$ (and $\mathcal{O}_\mu$).
    On the other hand, $\Omega\rbra{\sqrt{t}/\varepsilon}$ queries to $\mathcal{O}_S$ are necessary. 
\end{theorem}

To our knowledge, \cref{thm:query-main} gives the first near-optimal quantum query algorithm for estimating the state frame potential, improving the previous best general bound of $O\rbra{t/\varepsilon}$~\cite{NTKD25}.
The algorithm is also time-efficient. For ensembles of $n$-qubit states, it uses $O\rbra{n\sqrt{t\log\rbra{1/\varepsilon}}/\varepsilon}$ two-qubit gates, improving the previous bound of $O\rbra{nt/\varepsilon}$~\cite{NTKD25}.
It is therefore more efficient than the approach of~\cite{NTKD25} when $t = \Omega\rbra{\log\rbra{1/\varepsilon}}$, with the improvement particularly transparent for constant $\varepsilon$.

From the perspective of computational complexity, \cref{thm:query-main} shows that the quantum query complexity of estimating the state frame potential is $\widetilde{\Theta}\rbra{\sqrt{t}/\varepsilon}$. 
This nearly resolves the query-complexity question raised in~\cite{NTKD25}.
The key step is to reformulate the pure-state problem in terms of mixed-state operators, which enables the use of polynomial approximation and quantum singular value transformation.

As an application of \cref{thm:query-main}, we further obtain an improved quantum algorithm for estimating the unitary frame potential (see \cref{def:unitary_frame_potential}). Here, the unitary oracle takes the form
\begin{align}
    \mathcal{O}_S^U = \sum_{i=0}^{K-1} \ketbra{i}{i} \otimes U_i,
\end{align}
which is obtained by replacing the state-preparation oracles $\calO_i$ in \cref{eq:multi-state_preparation_oracle} with the unitaries $U_i$ from the ensemble $\calE^U=\cbra{\rbra{\mu_i,U_i}}$.

\begin{corollary}[Query algorithm for estimating unitary frame potential] \label{coro:unitary-main}
    Let $\calE^U = \cbra{\rbra{\mu_i, U_i}}_{i=0}^{K-1}$ be a discrete ensemble on $\calU$ and $t \in \N$. In the query model, the unitary frame potential $\calF_t^U\rbra{\calE^U}$ can be estimated to within additive error $\eps$ using $O\rbra{\sqrt{t\log\rbra{1/\varepsilon}} 4^{nt}/\varepsilon}$ queries to $\mathcal{O}_S^U$.
\end{corollary}

\cref{coro:unitary-main} improves the previous bound of $O\rbra{t 4^{nt}/\varepsilon}$~\cite{NTKD25} when $t = \Omega\rbra{\log\rbra{1/\varepsilon}}$.

\subsection{General sample model}\label{sec:general_sample_model_main_result}

The general sample model assumes that one can obtain multiple copies of a state sampled from the ensemble. Our algorithm builds on the generalized SWAP test~\cite{EAO+02,KLL+17}. Applying this test to two independent $t$-fold samples gives an unbiased estimator of the state frame potential.

The optimal sample-complexity bounds are stated below. The algorithm and proofs are given in~\cref{sec:general_sample_model}.

\begin{theorem}[Optimal quantum algorithm for estimating state frame potential with general sample access] \label{thm:sample-main}
    Let $\calE = \cbra{\rbra{\mu_i, \ket{\psi_i}}}_{i=0}^{K-1}$ be a discrete ensemble on $\calP$ and $t \in \N$. In the general sample model, the state frame potential $\calF_t\rbra{\calE}$ can be estimated to within additive error $\eps$, with probability at least $3/4$, using $O\rbra{t/\varepsilon^2}$ samples.
    On the other hand, $\Omega\rbra{t/\varepsilon^2}$ samples are necessary. 
\end{theorem}

\cref{thm:sample-main} gives an \textit{optimal} quantum algorithm for estimating the state frame potential under general sample access.
For ensembles of $n$-qubit states, the algorithm uses $O\rbra{nt/\varepsilon^2}$ two-qubit quantum gates.
Although its gate complexity does not improve on the query-model approaches in \cref{thm:query-main} and~\cite{NTKD25}, the general sample model requires less coherent control: it needs only identical copies of states drawn from the ensemble.

\subsection{Single-copy sample model}\label{sec:single-copy_model_main_result}
In the single-copy sample model, each sample contains only one labeled copy of a state drawn from the ensemble. We give an explicit estimation algorithm under the assumption that the R\'{e}nyi entropy of the weight distribution is known in advance. The full algorithm appears in~\cref{S:single_copy_state_frame_potential}. Here we summarize its main idea.

Our estimation algorithm consists of two stages. In the first stage, we repeatedly sample from the ensemble to identify the \emph{heavy set}, namely the set of elements whose probabilities exceed a prescribed threshold, and store sufficiently many copies of the corresponding states. In the second stage, we apply the generalized SWAP test to these stored states. The contribution from the low-probability elements is controlled by the choice of the threshold, while the statistical error is controlled by the number of stored copies. This allows us to estimate the frame potential to any desired accuracy.

The resulting sample-complexity bound is stated below and proved in~\cref{S:single_copy_state_frame_potential}.

\begin{theorem}[Quantum algorithm for estimating state frame potential with single-copy sample access]\label{Thm:one_shot_state_frame_potential_informal}
    Let $\calE = \cbra{\rbra{\mu_i, \ket{\psi_i}}}_{i=0}^{K-1}$ be a discrete ensemble on $\calP$, $t \in \N$, and $\alpha >1$. In the single-copy sample model, the state frame potential $\calF_t(\calE)$ can be estimated to within additive error $\eps$ with sample complexity
    \[
        \widetilde{O}\rbra*{\frac{t}{\eps^2}+\frac{2^{\mathrm{H}_{\alpha}(\mu)}}{\eps^{2-\frac{1}{\alpha-1}}}},
    \]
    where $\mathrm{H}_\alpha\rbra{\mu} \coloneqq \rbra{1-\alpha}^{-1}\log \sum_{i=0}^{K-1} \mu_i^{\alpha}$ is the R\'{e}nyi-$\alpha$ entropy of the probability distribution $\mu$.
\end{theorem}

The entropy-dependent term in~\cref{Thm:one_shot_state_frame_potential_informal} reflects the effective support size of the ensemble weights. In particular, for sufficiently concentrated distributions, the state frame potential can be estimated with a number of samples independent of the Hilbert-space dimension.

In general, however, the R\'{e}nyi entropy of the underlying distribution $\mu$ is not known in advance.
Removing this assumption would require estimating or upper-bounding the relevant R\'{e}nyi entropy from samples. Establishing a distribution-free guarantee for such a preliminary stage is left for future work.

\subsection{Projected state ensembles}\label{sec:projected_ensembles_main_result}

As an application of the estimation algorithm in the single-copy sample model, we consider the problem of estimating the state frame potential of a projected state ensemble generated by measuring subsystem $\mathsf{B}$ of a given pure state $\ket{\Psi}_{\mathsf{AB}}$~\cite{HC22,CMH+23}.

\begin{definition}[Projected state ensemble]
    Let $\ket{\Psi}_{\mathsf{AB}}$ be a pure state on $\mathsf{AB}$, and let $E=\cbra*{\ket{e_i}}_{i=0}^{d_{\mathsf{B}}-1}$ be an orthonormal basis of $\mathsf{B}$, where $d_{\mathsf{B}}$ is the dimension of $\mathsf{B}$.
    The projected state ensemble of $\ket{\Psi}_{\mathsf{AB}}$ with respect to the basis $E$ is the ensemble $\calE_{\Psi,E} \coloneqq \cbra{\rbra{\mu_i, \ket{\psi_i}_{\mathsf{A}}}}_{i=0}^{d_{\mathsf{B}}-1}$, where
    \begin{align}
        &\mu_i 
        = \tr\rbra*{\rbra*{ I_{\mathsf{A}} \otimes \ketbra{e_i}{e_i}_{\mathsf{B}}}\ketbra{\Psi}{\Psi}_{\mathsf{AB}}}, \\
        &\ket{\psi_i}_{\mathsf{A}} 
        = \frac{\rbra*{I_{\mathsf{A}} \otimes \bra{e_i}_{\mathsf{B}}}\ket{\Psi}_{\mathsf{AB}}}{\sqrt{\mu_i}} .
    \end{align}
    Here $I_{\mathsf{A}}$ denotes the identity operator on $\mathsf{A}$. Outcomes with $\mu_i=0$ are omitted from the ensemble.
\end{definition}

It is important to condition on, and retain, the measurement outcome $i$, which keeps the state in $\mathsf{A}$ pure.
If the outcome is ignored, the state in $\mathsf{A}$ is described by the averaged state $\Psi_{\mathsf{A}} = \tr_{\mathsf{B}}\rbra{\ketbra{\Psi}{\Psi}_{\mathsf{AB}}}$, and the ensemble structure is lost.

Projected state ensembles have recently attracted considerable attention in quantum many-body physics, especially as probes of scrambling dynamics and deep thermalization~\cite{IH23,Varikuti:2024xeq,MSE+24,MCR26}, where the pure state $\ket{\Psi}_{\mathsf{AB}}$ is generated by the Hamiltonian dynamics of a complex many-body system.
Projected state ensembles are also relevant to quantum information: related protocols have been used to benchmark analog quantum simulators~\cite{CSM+23,SCC+24} and to improve shadow tomography~\cite{MHS+25}.

A basic challenge is to determine whether the projected state ensemble generated from a given state $\ket{\Psi}_{\mathsf{AB}}$ is sufficiently random.
Except for a few analytically tractable classes of states~\cite{HC22,CL22,IH23}, this question is generally difficult to address analytically.
Here, we apply our single-copy estimation algorithm to quantify the randomness of projected state ensembles through the state frame potential.

Since a projected state ensemble $\calE_{\Psi, E}$ depends on both the state $\ket{\Psi}_{\mathsf{AB}}$ and the measurement $E$, its sample cost should also depend on $E$.
This dependence is naturally captured by \emph{observational entropy}, which describes the entropy of a quantum system relative to a specified observation or coarse-graining~\cite{Safranek:2019nwk,Safranek:2019wml,Safranek:2020tgg,Sinha:2023rwr}.
Indeed, the probabilities $\cbra{\mu_i}_i$ are exactly the outcome probabilities obtained by measuring
\begin{align}
    \Psi_{\mathsf{B}} = \tr_{\mathsf{A}}\rbra*{\ketbra{\Psi}{\Psi}_{\mathsf{AB}}}
\end{align}
in the basis $E$. Hence, the R\'{e}nyi entropy appearing in the sample complexity in~\cref{Thm:one_shot_state_frame_potential_informal} is the observational R\'{e}nyi entropy of $\Psi_{\mathsf{B}}$ with respect to $E$.

Specializing~\cref{Thm:one_shot_state_frame_potential_informal} to the collision entropy ($\alpha = 2$) gives the following sample upper bound.

\begin{corollary}[Quantum algorithm for estimating state frame potential of projected state ensemble]
    Let $\ket{\Psi}_{\mathsf{AB}}$ be a pure state on a bipartite system $\mathsf{AB}$, and let $E=\cbra{\ket{e_i}}_{i=0}^{d_{\mathsf{B}}-1}$ be an orthonormal basis of $\mathsf{B}$. Let $\calE_{\Psi,E}$ be the projected state ensemble induced by $\ket{\Psi}_{\mathsf{AB}}$ and $E$. Then, for any $t \in \N$ and $\varepsilon \in \interval[open left]{0}{1}$, there exists an explicit algorithm that estimates $\calF_t\rbra{\calE_{\Psi,E}}$ to within additive error $\varepsilon$ with constant success probability, using
    \[
        \widetilde{O}\rbra*{
        \frac{t}{\varepsilon^2}
        + \frac{2^{\mathrm{H}_{E}^{(2)}(\mathsf{B})_{\Psi}}}{\varepsilon}
        }
    \]
    copies of $\ket{\Psi}_{\mathsf{AB}}$. Here, $
        \mathrm{H}_{E}^{(2)}(\mathsf{B})_{\Psi} \coloneqq -\log \sum_{i=0}^{d_{\mathsf{B}}-1}\mu_i^2
    $
    is the observational collision entropy of $\Psi_{\mathsf{B}}=\tr_{\mathsf{A}}\rbra{\ketbra{\Psi}{\Psi}_{\mathsf{AB}}}$ with respect to the measurement basis $E$, where $\mu_i=\bra{e_i}\Psi_{\mathsf B }\ket{e_i}$.
\end{corollary}

This result shows that the state frame potential of a projected state ensemble can be assessed with a sample cost controlled by the observational collision entropy of the measurement used to generate the ensemble. Projected state ensembles with low observational collision entropy provide a natural regime in which randomness, quantified by the frame potential, can be efficiently estimated.

\section{Discussion}\label{sec:discussion}

This work provides efficient quantum algorithms for estimating the state frame potential in three access models, ordered by decreasing experimental control. In the query and general sample models, our upper and lower bounds match up to logarithmic factors. In the single-copy model, the upper bound additionally depends on the entropy of the ensemble weights.
By settling the \mbox{(near-)optimal} scaling of this fundamental estimation problem in the query and general sample models, our algorithms make higher-order randomness benchmarks tractable with fewer oracle calls or samples than prior approaches such as~\cite{NTKD25}.
In particular, the single-copy algorithm is suited to projected-ensemble experiments in analog quantum simulators, where post-selected measurement records furnish only one copy per outcome and the observational-entropy bound directly controls the experimental cost.

Our results suggest several directions for future work.
A first natural question is whether the remaining logarithmic factor in the query upper bound can be removed. For the single-copy model, an important open problem is to determine whether the entropy-dependent term is necessary.

Although this work also provides an algorithm for estimating the unitary frame potential, this is not our main focus.
Given the comparable importance of the unitary and state frame potentials, it would be natural to develop sample algorithms for estimating the unitary frame potential and to establish matching lower bounds.

The two sample models considered in this work represent two extreme cases: when estimating the order-$t$ frame potential, one is allowed either to prepare $t$-fold copies of a state or only a single copy at a time.
Our upper bounds indicate that the leading dependence on $t$ does not change substantially between these two regimes, that is, whether or not one can prepare several identical copies in succession.
This is operationally meaningful because, at least for the leading $t$-dependent term, preparing multiple identical copies does not improve the sample complexity.

Finally, in the single-copy sample model, our sample complexity is bounded in terms of the R\'{e}nyi entropy of the underlying distribution of the ensemble, which is unknown in advance.
As we have already mentioned, we may use the heavy part of $\mu$, obtained in our algorithm, to estimate the R\'{e}nyi entropy.
Analyzing this strategy and determining its additional computational cost are left for future work.

\section{Preliminaries}\label{sec:prelim}

We introduce our notation in~\cref{sec:notations}. We then briefly review state and unitary designs and formally define their frame potentials in~\cref{sec:design_theory}. Finally, we review the quantum algorithmic tools used in our algorithms in~\cref{sec:algorithmic_tools}.

\subsection{Notations}\label{sec:notations}

We denote the sets of natural numbers, integers, real numbers, and complex numbers by $\N$, $\Z$, $\R$, and $\C$, respectively. We use $\C^{2^n \times 2^n}$ to denote the set of $2^n$-by-$2^n$ complex matrices.
For a vector $\ket{\psi}$, its Euclidean norm is $\Abs{\ket{\psi}}=\Abs{\bra{\psi}}=\sqrt{\braket{\psi}{\psi}}$.
For a linear operator $A$, we denote by $\abs{A}$ the unique positive square root of the positive semidefinite operator $A^\dagger A$, where $A^\dagger$ is the Hermitian conjugate of $A$.
For $p \in \interval[open right]{1}{\infty}$, the Schatten $p$-norm of a linear operator $A$ is $\Abs{A}_p = \rbra{\tr\rbra{\abs{A}^p}}^{1/p}$.
As $p$ tends to $\infty$, this becomes the operator norm $\Abs{A}_{\infty}$, or simply $\Abs{A}$.
A function $f \colon \R \to \R$ can be extended to a matrix function of a Hermitian operator $A$ with spectral decomposition $A = V \Sigma V^{\dagger}$ by setting $f\rbra{A}=V f\rbra{\Sigma} V^{\dagger}$, where $\Sigma = \diag\rbra{\lambda_1,\ldots,\lambda_n}$ and $f\rbra{\Sigma} = \diag\rbra{f\rbra{\lambda_1},\ldots,f\rbra{\lambda_n}}$.
We denote by $\R\sbra{x}$ the set of polynomials with real coefficients.
Sans-serif font is used to denote registers.
We use $\widetilde{O}(\cdot)$ to suppress polylogarithmic factors in its argument. That is, $\widetilde{O}\rbra{f\rbra{n}} = O\rbra{ f\rbra{n}\polylog\rbra{ f\rbra{n} } }$.

\subsection{State and unitary design}\label{sec:design_theory}

Assume that $K = 2^k$ is a power of $2$.
Denote by $\calP$ and $\calU$ the sets of pure states and unitary operators, respectively.
Denote by $\calE=\cbra{\rbra{\mu_i,\ket{\psi_i}}}_{i=0}^{K-1}$ the ensemble of pure states associated with a probability distribution $\mu=\rbra{\mu_0,\ldots,\mu_{K-1}}$, where $\ket{\psi_i}$ is sampled with probability $\mu_i$.
Similarly, we denote by $\calE^U=\cbra{\rbra{\mu_i,U_i}}_{i=0}^{K-1}$ the ensemble of unitaries associated with a probability distribution $\mu=\rbra{\mu_0,\ldots,\mu_{K-1}}$, where $U_i$ is sampled with probability $\mu_i$.
We use $\ket{\psi} \sim \mu$ or $U \sim \mu$ to denote that the state $\ket{\psi}$ or the unitary $U$ is sampled from distribution $\mu$.
These collections may contain repeated elements. Sets $\cbra{\ket{\psi_i}}_{i=0}^{K-1}$ and $\cbra{U_i}_{i=0}^{K-1}$ may be multisets.
When there is no ambiguity, we suppress the index range and write $\cbra{\rbra{\mu_i,\ket{\psi_i}}}$.

A state $t$-design is defined as follows~\cite{Mel24,NTKD25}.

\begin{definition}[State $t$-design]
    Fix $t \in \N$. An ensemble $\calE=\cbra{\rbra{\mu_i,\ket{\psi_i}}}$ is called a state $t$-design if
    \begin{align}
        \Ex_{\ket{\psi} \sim \mu}\sbra*{\ketbra{\psi}{\psi}^{\otimes t}} = \Ex_{\ket{\phi} \sim \mu_H}\sbra*{\ketbra{\phi}{\phi}^{\otimes t}} ,
    \end{align}
    where $\mu_H$ is the uniform measure over $\calP$.
\end{definition}

For pure-state ensembles, allowing nonuniform probabilities gives a weighted state $t$-design. Equivalently, this is a weighted complex-projective $t$-design~\cite[Definition 2.1]{RS07}. Uniform multisets form the special case in which multiplicities encode the weights. We use the term ``state $t$-design'' for both weighted and unweighted ensembles.

The extent to which an ensemble approximates a state design is quantified by its state frame potential, defined as follows~\cite[Definition 2]{NTKD25}.

\begin{definition}[State frame potential]
    Fix $t \in \N$. For an ensemble $\calE=\cbra{\rbra{\mu_i,\ket{\psi_i}}}$, its state frame potential of order $t$ is defined as
    \begin{align}
        \calF_t\rbra*{\calE} = \Abs*{\Ex_{\ket{\psi} \sim  \mu}\sbra*{\ketbra{\psi}{\psi}^{\otimes t}}}_2^2  =  \Ex_{\ket{\psi},\ket{\phi} \sim \mu}\sbra*{\abs*{\braket{\psi}{\phi}}^{2t}} .
    \end{align}
\end{definition}

When the ensemble is discrete and has cardinality $K$, i.e., $\calE = \cbra{\rbra{\mu_i, \ket{\psi_i}}}_{i = 0}^{K-1}$, the state frame potential reduces to
\begin{equation}
    \calF_t(\calE) = \sum_{i=0}^{K-1}\sum_{j=0}^{K-1} \mu_i \mu_j \left| \braket{\psi_i}{\psi_j} \right|^{2t}.
\end{equation}

An analogous notion is that of a unitary $t$-design. As in the state-design setting, we do not distinguish between weighted and unweighted unitary designs~\cite{RS09}.

\begin{definition}[Unitary $t$-design]
    Fix $t \in \N$. An ensemble $\calE^U=\cbra{\rbra{\mu_i,U_i}}$ is called a unitary $t$-design if
    \begin{align}
        \Ex_{U \sim \mu}\sbra{U^{\otimes t} \otimes U^{\ast \otimes t}} 
        = \Ex_{U \sim \mu_H}\sbra{U^{\otimes t} \otimes U^{\ast \otimes t}}
    \end{align}
    where $U^\ast$ denotes the complex conjugate of $U$, and $\mu_H$ is the Haar measure.
\end{definition}

Similarly, the extent to which an ensemble approximates a unitary design is quantified by its unitary frame potential, defined as follows~\cite[Definition 4]{NTKD25}.

\begin{definition}[Unitary frame potential]\label{def:unitary_frame_potential}
    Fix $t \in \N$. For an ensemble $\calE^U=\cbra{\rbra{\mu_i,U_i}}$, its unitary frame potential of order $t$ is defined as
    \begin{align}
        \calF_t^U\rbra*{\calE^U} = \Ex_{U,V \sim \mu} \sbra*{\abs*{\tr\rbra*{U^\dagger V}}^{2t}}.
    \end{align}

\end{definition}

When the ensemble is discrete and has cardinality $K$, i.e., $\calE^U = \cbra{\rbra{\mu_i, U_i}}_{i=0}^{K-1}$, the unitary frame potential reduces to
\begin{align}
    \calF_t^U\rbra*{\calE^U} = \sum_{i=0}^{K-1}\sum_{j=0}^{K-1} \mu_i \mu_j \abs*{\tr\rbra*{U_i^\dagger U_j}}^{2t}.
\end{align}

\subsection{Quantum algorithmic tools}\label{sec:algorithmic_tools}

We review two basic algorithmic tools used later in this work: block-encoding in \cref{sec:block-encoding} and quantum singular value transformation in \cref{sec:qsvt}.

\subsubsection{Block-encoding}\label{sec:block-encoding}
A block-encoding embeds a non-unitary linear operator $A$ as the upper-left corner of a larger unitary $B$, so that quantum-circuit primitives acting on unitaries can be lifted to act on $A$.
This is the technical bridge used throughout our query algorithm: it lets us promote a density operator $\rho$ (or a product such as $\rho \sigma \rho$) into a unitary operation on a slightly enlarged register, after which polynomial functions of $\rho$ are accessible through quantum singular value transformation and traces become estimable through amplitude estimation.

\begin{definition}[Block-encoding \cite{GSLW19}]
    Suppose that $A$ is an $n$-qubit linear operator.
    For a real number $\alpha>0$, an error parameter $\eps\geq0$, and a nonnegative integer $a$, an $\rbra{n+a}$-qubit unitary operator $B$ is said to be an $\rbra{\alpha,a,\eps}$-block-encoding of $A$ if
    \begin{align}
        \Abs*{\alpha\bra{0}^{\otimes a}B\ket{0}^{\otimes a}-A} \leq \eps.
    \end{align}
\end{definition}

We next formalize what it means for a unitary to prepare a subnormalized density operator.

\begin{definition}[Subnormalized density operator]
    A positive semidefinite operator $A$ is called a subnormalized density operator if $0 \leq \tr\rbra{A} \leq 1$.
    Moreover, a density operator is a subnormalized density operator with unit trace.
    An $\rbra{n+a+b}$-qubit unitary $U$ is said to prepare the $n$-qubit subnormalized density operator $A$ if, for $\ket{\sigma}=U\ket{0}_{n+a+b}$, the reduced $\rbra{n+a}$-qubit density operator $\tr_b\rbra{\ketbra{\sigma}{\sigma}}$ is a $\rbra{1,a,0}$-block-encoding of $A$.
    Moreover, $a=0$ when $A$ is a normalized density operator.
\end{definition}

Based on this definition, we can generalize our state-preparation oracle to density operators.

\begin{definition}[Mixed state-preparation oracle]
    The mixed state-preparation oracle for $K$ density operators $\sigma_0, \sigma_1, \ldots, \sigma_{K-1}$ is defined by
    \begin{align}
        \calO_S = \sum_{i=0}^{K-1} \ketbra{i}{i} \otimes \calO_{\sigma_i},
    \end{align}
    where each $\calO_{\sigma_k}$ prepares the density operator $\sigma_k$.
    The (pure) state-preparation oracle is a special case of the mixed state-preparation oracle.
\end{definition}

We use the following result to convert purified access to a density operator into a block-encoding.
\begin{lemma}[Block-encoding of density operator,~{\cite[Lemma 25]{GSLW19}}]\label{lem:BE_density_operator}
    Suppose $\rho$ is an $n$-qubit density operator with purified access $\calO_\rho$ which is an $\rbra{n+a}$-qubit unitary that prepares a purification of $\rho$.
    One can construct a unitary $U_\rho$ that is an $\rbra{1,n+a,0}$-block-encoding of $\rho$, using one invocation each of $\calO_\rho$ and $\calO_\rho^\dagger$.
\end{lemma}

The following lemma constructs a block-encoding of the product of two matrices from block-encodings of the factors.
\begin{lemma}[Product of block-encoded matrices,~{\cite[Lemma 53 in the full version]{GSLW19}}]\label{lem:product_BE}
    Suppose $U$ is an $\rbra{\alpha,a,\eps_a}$-block-encoding of an $n$-qubit operator $A$ and $V$ is an $\rbra{\beta,b,\eps_b}$-block-encoding of an $n$-qubit operator $B$.
    Then $\mathsf{BEProduct}\rbra{U,V} \coloneqq \Tilde{U}=\rbra{U \otimes I_b}\rbra{V \otimes I_a}$ is an $\rbra{\alpha\beta,a+b,\alpha\eps_b+\beta\eps_a}$-block-encoding of $AB$.
\end{lemma}
The following lemma constructs a unitary describing the evolution of a subnormalized density operator.

\begin{lemma}[{\cite[Lemma II.2]{WGL+24}}]\label{lem:sandwiched_product}
    Suppose $U$ is an $\rbra{n+a}$-qubit unitary operator that prepares an $n$-qubit subnormalized density operator $A$ and $V$ is an $\rbra{n+b}$-qubit unitary operator which is a $\rbra{1,b,0}$-block-encoding of $B$.
    Then $\Tilde{U}=\mathsf{SubDensityOpEvolution}\rbra{V,U}=\rbra{V \otimes I_a}\rbra{U \otimes I_b}$ is an $\rbra{n+a+b}$-qubit unitary operator that prepares the $n$-qubit subnormalized density operator $B A B^\dagger$.
\end{lemma}

\subsubsection{Quantum singular value transformation}\label{sec:qsvt}
Given access to a block-encoding of a Hermitian operator $A=\sum_{i \in \sbra{n}}\lambda_i\ketbra{v_i}{v_i}$, QSVT can implement a block-encoding of $p\rbra{A}$ for suitable polynomials $p \in \R\sbra{x}$.

\begin{theorem}[Quantum singular value transformation {\cite[Theorem 31]{GSLW19}}]\label{thm:qsvt}
    Suppose that $A$ is Hermitian and that an $\rbra{\alpha,a,\eps}$-block-encoding $U$ of $A$ is given.
    Let $\delta\in\interval[open]{0}{1}$, and let $p\in\R\sbra{x}$ be a polynomial of degree $d$ such that $\abs{p\rbra{x}}\leq1/2$ for $x\in\sbra{-1,1}$.
    Then there is a quantum unitary operator $\Tilde{U}=\mathsf{EigenTrans}\rbra{U,p,\delta}$ that is an $\rbra{1,a+2,4d\sqrt{\eps/\alpha}+\delta}$-block-encoding of $p\rbra{A/\alpha}$.
    The circuit uses
    \begin{itemize}
        \item $d$ applications of $U$ or $U^\dagger$ and one application of controlled-$U$, and
        \item $O\rbra{\rbra{a+1}d}$ additional one- and two-qubit gates.
    \end{itemize}
    Moreover, the classical description of $\Tilde{U}$ can be computed in time $\poly\rbra{d,\log\rbra{1/\delta}}$.
\end{theorem}

We also need the following results on polynomial approximations to monomials and on quantum amplitude estimation.

\begin{theorem}[Polynomial approximation for monomials {\cite[Theorem 3.3]{SV14}}]\label{thm:poly_approx_monomials}
    For any integer $q \geq 1$ and real number $\eps \in \interval[open]{0}{1}$, there is an efficiently computable polynomial $p \in \R\sbra{x}$ of degree $O\rbra{\sqrt{q\log\rbra{1/\eps}}}$ and parity $q \bmod 2 $ such that $\abs{p\rbra{x}-x^q} \leq \eps$ and $\abs{p\rbra{x}} \leq 1$ for all $x \in \interval[open]{-1}{1}$.
\end{theorem}

\begin{theorem}[Quantum amplitude estimation~{\cite[Theorem 12]{BHMT02}}]\label{thm:amp_est}
    Suppose $U$ is a unitary operator such that
    \begin{align}
        U\ket{0} = \sqrt{p}\ket{0}\ket{\phi_0}+\sqrt{1-p}\ket{1}\ket{\phi_1},
    \end{align}
    where $\ket{\phi_0}$ and $\ket{\phi_1}$ are normalized pure states.
    There is a quantum algorithm $\mathsf{AmpEst}\rbra{U,\eps,\delta}$ that outputs an estimate of $p$ to within additive error $\eps$ with success probability at least $1-\delta$ using $O\rbra{\frac{1}{\eps}\log\rbra{\frac{1}{\delta}}}$ queries to $U$.
\end{theorem}

\section{Query model}\label{sec:query_model}

In this section, we establish nearly optimal upper and lower bounds on the query complexity of estimating the state frame potential of a given state ensemble. In the query model, the algorithm has access to a multi-state-preparation oracle (see \cref{eq:multi-state_preparation_oracle}) and a sampling oracle (see \cref{eq:sampling_oracle}). Query complexity is measured by the total number of calls to $\calO_S$ and $\calO_\mu$, including calls to their controlled versions and inverses.
The upper and lower bounds are established in \cref{sec:query_upper} and \cref{sec:query_lower}, respectively.

\subsection{Upper bound}\label{sec:query_upper}

The frame potential is built from $2t$-th moments of pairwise overlaps, so a naive approach would estimate each overlap moment by an order-$t$ product of state-preparation calls, yielding the same linear dependence on $t$ as the algorithm of~\cite{NTKD25}.
Our key observation is that the same quantity admits a representation as the trace of a low-rank polynomial in two density operators, allowing the linear scaling to be replaced by the degree of a polynomial approximation of $x^s$, which is only $O\rbra{\sqrt{s}}$ up to logarithmic factors~\cite{SV14}.

We present the resulting algorithm, which admits the query complexity $O\rbra{\sqrt{t\log\rbra{1/\eps}}/\eps}$, a quadratic speedup in $t$ over the previous quantum algorithm in~\cite{NTKD25}.
The construction is described in terms of the target unitary $U_{i}$ of the state-preparation oracle.

We first show how to construct the crucial unitaries in \cref{sec:construct_unitary}.
The algorithm and the analysis of its query complexity are described in \cref{sec:algorithm}. We then connect the algorithm to the estimation of the state frame potential.

\subsubsection{Subnormalized-density-operator construction}\label{sec:construct_unitary}

After the trace reformulation introduced above, estimating $\calF_t\rbra{\calE}$ reduces to estimating the traces of the \emph{pivot matrices} defined below. These are subnormalized density operators built from $\sigma_i$ and $\sigma_j$.
The goal of this section is to construct a quantum circuit that prepares each pivot matrix as the subnormalized state of a designated register.
Its trace can then be read out by amplitude estimation in the next step.
The construction has two parity-dependent variants because the polynomial $\abs{\braket{\psi}{\phi}}^{2t}$ has the structurally distinct trace representations stated in~\cref{eq:single_matrix_even,eq:single_matrix_odd} below.

We show the common building block.
We then present separate constructions for even and odd $t$.
For fixed $s \in \N$ and each pair of density operators $\sigma_i,\sigma_j$, where $i,j \in \cbra{0,\ldots,K-1}$, define the pivot matrices $M_{i,j,s}^{\textup{even}}$ and $M_{i,j,s}^{\textup{odd}}$ by
\begin{align}\label{eq:single_matrix_even}
    M_{i,j,s}^{\textup{even}} = \rbra*{\sigma_i\sigma_j\sigma_i}^s \sigma_i \rbra*{\sigma_i\sigma_j\sigma_i}^s
\end{align}
and
\begin{align}\label{eq:single_matrix_odd}
    M_{i,j,s}^{\textup{odd}} = \rbra*{\sigma_i\sigma_j\sigma_i}^s \sigma_i \sigma_j \sigma_i \rbra*{\sigma_i\sigma_j\sigma_i}^s.
\end{align}
For an ensemble of density operators $\calE =\cbra{\rbra{\mu_i,\sigma_i}}_{i=0}^{K-1}$, define the corresponding mixed pivot matrices by
\begin{align}\label{eq:dist_matrix_even}
    M_{\calE,s}^{\textup{even}} = \sum_{i}\sum_j\mu_i \mu_j M_{i,j,s}^{\textup{even}}
\end{align}
and 
\begin{align}\label{eq:dist_matrix_odd}
    M_{\calE,s}^{\textup{odd}} = \sum_{i}\sum_j\mu_i \mu_j M_{i,j,s}^{\textup{odd}}.
\end{align}

The following lemma describes the evolution of the subnormalized density operator.
\begin{lemma}\label{lem:trace_est_subnormalized}
    Given an $\rbra{n+a}$-qubit unitary $U$ that prepares an $n$-qubit subnormalized density operator $A$, we have
    \begin{align}
        U_{\mathsf{NA_1A_2}}\ket{0}_{\mathsf{NA_1A_2}} = \sqrt{\tr\rbra*{A}}\ket{0}_{\mathsf{A_1}}\ket{\phi_0}_{\mathsf{NA_2}} 
        + \sqrt{1-\tr\rbra*{A}} \ket{0^\perp_{\mathsf{A_1}}}_{\mathsf{NA_1A_2}},
    \end{align}
    where $\mathsf{N}$, $\mathsf{A_1}$, and $\mathsf{A_2}$ are $n$-, $a_1$-, and $a_2$-qubit registers, respectively, with $a=a_1+a_2$, and $\Abs{\bra{0}_{\mathsf{A_1}}\ket{0^\perp_{\mathsf{A_1}}}_{\mathsf{NA_1A_2}}}=0$.
    Moreover, let $\mathsf{T}$ be a single-qubit system, and 
    \begin{align}
        V_{\mathsf{A_1T}} = \left(X_{\mathsf{A_1}} \otimes X_{\mathsf{T}}\right)\textup{Ctrl}_{\mathsf{A_1}}\textup{-}X_{\mathsf{T}}
        \left(X_{\mathsf{A_1}} \otimes I_{\mathsf{T}}\right)
    \end{align}
    where $X_{\mathsf{A}_1} \coloneqq X^{\otimes a_1}$, and $\textup{Ctrl}_{\mathsf{A_1}}\textup{-}X_{\mathsf{T}}$ is the controlled unitary on $\mathsf{A_1T}$ that applies $X_{\mathsf{T}}$ only when $\mathsf{A_1}$ is in the state $\ket{1\dots1}_{\mathsf{A_1}}$. It then follows that
    \begin{align}
        V_{\mathsf{A_1T}}\rbra*{\ket{0}_{\mathsf{T}}\otimes U_{\mathsf{NA_1A_2}}\ket{0}_{\mathsf{NA_1A_2}}}
        = \sqrt{\tr\rbra*{A}}\ket{0}_{\mathsf{T}}\ket{0}_{\mathsf{A_1}}\ket{\phi_0}_{\mathsf{NA_2}} + \sqrt{1-\tr\rbra*{A}} \ket{1}_{\mathsf{T}}\ket{0^\perp_{\mathsf{A_1}}}_{\mathsf{NA_1A_2}}.
\end{align}
\end{lemma}

\begin{proof}
    Without loss of generality, assume that $U$ prepares an $\rbra{n+a_1}$-qubit density operator $\sigma$ and $U\ket{0}_{\mathsf{NA_1A_2}} = \ket{\sigma}_{\mathsf{NA_1A_2}}$, where $a=a_1+a_2$ and $\ket{\sigma}$ is a purification of $\sigma$.
    Using the identity $I_{\mathsf{A_1}} = \ketbra{0}{0}_{\mathsf{A_1}} + \Pi^{\neq 0}_{\mathsf{A_1}}$, where $\Pi^{\neq 0}_{\mathsf{A_1}} = \sum_{j=1}^{2^{a_1}-1}\ketbra{j}{j}_{\mathsf{A_1}}$, we have
    \begin{align}
        U\ket{0}_{\mathsf{NA_1A_2}}
        = \ket{0}_{\mathsf{A_1}} \otimes \bra{0}_{\mathsf{A_1}}\ket{\sigma}_{\mathsf{NA_1A_2}} + \Pi^{\neq 0}_{\mathsf{A_1}}\ket{\sigma}_{\mathsf{NA_1A_2}}.
    \end{align}
    Clearly, the former term is orthogonal to the latter term. To normalize each term, we compute $\Abs{ \bra{0}_{\mathsf{A_1}}\ket{\sigma}_{\mathsf{NA_1A_2}}}$, which is
    \begin{align}
        \Abs*{ \bra{0}_{\mathsf{A_1}}\ket{\sigma}_{\mathsf{NA_1A_2}} }
        & = \sqrt{\tr\rbra*{\bra{0}_{\mathsf{A_1}}\ketbra{\sigma}{\sigma}_{\mathsf{NA_1A_2}}\ket{0}_{\mathsf{A_1}} }} \\
        & = \sqrt{\tr\rbra*{\bra{0}_{\mathsf{A_1}}\sigma_{\mathsf{NA_1}}\ket{0}_{\mathsf{A_1}} }} \\
        & = \sqrt{\tr\rbra*{ A } },
    \end{align}
    where the last equality follows because $U$ prepares a block-encoding of $A$.
    
    By defining normalized state vectors as
    \begin{align}
        \ket{\phi_0}_{\mathsf{NA_2}} = \frac{\bra{0}_{\mathsf{A_1}}\ket{\sigma}_{\mathsf{NA_1A_2}}}{\Abs*{\bra{0}_{\mathsf{A_1}}\ket{\sigma}_{\mathsf{NA_1A_2}} }}
        = \frac{\bra{0}_{\mathsf{A_1}}\ket{\sigma}_{\mathsf{NA_1A_2}}}{\sqrt{\tr\rbra*{A}}}
    \end{align}
    and
    \begin{align}
        \ket{0^\perp_{\mathsf{A_1}}}_{\mathsf{NA_1A_2}} = \frac{\Pi^{\neq 0}_{\mathsf{A_1}}\ket{\sigma}_{\mathsf{NA_1A_2}}}{\sqrt{1-\tr\rbra*{A}}},
    \end{align}
    we obtain
    \begin{align}
        U\ket{0}_{\mathsf{NA_1A_2}}
        = \sqrt{\tr\rbra*{A}}\ket{0}_{\mathsf{A_1}}\ket{\phi_0}_{\mathsf{NA_2}}
        + \sqrt{1-\tr\rbra*{A}} \ket{0^\perp_{\mathsf{A_1}}}_{\mathsf{NA_1A_2}}.
    \end{align}
    Note that $\Abs{\bra{0}_{\mathsf{A_1}}\ket{0^\perp_{\mathsf{A_1}}}_{\mathsf{NA_1A_2}}} = 0$ by construction.
    
    The last statement follows from $V_{\mathsf{A_1T}} = \ketbra{0}{0}_{\mathsf{A_1}} \otimes I_{\mathsf{T}} + \Pi^{\neq 0}_{\mathsf{A_1}} \otimes X_{\mathsf{T}}$, which follows directly from the definition.
\end{proof}

\paragraph{Common gadget.}

The two parity-specific pivot matrices share a common factor $\rbra{\sigma_i\sigma_j\sigma_i}^s$.
We therefore first build a block-encoding of this shared factor, which both subsequent constructions reuse.

First, we construct a block-encoding of $\sigma_i$ for each $i \in \cbra{0,\ldots,K-1}$.
Given the $\rbra{n+n_\sigma}$-qubit unitary $\calO_{\sigma_i}$ providing purified access to $\sigma_i$, \cref{lem:BE_density_operator} constructs a unitary $U_{\sigma_i}$ that is a $\rbra{1,n+n_\sigma,0}$-block-encoding of $\sigma_i$.
Note that each $U_{\sigma_i}$ takes $O\rbra{1}$ queries to $\calO_{\sigma_i}$.

For $i,j \in \cbra{0,\ldots,K-1}$, we next construct a block-encoding of the subnormalized density operator $A_{i,j} = \sigma_i\sigma_j\sigma_i$.
We use \cref{lem:product_BE} to construct the $\rbra{4n+3n_\sigma}$-qubit unitary operator $U_{A_{i,j}}$, which is $\rbra{1,3n+3n_\sigma,0}$-block-encoding of subnormalized density operator $A_{i,j}$.

For $i,j \in \cbra{0,\ldots,K-1}$, we then construct a block-encoding of $\frac{1}{2}p\rbra{\sigma_i\sigma_j\sigma_i}$, where $p\rbra{x} \approx x^s$.
By~\cref{thm:poly_approx_monomials}, there exists a polynomial $p \in \R\sbra{x}$ of degree $O\rbra{\sqrt{s\log\rbra{1/\eps_1}}}$ and parity $s \bmod 2$ such that
\begin{align}\label{eq:monomial_approx}
    \abs*{p\rbra*{x}-x^s} \leq \eps_1 \textup{ for } x \in \interval[open]{-1}{1}
\end{align}
and
\begin{align}\label{eq:monomial_bound}
    \abs*{p\rbra*{x}} \leq 1 \textup{ for } x \in \interval[open]{-1}{1}.
\end{align}

By~\cref{thm:qsvt}, with the polynomial $\frac{1}{2}p$, $\alpha \coloneqq 1$, $a \coloneqq 3n+3n_\sigma$ and $\eps \coloneqq 0$, we can implement the unitary operator $U_{p\rbra{A_{i,j}}}$ which is a $\rbra{1,3n+3n_\sigma+2,\delta_2}$-block-encoding of $\frac{1}{2}p\rbra{\sigma_i\sigma_j\sigma_i}$, by using $O\rbra{\sqrt{s\log\rbra{\frac{1}{\eps_1}}}}$ queries to $U_{A_{i,j}}$ and the circuit description of $U_{p\rbra{A_{i,j}}}$ can be computed in classical time $\poly\rbra{s,\log{\frac{1}{\delta_2}}}$.
Suppose that $U_{p\rbra{A_{i,j}}}$ is the $\rbra{1,3n+3n_\sigma+2,0}$-block-encoding of $P_{i,j}$ that satisfies
\begin{align}\label{eq:monomial_round}
    \Abs*{P_{i,j}-\frac{1}{2}p\rbra*{\sigma_i\sigma_j\sigma_i}} \leq \delta_2 \quad \textup{ and } \quad \Abs{P_{i,j}} \leq 1.
\end{align}

\paragraph{Even order gadget.}

For even $t = 2s$, the pivot matrix has the symmetric form
\begin{equation}
    \rbra{\sigma_i\sigma_j\sigma_i}^s \sigma_i \rbra*{\sigma_i\sigma_j\sigma_i}^s.
\end{equation}
Sandwiching $\sigma_i$ by the common factor on both sides produces a subnormalized density operator that we can prepare directly.

For $i,j \in \cbra{0,\ldots,K-1}$, construct the unitary $U_{B_{i,j}}^{\textup{even}}$ that approximately prepares
\begin{equation}
    \frac{1}{4}p\rbra{\sigma_i\sigma_j\sigma_i}\sigma_i{p\rbra{\sigma_i\sigma_j\sigma_i}}^\dagger.
\end{equation}
Because $\calO_{\sigma_i}$ is an $\rbra{n+n_\sigma}$-qubit unitary that prepares $\sigma_i$, \cref{lem:sandwiched_product} lets us construct $U_{B_{i,j}}^{\textup{even}}=\rbra{U_{p\rbra{A_{i,j}}} \otimes I}\rbra{\calO_{\sigma_i} \otimes I}$. This unitary prepares the $n$-qubit subnormalized density operator
\begin{align}\label{eq:evolution_even}
    E_{i,j}^{\textup{even}} = P_{i,j}\sigma_i{P_{i,j}}^\dagger.
\end{align} 

By~\cref{lem:trace_est_subnormalized}, we obtain the following result.
\begin{proposition}\label{prop:construct_ub_even}
    Suppose that the $\rbra{4n+4n_\sigma}$-qubit unitary operator $U_{B_{i,j}}^{\textup{even}}$ can prepare $\rbra{n+n_{\eta}}$-qubit density operator $\eta$ where $4n+4n_\sigma=n+n_{\eta}+n_{\bar{\eta}}$, which is a $\rbra{1,n_\eta,0}$-block-encoding of $E_{i,j}^{\textup{even}}$.
    Then
    \begin{align}
        U_{B_{i,j}}^{\textup{even}}\ket{0}_{\mathsf{N}\mathsf{N}_{\eta}\mathsf{N}_{\bar{\eta}}}
        = \sqrt{\tr\rbra{E_{i,j}^{\textup{even}}}}\ket{0}_{\mathsf{N}_\eta}\ket{\phi_0}_{\mathsf{N }\mathsf{N}_{\bar{\eta}}}
        + \sqrt{1-\tr\rbra{E_{i,j}^{\textup{even}}}}\ket{0^\perp_{\mathsf{N}_\eta}}_{\mathsf{N}\mathsf{N}_\eta\mathsf{N}_{\bar{\eta}}} ,
    \end{align}    
    where $\mathsf{N}$, $\mathsf{N}_{\eta}$, and $\mathsf{N}_{\bar{\eta}}$ are $n$-, $n_{\eta}$-, and $n_{\bar{\eta}}$-qubit registers, respectively, and $\Abs{\bra{0}_{\mathsf{N}_{\eta}}\ket{0^\perp_{\mathsf{N}_{\eta}}}_{\mathsf{N}\mathsf{N}_{\eta}\mathsf{N}_{\bar{\eta}}}}=0$.
\end{proposition}

To apply~\cref{thm:amp_est}, append a new single-qubit register $\ket{0}_{\mathsf{T}}$.
We furthermore construct the following unitary $U_{C_{i,j}}^{\textup{even}}$ such that
\begin{align}
    U_{C_{i,j}}^{\textup{even}}
    = \rbra*{X^{\otimes n_{\eta}} \otimes I_{\mathsf{N}\mathsf{N}_{\bar{\eta}}} \otimes X_{\mathsf{T}}} \rbra*{\textup{Ctrl}_{\mathsf{N}_\eta}\textup{-}X_{\mathsf{T}} \otimes I_{\mathsf{N}\mathsf{N}_{\bar{\eta}}}}
    \rbra*{X^{\otimes n_{\eta}} \otimes I_{\mathsf{N}\mathsf{N}_{\bar{\eta}}} \otimes I_{\mathsf{T}}} \rbra*{U_{B_{i,j}}^{\textup{even}} \otimes I_{\mathsf{T}}} .
\end{align}

More specifically, the unitary operator $U_{C_{i,j}}^{\textup{even}}$ works as follows.

\begin{proposition}\label{prop:construct_uc_even}
    The unitary operator $U_{C_{i,j}}^{\textup{even}}=\mathsf{AmpPurification}\rbra{U_{B_{i,j}}^{\textup{even}}}$ prepares the following state
    \begin{align}
        U_{C_{i,j}}^{\textup{even}}\ket{0}_{\mathsf{N}\mathsf{N}_{\eta}\mathsf{N}_{\bar{\eta}}\mathsf{T}}
        = \sqrt{\tr\rbra*{E_{i,j}^{\textup{even}}}}\ket{0}_{\mathsf{T}}\ket{\xi_0}_{\mathsf{N}\mathsf{N}_{\eta}\mathsf{N}_{\bar{\eta}}}
        +
        \sqrt{1-\tr\rbra*{E_{i,j}^{\textup{even}}}}\ket{1}_{\mathsf{T}}\ket{\xi_1}_{\mathsf{N}\mathsf{N}_\eta\mathsf{N}_{\bar{\eta}}} ,
    \end{align}
    where $\braket{\xi_0}{\xi_1}=0$.
\end{proposition}

\begin{proof}
    By~\cref{prop:construct_ub_even}, suppose that $\ket{0^\perp_{\mathsf{N}_\eta}}_{\mathsf{N}\mathsf{N}_\eta\mathsf{N}_{\bar{\eta}}}$ has the following expansion:
    \begin{align}
        \ket{0^\perp_{\mathsf{N}_\eta}}_{\mathsf{N}\mathsf{N}_\eta\mathsf{N}_{\bar{\eta}}}
        = \sum_{l=0}^{2^n-1}\sum_{r=1}^{2^{n_\eta}-1}\sum_{m=0}^{2^{n_{\bar{\eta}}}-1}\frac{\beta_{l,r,m}}{\sqrt{1-\tr\rbra*{E_{i,j}^{\textup{even}}}}}\ket{l}_{\mathsf{N}}\ket{r}_{\mathsf{N}_\eta}\ket{m}_{\mathsf{N}_{\bar{\eta}}}.
    \end{align}
    
    We expand the expression step by step:
    \begin{align}
        & U_{C_{i,j}}^{\textup{even}}\ket{0}_{\mathsf{N}\mathsf{N}_{\eta}\mathsf{N}_{\bar{\eta}}\mathsf{T}} \\
        = {} & \rbra*{X^{\otimes n_{\eta}} \otimes I_{\mathsf{N}\mathsf{N}_{\bar{\eta}}} \otimes X_{\mathsf{T}}} \rbra*{\textup{Ctrl}_{\mathsf{N}_\eta}\textup{-}X_{\mathsf{T}} \otimes I_{\mathsf{N}\mathsf{N}_{\bar{\eta}}}}\rbra*{X^{\otimes n_{\eta}} \otimes I_{\mathsf{N}\mathsf{N}_{\bar{\eta}}} \otimes I_{\mathsf{T}}} \rbra*{U_{B_{i,j}}^{\textup{even}} \otimes I_{\mathsf{T}}} \ket{0}_{\mathsf{N}\mathsf{N}_{\eta}\mathsf{N}_{\bar{\eta}}\mathsf{T}}\\
        = {} & \rbra*{X^{\otimes n_{\eta}} \otimes I_{\mathsf{N}\mathsf{N}_{\bar{\eta}}} \otimes X_{\mathsf{T}}} \rbra*{\textup{Ctrl}_{\mathsf{N}_\eta}\textup{-}X_{\mathsf{T}} \otimes I_{\mathsf{N}\mathsf{N}_{\bar{\eta}}}}\rbra*{X^{\otimes n_{\eta}} \otimes I_{\mathsf{N}\mathsf{N}_{\bar{\eta}}} \otimes I_{\mathsf{T}}} \nonumber \\
        &\rbra*{\sqrt{\tr\rbra{E_{i,j}^{\textup{even}}}}\ket{0}_{\mathsf{N}_\eta}\ket{\phi_0}_{\mathsf{N}\mathsf{N}_{\bar{\eta}}} \otimes \ket{0}_{\mathsf{T}} +  \sum_{l=0}^{2^n-1}\sum_{r=1}^{2^{n_\eta}-1}\sum_{m=0}^{2^{n_{\bar{\eta}}}-1}\beta_{l,r,m}\ket{r}_{\mathsf{N}_\eta}\ket{l}_{\mathsf{N}}\ket{m}_{\mathsf{N}_{\bar{\eta}}} \otimes \ket{0}_{\mathsf{T}}}\\
        = {} & \rbra*{X^{\otimes n_{\eta}} \otimes I_{\mathsf{N}\mathsf{N}_{\bar{\eta}}} \otimes X_{\mathsf{T}}} \rbra*{\textup{Ctrl}_{\mathsf{N}_\eta}\textup{-}X_{\mathsf{T}} \otimes I_{\mathsf{N}\mathsf{N}_{\bar{\eta}}}} \nonumber \\
        &\rbra*{\sqrt{\tr\rbra*{E_{i,j}^{\textup{even}}}}\ket{\bar{1}}_{\mathsf{N}_\eta}\ket{\phi_0}_{\mathsf{N}\mathsf{N}_{\bar{\eta}}} \otimes \ket{0}_{\mathsf{T}} + \sum_{l=0}^{2^n-1}\sum_{r=1}^{2^{n_\eta}-1}\sum_{m=0}^{2^{n_{\bar{\eta}}}-1}\beta_{l,r,m}\ket{\neg r}_{\mathsf{N}_\eta}\ket{l}_{\mathsf{N}}\ket{m}_{\mathsf{N}_{\bar{\eta}}} \otimes \ket{0}_{\mathsf{T}}}\\
        = {} & \rbra*{X^{\otimes n_{\eta}} \otimes I_{\mathsf{N}\mathsf{N}_{\bar{\eta}}} \otimes X_{\mathsf{T}}} \nonumber \\ 
        &\rbra*{\sqrt{\tr\rbra*{E_{i,j}^{\textup{even}}}}\ket{\bar{1}}_{\mathsf{N}_\eta}\ket{\phi_0}_{\mathsf{N}\mathsf{N}_{\bar{\eta}}} \otimes \ket{1}_{\mathsf{T}} + \sum_{l=0}^{2^n-1}\sum_{r=1}^{2^{n_\eta}-1}\sum_{m=0}^{2^{n_{\bar{\eta}}}-1}\beta_{l,r,m}\ket{\neg r}_{\mathsf{N}_\eta}\ket{l}_{\mathsf{N}}\ket{m}_{\mathsf{N}_{\bar{\eta}}} \otimes \ket{0}_{\mathsf{T}}}\\
        = {} & \sqrt{\tr\rbra*{E_{i,j}^{\textup{even}}}}\ket{0}_{\mathsf{N}_\eta}\ket{\phi_0}_{\mathsf{N}\mathsf{N}_{\bar{\eta}}} \otimes \ket{0}_{\mathsf{T}} + \sum_{l=0}^{2^n-1}\sum_{r=1}^{2^{n_\eta}-1}\sum_{m=0}^{2^{n_{\bar{\eta}}}-1}\beta_{l,r,m}\ket{r}_{\mathsf{N}_\eta}\ket{l}_{\mathsf{N}}\ket{m}_{\mathsf{N}_{\bar{\eta}}} \otimes \ket{1}_{\mathsf{T}}\\
        = {} & \sqrt{\tr\rbra*{E_{i,j}^{\textup{even}}}}\ket{0}_{\mathsf{T}}\ket{0}_{\mathsf{N}_\eta}\ket{\phi_0}_{\mathsf{N}\mathsf{N}_{\bar{\eta}}} + \sqrt{1-\tr\rbra*{E_{i,j}^{\textup{even}}}}\ket{1}_{\mathsf{T}}\ket{0^\perp_{\mathsf{N}_\eta}}_{\mathsf{N}\mathsf{N}_{\eta}\mathsf{N}_{\bar{\eta}}},
    \end{align}
    where we have rearranged the registers in the last line.
    Renaming $\ket{\xi_0}_{\mathsf{N}\mathsf{N}_\eta\mathsf{N}_{\bar{\eta}}} = \ket{0}_{\mathsf{N}_\eta}\ket{\phi_0}_{\mathsf{N}\mathsf{N}_{\bar{\eta}}}$ and $\ket{\xi_1}_{\mathsf{N}\mathsf{N}_\eta\mathsf{N}_{\bar{\eta}}} = \ket{0^\perp_{\mathsf{N}_\eta}}_{\mathsf{N}\mathsf{N}_\eta\mathsf{N}_{\bar{\eta}}}$ gives the desired result.
\end{proof}

\paragraph{Odd-order gadget.}

For odd $t = 2s+1$, the pivot matrix is $\rbra*{\sigma_i\sigma_j\sigma_i}^s \sigma_i \sigma_j \sigma_i \rbra*{\sigma_i\sigma_j\sigma_i}^s$, i.e., an extra length-three product appears between the symmetric sandwiches.
Compared to the gadget for the even case, the evolution is now on the subnormalized density operator $\sigma_i\sigma_j\sigma_i$ rather than $\sigma_i$.

By~\cref{lem:sandwiched_product}, we can construct the $\rbra{2n+2n_\sigma}$-qubit unitary $U_{D_{i,j}}$ that prepares the subnormalized density operator $\sigma_i\sigma_j\sigma_i$.
For $i,j \in \cbra{0,\ldots,K-1}$, we construct the unitary $U_{B_{i,j}}^{\textup{odd}}$ that approximately prepares $\frac{1}{4}p\rbra{\sigma_i\sigma_j\sigma_i}\rbra{\sigma_i\sigma_j\sigma_i}{p\rbra{\sigma_i\sigma_j\sigma_i}}^\dagger$.
Applying~\cref{lem:sandwiched_product} again to the unitary $U_{D_{i,j}}$, which prepares the $n$-qubit subnormalized density operator $D_{i,j}$, we construct $U_{B_{i,j}}^{\textup{odd}}=\rbra{U_{p\rbra{A_{i,j}}} \otimes I}\rbra{U_{D_{i,j}} \otimes I}$. This unitary prepares the $n$-qubit subnormalized density operator
\begin{align}\label{eq:evolution_odd}
    E_{i,j}^{\textup{odd}} = P_{i,j}\rbra*{\sigma_i\sigma_j\sigma_i}{P_{i,j}}^\dagger.
\end{align} 

By~\cref{lem:trace_est_subnormalized}, we obtain the following result.
\begin{proposition}\label{prop:construct_ub_odd}
    Suppose that the unitary operator $U_{B_{i,j}}^{\textup{odd}}$ can prepare $\rbra{n+n_{\eta}}$-qubit density operator $\eta$ where $5n+5n_\sigma=n+n_{\eta}+n_{\bar{\eta}}$, which is a $\rbra{1,n_\eta,0}$-block-encoding of $E_{i,j}^{\textup{odd}}$.
    Then
    \begin{align}
        U_{B_{i,j}}^{\textup{odd}}\ket{0}_{\mathsf{N}\mathsf{N}_{\eta}\mathsf{N}_{\bar{\eta}}}
        = \sqrt{\tr\rbra{E_{i,j}^{\textup{odd}}}}\ket{0}_{\mathsf{N}_\eta}\ket{\phi_0}_{\mathsf{N}\mathsf{N}_{\bar{\eta}}}
+
        \sqrt{1-\tr\rbra{E_{i,j}^{\textup{odd}}}}\ket{0^\perp_{\mathsf{N}_\eta}}_{\mathsf{N}\mathsf{N}_\eta\mathsf{N}_{\bar{\eta}}} ,
    \end{align}    
    where $\mathsf{N}$, $\mathsf{N}_{\eta}$, $\mathsf{N}_{\bar{\eta}}$ are $n$-, $n_{\eta}$- and $n_{\bar{\eta}}$-qubit registers respectively and $\Abs*{\bra{0}_{\mathsf{N}_{\eta}}\ket{0^\perp_{\mathsf{N}_{\eta}}}_{\mathsf{N}\mathsf{N}_{\eta}\mathsf{N}_{\bar{\eta}}}}=0$.
\end{proposition}

To apply~\cref{thm:amp_est}, append a new single-qubit register $\ket{0}_{\mathsf{T}}$.
We furthermore construct the following unitary $U_{C_{i,j}}^{\textup{odd}}$ such that
\begin{align}
    U_{C_{i,j}}^{\textup{odd}}
    = {} & \rbra*{X^{\otimes n_{\eta}} \otimes I_{\mathsf{N}\mathsf{N}_{\bar{\eta}}} \otimes X_{\mathsf{T}}} \rbra*{\textup{Ctrl}_{\mathsf{N}_\eta}\textup{-}X_{\mathsf{T}} \otimes I_{\mathsf{N}\mathsf{N}_{\bar{\eta}}}}
    \rbra*{X^{\otimes n_{\eta}} \otimes I_{\mathsf{N}\mathsf{N}_{\bar{\eta}}} \otimes I_{\mathsf{T}}} \rbra*{U_{B_{i,j}}^{\textup{odd}} \otimes I_{\mathsf{T}}} .
\end{align}

More specifically, the unitary operator $U_{C_{i,j}}^{\textup{odd}}$ works as follows.

\begin{proposition}\label{prop:construct_uc_odd}
    The unitary operator $U_{C_{i,j}}^{\textup{odd}}=\mathsf{AmpPurification}\rbra{U_{B_{i,j}}^{\textup{odd}}}$ prepares the following state
    \begin{align}
        U_{C_{i,j}}^{\textup{odd}}\ket{0}_{\mathsf{N}\mathsf{N}_{\eta}\mathsf{N}_{\bar{\eta}}\mathsf{T}} 
        = \sqrt{\tr\rbra*{E_{i,j}^{\textup{odd}}}}\ket{0}_{\mathsf{T}}\ket{\xi_0}_{\mathsf{N}\mathsf{N}_\eta\mathsf{N}_{\bar{\eta}}}
        + \sqrt{1-\tr\rbra*{E_{i,j}^{\textup{odd}}}}\ket{1}_{\mathsf{T}}\ket{\xi_1}_{\mathsf{N}\mathsf{N}_\eta\mathsf{N}_{\bar{\eta}}}
    \end{align}
    where $\braket{\xi_0}{\xi_1}=0$.
\end{proposition}

The proof is analogous to that of~\cref{prop:construct_uc_even}.

\subsubsection{Algorithms}\label{sec:algorithm}

To help clarify our algorithm, we use the following structural theorem.

\begin{lemma}\label{lem:input_model_equivalence}
    Let $\calO_{\mathrm{ctrl}}=\sum_{k=0}^{K-1}\ketbra{k}{k}\otimes\calO_k$.
    For any quantum algorithm $\calA\sbra{\calO_i,\calO_j}$ that queries $\calO_i$ and $\calO_j$, there is a coherently controlled algorithm $\calA^\prime$ using $\calO_{\mathrm{ctrl}}$ such that
    \begin{align}
        \calA^\prime\sbra{\calO_{\mathrm{ctrl}}}
        = \sum_{i,j=0}^{K-1}\ketbra{i}{i}\otimes\ketbra{j}{j}
        \otimes\calA\sbra{\calO_i,\calO_j}.
    \end{align}
    Moreover, $\calA^\prime$ makes the same number of oracle queries as $\calA$.
\end{lemma}

\begin{proof}
    Assume that $\calA$ is a $T$-query quantum algorithm of the form
    \begin{align}
        \calA\sbra*{\calO_i,\calO_j} = U_T Q_T U_{T-1} Q_{T-1} \ldots U_1 Q_1 U_0
    \end{align}
    where $Q_t$, for $t=1,\ldots,T$, are query unitaries.
    Suppose there are two index sets $\calI, \calJ \subseteq \cbra{1,2,\ldots,T}$ where $\calI \cup \calJ = \cbra{1,2,\ldots,T}$ and $\calI \cap \calJ = \emptyset$ such that $Q_t = \calO_i$ if $t \in \calI$ and $Q_t = \calO_j$ if $t \in \calJ$. 
        
    Now append two additional registers, each containing $\log\rbra{K}$ qubits.
    \begin{align}
        \sum_{i=0}^{K-1}\sum_{j=0}^{K-1} \ketbra{i}{i} \otimes \ketbra{j}{j} \otimes \calA\sbra*{\calO_i,\calO_j}
        = {} &\rbra*{\sum_{i=0}^{K-1} \ketbra{i}{i} \otimes \sum_{j=0}^{K-1}\ketbra{j}{j} \otimes \rbra*{U_T Q_T}}\\
        & \rbra*{\sum_{i=0}^{K-1} \ketbra{i}{i} \otimes \sum_{j=0}^{K-1}\ketbra{j}{j} \otimes \rbra*{U_{T-1} Q_{T-1}}}\\
        & \ldots \\
        & \rbra*{\sum_{i=0}^{K-1} \ketbra{i}{i} \otimes \sum_{j=0}^{K-1}\ketbra{j}{j} \otimes \rbra*{U_1Q_1}}\\
        & \rbra*{I \otimes I \otimes U_{0}} .
    \end{align}
        
    Let
    \begin{align}
        V_t = \rbra*{\sum_{i=0}^{K-1} \ketbra{i}{i} \otimes \sum_{j=0}^{K-1}\ketbra{j}{j} \otimes U_tQ_t} V_{t-1}
    \end{align}
    for all $t = 1,\ldots,T$, and 
    \begin{align}
        V_0 = I \otimes I \otimes U_0.        
    \end{align}
    
    For every $t=1, \ldots, T$,
    \begin{align}
        V_t
        = & \rbra*{\sum_{i=0}^{K-1} \ketbra{i}{i} \otimes \sum_{j=0}^{K-1}\ketbra{j}{j} \otimes U_t Q_t} V_{t-1}\\
        = & \rbra*{I \otimes I \otimes U_t} \rbra*{\sum_{i=0}^{K-1} \ketbra{i}{i} \otimes \sum_{j=0}^{K-1}\ketbra{j}{j} \otimes Q_t} \\
        = & U_t^\prime Q_t^\prime V_{t-1}
    \end{align}
    where $U_t^\prime = \rbra{I \otimes I \otimes U_t}$ for $t = 0, \ldots, T$ and $Q_t^\prime = \sum_{i=0}^{K-1} \ketbra{i}{i} \otimes \sum_{j=0}^{K-1}\ketbra{j}{j} \otimes Q_t$ for $t=1, \ldots, T$.
    Moreover,
    \begin{align}
        Q_t^\prime =
        \begin{cases}
            \sum_{i=0}^{K-1} \ketbra{i}{i} \otimes I \otimes \calO_i, & \text{if } t \in \calI,\\
            I \otimes \sum_{j=0}^{K-1} \ketbra{j}{j} \otimes \calO_j, & \text{if } t \in \calJ.\\
        \end{cases}
    \end{align}
    The constructed quantum algorithm $\calA^\prime$ is
    \begin{align}
        \calA^\prime\sbra*{\sum_{k=0}^{K-1}\ketbra{k}{k}\otimes \calO_k} = U_T^\prime Q_T^\prime U_{T-1}^\prime Q_{T-1}^\prime \ldots U_1^\prime Q_1^\prime U_0^\prime
    \end{align}
    It is evident that the number of queries to $\sum_{k=0}^{K-1} \ketbra{k}{k} \otimes \calO_k$ is $T$.
\end{proof}

By~\cref{lem:input_model_equivalence}, we can describe the algorithm using the uncontrolled version $\calO_{\sigma_i}$ instead of $\sum_k\ketbra{k}{k}\otimes \calO_k$.
Our algorithms for estimating the trace of the pivot matrices are as follows.

\paragraph{Estimating the trace of an even-order pivot matrix.}

We describe the algorithm as follows and formally state it in \cref{alg:frame_potential_even}.

\begin{algorithm}[htbp]
    \caption{$\mathsf{TraceEvenOrderPivotMatrix}\rbra{\calO_\mu, \calO_S, t, \eps}$}
    \label{alg:frame_potential_even}
    \begin{algorithmic}[1]
        \REQUIRE Input oracles $\calO_\mu$, $\calO_S$, order $t \in \N$, and accuracy $\eps \in \interval[open]{0}{1}$.
        \ENSURE An estimate of $\tr\rbra{M_{\calE,s}^{\textup{even}}}$ within additive error $\eps$.

        \STATE $\eps_{\textup{AE}} \gets \eps/16$.
        \STATE $\delta_2 \gets \eps/16$.
        \STATE $\eps_1 \gets \eps/4$.
        \STATE $\delta_{\textup{AE}} \gets 1/4$.
        \STATE $s \gets t/2$.
        \STATE Let $p$ be the polynomial specified in \cref{thm:poly_approx_monomials} that approximates $x^s$ (with $\eps_1$).
        
        \FOR {$k=0,\ldots,K-1$}
            \STATE Construct unitary $U_{\sigma_k}$ (\cref{lem:BE_density_operator}).
            \Comment{$\rbra{1, n+n_\sigma, 0}$-block-encoding of $\sigma_k$.}
        \ENDFOR

        \FOR {$i=0,\ldots,K-1$ and $j=0,\ldots,K-1$}
            \STATE $U_{A_{i,j}} \gets \mathsf{BEProduct}\rbra{U_{\sigma_i},\mathsf{BEProduct}\rbra{U_{\sigma_j},U_{\sigma_i}}}$ (\cref{lem:product_BE})
            \Comment{$\rbra{1, 3n+3n_\sigma, 0}$-block-encoding of $A_{i,j}=\sigma_i\sigma_j\sigma_i$}
            \STATE $U_{p\rbra{A_{i,j}}} \gets \mathsf{EigenTrans}\rbra{U_{A_{i,j}},p,\delta_2}$.
            \Comment{$\rbra{1,3n+3n_\sigma+2,\delta_2}$-block-encoding of $\frac{1}{2}p\rbra{\sigma_i\sigma_j\sigma_i}$}
            \STATE $U_{B_{i,j}}^{\textup{even}} \gets \mathsf{SubDensityOpEvolution}\rbra{U_{p\rbra{A_{i,j}}},\calO_{\sigma_i}}$.
            \STATE $U_{C_{i,j}}^{\textup{even}} \gets \mathsf{AmpPurification}\rbra{U_{B_{i,j}}^{\textup{even}}}$.
        \ENDFOR

        \STATE Construct $U = \rbra{\textup{Ctrl}\textup{-}U^{\textup{even}}_{C_{ij}}}\rbra{\calO_{\mu} \otimes \calO_{\mu} \otimes I}$.

        \Comment{ $\calO_{\mu}$ act on registers $\mathsf{I},\mathsf{J}$ and ${\textup{Ctrl}\textup{-}U^{\textup{even}}_{C_{ij}}}$ on target registers $\mathsf{W},\mathsf{T}$ controlled on registers $\mathsf{I},\mathsf{J}$.}

        \STATE $\widetilde{X} \gets \mathsf{AmpEst}\rbra{U,\eps_{\textup{AE}},\delta_{\textup{AE}}}$.
        
        \STATE \textbf{return} $4\widetilde{X}$.
    \end{algorithmic}
\end{algorithm}

Starting from $\ket{0}_{\mathsf{I}}\ket{0}_{\mathsf{J}}$, we use $\calO_{\mu}$ to sample the pair $\rbra{i,j}$ according to the product distribution $\mu\otimes\mu$.
\begin{align}
    \rbra*{\calO_\mu \otimes \calO_\mu}\ket{0}_{\mathsf{I}}\ket{0}_{\mathsf{J}} = \sum_{i=0}^{K-1}\sum_{j=0}^{K-1}{\sqrt{\mu_i\mu_j}}\ket{i}_{\mathsf{I}}\ket{j}_{\mathsf{J}} .
\end{align}

Append another working space $\ket{0}_{\mathsf{T}} \ket{0}_{\mathsf{W}}$ where $\mathsf{W}$ is a $\rbra{4n+4n_\sigma}$-qubit register and $\mathsf{T}$ is one-qubit register. 
Then we apply the controlled-$U_{C_{i,j}}^{\textup{even}}$ on registers $\mathsf{T}$ and $\mathsf{W}$ controlled on registers $\mathsf{I}$ and $\mathsf{J}$ by~\cref{prop:construct_uc_even}.
Then we have
\begin{align}\label{eq:prepare_mixed_evan_pre}
    \sum_{i=0}^{K-1}\sum_{j=0}^{K-1}\sqrt{\mu_i\mu_j}\ket{i}_{\mathsf{I}}\ket{j}_{\mathsf{J}} \rbra*{ \sqrt{\tr\rbra*{E_{i,j}^{\textup{even}}}}\ket{0}_{\mathsf{T}}\ket{\xi_{i,j,0}}_{\mathsf{W}} 
    + \sqrt{1-\tr\rbra*{E_{i,j}^{\textup{even}}}}\ket{1}_
    {\mathsf{T}}\ket{\xi_{i,j,1}}_{\mathsf{W}} }, 
\end{align}
where $\braket{\xi_{i,j,0}}{\xi_{i,j,1}} = 0$.

Let $X=\sum_{i,j}\mu_i\mu_j \tr\rbra{E_{i,j}^{\textup{even}}}$, and let $\ket{\zeta_0}$ and $\ket{\zeta_1}$ be pure states on $\mathsf{IJW}$ defined as
\begin{align}
    \ket{\zeta_0}_{\mathsf{IJW}}
    = X^{-1/2} \sum_{i,j} \sqrt{\mu_i \mu_j \tr\rbra*{ E_{i,j}^{\textup{even}}}} \ket{i}_{\mathsf{I}}\ket{j}_{\mathsf{J}}\ket{\xi_{i,j,0}}_{\mathsf{W}}
\end{align}
and
\begin{align}
    \ket{\zeta_1}_{\mathsf{IJW}}
    = \rbra*{1-X}^{-1/2} \sum_{i,j} \sqrt{\mu_i \mu_j \rbra*{1- \tr\rbra*{ E_{i,j}^{\textup{even}}}}} \ket{i}_{\mathsf{I}}\ket{j}_{\mathsf{J}}\ket{\xi_{i,j,1}}_{\mathsf{W}},
\end{align}
respectively.
Using these states, \cref{eq:prepare_mixed_evan_pre} can be rewritten as
\begin{align}
    \sqrt{X}\ket{0}_{\mathsf{T}}\ket{\zeta_{0}}_{\mathsf{IJW}}+\sqrt{1-X}\ket{1}_{\mathsf{T}}\ket{\zeta_{1}}_{\mathsf{IJW}}.
\end{align}
If $X\in\cbra{0,1}$, the state multiplying the zero-amplitude branch may be chosen arbitrarily.

We then apply amplitude estimation (\cref{thm:amp_est}) to estimate the probability $X$ of observing $\ket{0}_{\mathsf{T}}$ to within additive error $\eps_{\textup{AE}}$, with success probability at least $1-\delta_{\textup{AE}}$. Denote the output by $\tilde{X}$.

\begin{align}
    \Prob\sbra*{\abs*{\tilde{X}-X} \leq \eps_{\textup{AE}}} \geq 1-\delta_{\textup{AE}}.
\end{align}

We return the estimate of $4\tilde{X}$ as the output of our algorithm.

Let $\eps_1$, $\delta_2$, $\eps_{\textup{AE}}$, $\delta_{\textup{AE}}$ be parameters to decide later. We list the following proposition to prove the correctness of our algorithm.

\begin{proposition}[QSVT rounding error of even order]\label{prop:qsvt_round_even}
    Let $p$ be the polynomial in \cref{eq:monomial_approx,eq:monomial_bound} and $\delta_2 \in \interval[open]{0}{1}$. For any two density operators $\sigma_i$ and $\sigma_j$,
    \begin{align}
        \abs*{\tr\rbra*{E_{i,j}^{\textup{even}}} - \tr\rbra*{\frac{1}{4}p\rbra*{\sigma_i\sigma_j\sigma_i} \sigma_i p\rbra{\sigma_i\sigma_j\sigma_i}^\dagger}} \leq 2\delta_2.
    \end{align}
\end{proposition}

\begin{proof}
    By \cref{eq:evolution_even}, we know
    \begin{align}
        \tr\rbra*{E_{i,j}^{\textup{even}}} = \tr\rbra*{P_{i,j}\sigma_i{P_{i,j}}^\dagger}.
    \end{align}
    Let $Q_{i,j}=\frac{1}{2}p\rbra{\sigma_i\sigma_j\sigma_i}$. We bound the difference as
    \begin{align}
        &\abs*{\tr\rbra*{P_{i,j}\sigma_i{P_{i,j}}^\dagger} - \tr\rbra*{Q_{i,j}\sigma_i Q_{i,j}^\dagger}} \nonumber \\
        &\leq \Abs*{\rbra{P_{i,j}-Q_{i,j}}\sigma_iP_{i,j}^\dagger}_1
        + \Abs*{Q_{i,j}\sigma_i\rbra{P_{i,j}^\dagger-Q_{i,j}^\dagger}}_1 \\
        &\leq \Abs*{P_{i,j}-Q_{i,j}}\Abs*{\sigma_i}_1\Abs*{P_{i,j}}
        + \Abs*{Q_{i,j}}\Abs*{\sigma_i}_1\Abs*{P_{i,j}-Q_{i,j}} \\
        &\leq 2\delta_2.
    \end{align}
    Here we used H\"older's inequality, $\Abs{\sigma_i}_1=1$, and the contraction bounds $\Abs{P_{i,j}},\Abs{Q_{i,j}}\leq1$.
\end{proof}

\begin{proposition}[QSVT approximation error of even order]\label{prop:qsvt_approx_even}
    Let $p$ be the polynomial in \cref{eq:monomial_approx,eq:monomial_bound} and $\eps_1 \in \interval[open]{0}{1}$. For any two density operators $\sigma_i$ and $\sigma_j$,
    \begin{align}
        \abs*{\tr\rbra*{p\rbra*{\sigma_i\sigma_j\sigma_i}\sigma_ip\rbra{\sigma_i\sigma_j\sigma_i}} - \tr\rbra*{\rbra*{\sigma_i\sigma_j\sigma_i}^s\sigma_i\rbra*{\sigma_i\sigma_j\sigma_i}^s}} \leq 2\eps_1.
    \end{align}
\end{proposition}

\begin{proof}
    Expand the left-hand side,
    \begin{align}
        &\abs*{\tr\rbra*{\pSIJ\sigma_i\pSIJ}-\tr\rbra*{\SIJs\sigma_i\SIJs}}\\
        & \leq \abs*{\tr\rbra*{\pSIJ\sigma_i\pSIJ}-\tr\rbra*{\SIJs\sigma_i \pSIJ}} \nonumber \\
        & \quad + \abs*{\tr\rbra*{\SIJs \sigma_i \pSIJ}-\tr\rbra*{\SIJs\sigma_i\SIJs}}\\
        & \leq \abs*{\tr\rbra*{\rbra*{\pSIJ - \SIJs} \sigma_i \pSIJ}} + \abs*{\tr\rbra*{\SIJs \sigma_i \rbra*{\pSIJ-\SIJs}}}\\
        & \leq \Abs*{ \rbra*{\pSIJ - \SIJs} \sigma_i \pSIJ }_1 + \Abs*{ \SIJs \sigma_i \rbra*{\pSIJ-\SIJs} }_1\\
        & \leq \Abs*{ \pSIJ - \SIJs }_{\infty} \Abs*{ \sigma_i \pSIJ }_1 + \Abs*{ \SIJs \sigma_i }_1 \Abs*{ \pSIJ-\SIJs }_{\infty}\\
        & \leq 2 \Abs*{\pSIJ - \SIJs}_{\infty}.
    \end{align}
    The last inequality follows from the fact that
    \begin{align}
        \Abs*{\sigma_i \pSIJ }_1 \leq \Abs*{\sigma_i}_1 \Abs*{ \pSIJ }_{\infty} \leq 1,
    \end{align}
    where we have used $\Abs{\sigma_i}_1 = 1$, and $\Abs{ \pSIJ }_{\infty}\leq 1$.
    The result follows from~\cref{eq:monomial_approx,eq:monomial_bound}.
\end{proof}

\begin{proposition}[Amplitude-estimation error of even order]\label{prop:amp_est_even}
    Let $\tilde{X}$, $\eps_1$, $\delta_2$, $\delta_{\textup{AE}}$, $\eps_{\textup{AE}}$, $s$ be the parameters specified in \cref{alg:frame_potential_even}. Then
    \begin{align}
        \abs*{4\tilde{X}-\tr\rbra*{M_{\calE,s}^{\textup{even}}}} \leq  4\eps_{\textup{AE}} + 8\delta_2 + 2\eps_1 .
    \end{align}
\end{proposition}

\begin{proof}
    By \cref{eq:dist_matrix_even},
    \begin{align}
        \tr\rbra*{M_{\calE,s}^{\textup{even}}} = \sum_{i}\sum_{j}\mu_i\mu_j\tr\rbra*{\SIJs \sigma_i \SIJs} .
    \end{align}
    According to \cref{thm:amp_est},
    \begin{align}
        \abs*{\tilde{X}-X} = \abs*{\tilde{X}-\sum_{i,j}\mu_i\mu_j\tr\rbra*{E_{i,j}^{\textup{even}}}} \leq \eps_{\textup{AE}} .
    \end{align}
    Moreover,
    \begin{align}
        \abs*{4\tilde{X} - 4\sum_{i,j}\mu_i\mu_j\tr\rbra*{P_{i,j} \sigma_i {P_{i,j}}^\dagger}}  \leq 4\eps_{\textup{AE}} .
    \end{align}
    Combining \cref{prop:qsvt_round_even,prop:qsvt_approx_even}, we bound the triangle inequality term by term and obtain
    \begin{align}
        \abs*{4\tilde{X}-\sum_{i,j}\mu_i\mu_j\tr\rbra*{\SIJs \sigma_i \SIJs}} \leq 4\eps_{\textup{AE}} + 8\delta_2 + 2\eps_1 .
    \end{align}
\end{proof}

\begin{theorem}[Estimate the trace of even-order mixed pivot matrix]\label{thm:main_estimate_even}
    Fix a positive integer $s$.
    Given an ensemble $\calE = \cbra{\rbra{\mu_k,\sigma_k}}$, \cref{alg:frame_potential_even} estimates $\tr\rbra{M_{\calE,s}^{\textup{even}}}$ to within additive error $\eps$ using $O\rbra{\sqrt{s\log{\rbra{1/\eps}}}/\eps}$ queries to $\calO_S$.
\end{theorem}

\begin{proof}
    Let $\eps_{\textup{AE}} = \eps/16$, $\delta_2 = \eps/16$ and $\eps_1 = \eps/4$, where $\eps \in \interval[open]{0}{1}$.
    By~\cref{prop:amp_est_even}, \cref{alg:frame_potential_even} can estimate $\tr\rbra{M_{\calE,s}^{\textup{even}}}$ to within additive error $\eps$ with probability at least $3/4$.
    Now we analyze the query complexity. The dominating part is from QSVT.
    By our choice of parameters, the degree of the polynomial $p$ is $O\rbra{\sqrt{s\log{\rbra{1/\eps_1}}}}=O\rbra{\sqrt{s\log{\rbra{1/\eps}}}}$, thus the algorithm takes $O\rbra{\sqrt{s\log{\rbra{1/\eps}}}}$ queries to $U_{A_{i,j}}$.
    Each use of $U_{A_{i,j}}$ makes two queries to $U_{\sigma_i}$ and one to $U_{\sigma_j}$, and each $U_{\sigma_k}$ uses $O\rbra{1}$ queries to $\calO_{\sigma_k}$. Multiplying by the number of amplitude-estimation rounds, the total query complexity is $O\rbra{\sqrt{s\log{\rbra{1/\eps}}}/\eps}$.
\end{proof}

\paragraph{Estimating the trace of an odd-order pivot matrix.}

We describe the algorithm as follows and formally state it in \cref{alg:frame_potential_odd}.

\begin{algorithm}[htbp]
    \caption{$\mathsf{TraceOddOrderPivotMatrix}\rbra{\calO_\mu, \calO_S, t, \eps}$}
    \label{alg:frame_potential_odd}
    \begin{algorithmic}[1]
        \REQUIRE Input oracles $\calO_\mu$, $\calO_S$, order $t \in \N$, and accuracy $\eps \in \interval[open]{0}{1}$.
        \ENSURE An estimate of $\tr\rbra{M_{\calE,s}^{\textup{odd}}}$ within additive error $\eps$.

        \STATE $\eps_{\textup{AE}} \gets \eps/16$.
        \STATE $\delta_2 \gets \eps/16$.
        \STATE $\eps_1 \gets \eps/4$
        \STATE $\delta_{\textup{AE}} \gets 1/4$
        \STATE $s \gets \rbra{t-1}/2$
        
        \STATE Let $p$ be the polynomial specified in \cref{thm:poly_approx_monomials} that approximates $x^s$ (with $\eps_1$).
        
        \FOR {$k \gets 0,\ldots,K-1$}
            \STATE Construct unitary $U_{\sigma_k}$ (\cref{lem:BE_density_operator}).
            \Comment{$\rbra{1, n+n_\sigma, 0}$-block-encoding of $\sigma_k$.}
        \ENDFOR

        \FOR {$i \gets 0, \ldots, K-1$ and $j \gets 0, \ldots, K-1$}
            \STATE $U_{A_{i,j}} \gets \mathsf{BEProduct}\rbra{U_{\sigma_i},\mathsf{BEProduct}\rbra{U_{\sigma_j},U_{\sigma_i}}}$ (\cref{lem:product_BE})
            \Comment{$\rbra{1, 3n+3n_\sigma, 0}$-block-encoding of $A_{i,j}=\sigma_i\sigma_j\sigma_i$}
            \STATE $U_{p\rbra{A_{i,j}}} \gets \mathsf{EigenTrans}\rbra{U_{A_{i,j}},p,\delta_2}$.
            \Comment{$\rbra{1,3n+3n_\sigma+2,\delta_2}$-block-encoding of $\frac{1}{2}p\rbra{\sigma_i\sigma_j\sigma_i}$}
            \STATE $U_{D_{i,j}} \gets \mathsf{SubDensityOpEvolution}\rbra{U_{\sigma_i},\calO_{\sigma_j}}$.
            \STATE $U_{B_{i,j}}^{\textup{odd}} \gets \mathsf{SubDensityOpEvolution}\rbra{U_{p\rbra{A_{i,j}}},U_{D_{i,j}}}$.
            \STATE $U_{C_{i,j}}^{\textup{odd}} \gets \mathsf{AmpPurification}\rbra{U_{B_{i,j}}^{\textup{odd}}}$.
        \ENDFOR

        \STATE Construct $U = \rbra{\textup{Ctrl}\textup{-}U^{\textup{odd}}_{C_{ij}}}\rbra{\calO_{\mu} \otimes \calO_{\mu} \otimes I}$.

        \Comment{ $\calO_{\mu}$ act on registers $\mathsf{I},\mathsf{J}$ and $\textup{Ctrl}\textup{-}U^{\textup{odd}}_{C_{ij}}$ on target registers $\mathsf{W},\mathsf{T}$ controlled on registers $\mathsf{I},\mathsf{J}$.}

        \STATE $\widetilde{X} \gets \mathsf{AmpEst}\rbra{U,\eps_{\textup{AE}},\delta_{\textup{AE}}}$.
        
        \STATE \textbf{return} $4\widetilde{X}$.
    \end{algorithmic}
\end{algorithm}

Starting from $\ket{0}_{\mathsf{I}}\ket{0}_{\mathsf{J}}$, we use $\calO_{\mu}$ to sample the pair $\rbra{i,j}$ according to the product distribution $\mu\otimes\mu$.
\begin{align}
    \rbra{\calO_\mu \otimes \calO_\mu}\ket{0}_{\mathsf{I}}\ket{0}_{\mathsf{J}}
    = \sum_{i=0}^{K-1}\sum_{j=0}^{K-1}{\sqrt{\mu_i\mu_j}}\ket{i}_{\mathsf{I}}\ket{j}_{\mathsf{J}} .
\end{align}

Append another working space $\ket{0}_{\mathsf{T}} \ket{0}_{\mathsf{W}}$ where $\mathsf{W}$ is a $\rbra{5n+5n_\sigma}$-qubit register and $\mathsf{T}$ is one-qubit register.
Then we apply the controlled-$U_{C_{i,j}}^{\textup{odd}}$ on registers $\mathsf{T}$ and $\mathsf{W}$ controlled on registers $\mathsf{I}$ and $\mathsf{J}$ by \cref{prop:construct_uc_odd}.
Then we prepare the state
\begin{align}
    \sum_{i=0}^{K-1}\sum_{j=0}^{K-1}\sqrt{\mu_i\mu_j}\ket{i}_{\mathsf{I}}\ket{j}_{\mathsf{J}}
    \rbra*{ \sqrt{\tr\rbra*{E_{i,j}^{\textup{odd}}}}\ket{0}_{\mathsf{T}}\ket{\xi_{i,j,0}}_{\mathsf{W}} +\sqrt{1-\tr\rbra*{E_{i,j}^{\textup{odd}}}}\ket{1}_{\mathsf{T}}\ket{\xi_{i,j,1}}_{\mathsf{W}} } ,
\end{align}
where $\braket{\xi_{i,j,0}}{\xi_{i,j,1}}=0$.
Let $X=\sum_{i,j}\mu_i\mu_j \tr\rbra*{E_{i,j}^{\textup{odd}}}$, and let $\ket{\zeta_0}$ and $\ket{\zeta_1}$ be pure states on $\mathsf{IJW}$ defined as
\begin{align}
    \ket{\zeta_{0}}_{\mathsf{IJW}}
    = X^{-1/2}\sum_{i,j}\sqrt{\mu_i\mu_j \tr\rbra*{E_{i,j}^{\textup{odd}}}}\ket{i}_{\mathsf{I}}\ket{j}_{\mathsf{J}}\ket{\xi_{i,j,0}}_{\mathsf{W}}
\end{align}
and
\begin{align}
    \ket{\zeta_{1}}_{\mathsf{IJW}}
    = \rbra*{1-X}^{-1/2}\sum_{i,j}\sqrt{\mu_i\mu_j \rbra*{1-\tr\rbra*{E_{i,j}^{\textup{odd}}}}}\ket{i}_{\mathsf{I}}\ket{j}_{\mathsf{J}}\ket{\xi_{i,j,1}}_{\mathsf{W}} .
\end{align}
Rearranging gives
\begin{align}\label{eq:prepare_mixed_odd}
    {}&\sqrt{X}\ket{0}_{\mathsf{T}}\ket{\zeta_{0}}_{\mathsf{IJW}} \\
    {}&\quad+ \sqrt{1-X}\ket{1}_{\mathsf{T}}\ket{\zeta_{1}}_{\mathsf{IJW}}.
\end{align}
If $X\in\cbra{0,1}$, the state multiplying the zero-amplitude branch may be chosen arbitrarily.

We then apply amplitude estimation (\cref{thm:amp_est}) to estimate the probability $X$ of observing $\ket{0}_{\mathsf{T}}$ to within additive error $\eps_{\textup{AE}}$, with success probability at least $1-\delta_{\textup{AE}}$.
Suppose the output of the amplitude estimation is $\tilde{X}$.
\begin{align}
    \Prob\sbra*{\abs*{\tilde{X}-X} \leq \eps_{\textup{AE}}} \geq 1 - \delta_{\textup{AE}}.
\end{align}

We return the estimate of $4\tilde{X}$ as the output of our algorithm.

Let $\eps_1$, $\delta_2$, $\eps_{\textup{AE}}$, $\delta_{\textup{AE}}$ be parameters to decide later.
We list the following propositions to prove the correctness of our algorithm.

\begin{proposition}[QSVT rounding error of odd order]\label{prop:qsvt_round_odd}
    Let $p$ be the polynomial in \cref{eq:monomial_approx,eq:monomial_bound}.
    And let $\delta_2 \in \interval[open]{0}{1}$.
    For any two density operators $\sigma_i$ and $\sigma_j$,
    \begin{align}
        \abs*{\tr\rbra*{E_{i,j}^{\textup{odd}}} - \tr\rbra*{\frac{1}{4}p\rbra*{\sigma_i\sigma_j\sigma_i} \rbra*{\sigma_i\sigma_j\sigma_i} p\rbra*{\sigma_i\sigma_j\sigma_i}^\dagger}} \leq 2\delta_2.
    \end{align}
\end{proposition}

\begin{proof}
    By \cref{eq:evolution_odd},
    \begin{align}
        \tr\rbra*{E_{i,j}^{\textup{odd}}} = \tr\rbra*{P_{i,j} \rbra*{\SIJ} {P_{i,j}}^\dagger}.
    \end{align}
    The proof is analogous to that of \cref{prop:qsvt_round_even}, with $\sigma_i$ replaced by $\rbra{\SIJ}$ and using $\Abs{\SIJ}_1 \leq 1$.
    The result is $2\delta_2$.
\end{proof}

\begin{proposition}[QSVT approximation error of odd order]\label{prop:qsvt_approx_odd}
    Let $p$ be the polynomial in \cref{eq:monomial_approx,eq:monomial_bound}. And let $\eps_1 \in \interval[open]{0}{1}$. For any two density operators $\sigma_i$ and $\sigma_j$
    \begin{align}
        \abs*{ \tr\rbra*{p\rbra*{\sigma_i\sigma_j\sigma_i}\rbra*{\sigma_i\sigma_j\sigma_i}p\rbra*{\sigma_i\sigma_j\sigma_i}}
        -\tr\rbra*{\rbra*{\sigma_i\sigma_j\sigma_i}^s\rbra*{\sigma_i\sigma_j\sigma_i}\rbra*{\sigma_i\sigma_j\sigma_i}^s} } \leq 2\eps_1.
    \end{align}
\end{proposition}

\begin{proof}
    The proof is analogous to that of \cref{prop:qsvt_approx_even}, with $\sigma_i$ replaced by $\rbra{\SIJ}$.
    The result is $2\Abs{\pSIJ - \SIJs}$.
    The result follows from~\cref{eq:monomial_approx,eq:monomial_bound}.
\end{proof}

\begin{proposition}[Amplitude-estimation error of odd order]\label{prop:amp_est_odd}
    Let $\tilde{X}$, $\eps_1$, $\delta_2$, $\delta_{\textup{AE}}$, $\eps_{\textup{AE}}$, $s$ be the parameters specified in \cref{alg:frame_potential_odd}. Then
    \begin{align}
        \abs*{4\tilde{X}-\tr\rbra*{M_{\calE,s}^{\textup{odd}}}}
        \leq 4\eps_{\textup{AE}} + 8\delta_2 + 2\eps_1 .
    \end{align}
\end{proposition}

\begin{proof}
    By \cref{eq:dist_matrix_odd},
    \begin{align}
        \tr\rbra*{M_{\calE,s}^{\textup{odd}}} = \sum_{i,j}\mu_i\mu_j\tr\rbra*{\SIJs \rbra*{\SIJ} \SIJs}.
    \end{align}
    The proof is analogous to that of \cref{prop:amp_est_even}, with $\sigma_i$ replaced by $\rbra{\SIJ}$.
    Combining \cref{prop:qsvt_round_odd,prop:qsvt_approx_odd} gives
    \begin{align}
        \abs*{4\tilde{X}-\sum_{i,j}\mu_i\mu_j\tr\rbra*{\SIJs \rbra*{\SIJ} \SIJs}} \leq 4\eps_{\textup{AE}} + 8\delta_2 + 2\eps_1 .
    \end{align}
\end{proof}

\begin{theorem}[Estimate the trace of odd-order mixed pivot matrix]\label{thm:main_estimate_odd}
    Fix a positive integer $s$.
    Given an ensemble $\calE = \cbra{\rbra{\mu_k,\sigma_k}}$, \cref{alg:frame_potential_odd} estimates $\tr\rbra{M_{\calE,s}^{\textup{odd}}}$ to within additive error $\eps$ using $O\rbra{\sqrt{s\log{\rbra{1/\eps}}}/\eps}$ queries to $\calO_S$.
\end{theorem}

\begin{proof}
    Let $\eps_{\textup{AE}} = \eps/16$, $\delta_2 = \eps/16$ and $\eps_1 = \eps/4$, where $\eps \in \interval[open]{0}{1}$.
    By~\cref{prop:amp_est_odd}, \cref{alg:frame_potential_odd} can estimate $\tr\rbra{M_{\calE,s}^{\textup{odd}}}$ to within additive error $\eps$ with probability at least $3/4$.
    Now we analyze the query complexity.
    The dominant cost comes from QSVT. The unitary $U_{D_{i,j}}$ can be constructed using $O\rbra{1}$ queries to each of $\calO_i$ and $\calO_j$.
    By our choice of parameters, the degree of the polynomial $p$ is $O\rbra{\sqrt{s\log{\rbra{1/\eps_1}}}}=O\rbra{\sqrt{s\log{\rbra{1/\eps}}}}$, thus the algorithm takes $O\rbra{\sqrt{s\log{\rbra{1/\eps}}}}$ queries to $U_{A_{i,j}}$.
    Each use of $U_{A_{i,j}}$ makes two queries to $U_{\sigma_i}$ and one to $U_{\sigma_j}$, and each $U_{\sigma_k}$ uses $O\rbra{1}$ queries to $\calO_{\sigma_k}$.
    Multiplying by the rounds of amplitude estimation, the total number of queries is $O\rbra{\sqrt{s\log{\rbra{1/\eps}}}/\eps}$.
\end{proof}

\paragraph{State frame potential estimator.} Now we are ready to construct our estimator for state frame potential. The correctness is demonstrated by the following two propositions.

\begin{proposition}[Even order]\label{prop:even_order_pure}
    If each $\sigma_i$ is a pure-state density operator and $t$ is even, then
    \begin{align}
        \calF_t\rbra*{\calE} = \tr\rbra*{M_{\calE,t/2}^{\textup{even}}}.
    \end{align}
\end{proposition}

\begin{proof}
    Let $t=2s$.
    Since $\tr\rbra{M_{\calE,s}^{\textup{even}}} = \sum_{i}\sum_{j}\mu_i\mu_j\tr\rbra{M_{i,j,s}^{\textup{even}}}$,
    \begin{align}
        \tr\rbra{M_{i,j,s}^{\textup{even}}}
        & = \tr\rbra*{\SIJs \sigma_i \SIJs} \\
        & = \tr\rbra*{ \rbra*{ \KBp{i}\KBp{j}\KBp{i} }^{s} \KBp{i} \rbra*{ \KBp{i}\KBp{j}\KBp{i} }^{s} }\\
        & = \tr\rbra*{ \rbra*{ \KBp{i}\KBp{j}\KBp{i} }^{s-1} \KBp{i} \rbra*{ \KBp{i}\KBp{j}\KBp{i} }^{s-1} } \cdot \abs*{\braket{\psi_i}{\psi_j}}^4 \\
        & = \abs*{\braket{\psi_i}{\psi_j}}^{4s}.
    \end{align}
    Therefore,
    \begin{align}
        \sum_{i=0}^{K-1}\sum_{j=0}^{K-1}\mu_i\mu_j\tr\rbra*{M_{i,j,s}^{\textup{even}}}
        = \Ex_{i,j \sim \mu}\sbra*{\abs*{\braket{\psi_i}{\psi_j}}^{4s}}
        = \Ex_{i,j \sim \mu}\sbra*{\abs*{\braket{\psi_i}{\psi_j}}^{2t}}
        = \calF_t\rbra{\calE}.
    \end{align}
\end{proof}

\begin{proposition}[Odd order]\label{prop:odd_order_pure}
    If each $\sigma_i$ is a pure-state density operator and $t$ is odd, then
    \begin{align}
        \calF_t\rbra*{\calE} = \tr\rbra*{M_{\calE,(t-1)/2}^{\textup{odd}}}.
    \end{align}
\end{proposition}

\begin{proof}
    Let $t=2s+1$.
    Since $\tr\rbra{M_{\calE,s}^{\textup{odd}}} = \sum_{i}\sum_{j}\mu_i\mu_j\tr\rbra{M_{i,j,s}^{\textup{odd}}}$,
    \begin{align}
        \tr\rbra*{M_{i,j,s}^{\textup{odd}}}& = \tr\rbra*{\SIJs \SIJ \SIJs}\\
        & = \tr\rbra*{ \rbra*{ \KBp{i}\KBp{j}\KBp{i} }^{s} \KBp{i} \KBp{j} \KBp{i} \rbra*{ \KBp{i}\KBp{j}\KBp{i} }^{s} }\\
        & = \tr\rbra*{ \rbra*{ \KBp{i}\KBp{j}\KBp{i} }^{s} \KBp{i} \rbra*{ \KBp{i}\KBp{j}\KBp{i} }^{s} } \cdot \abs*{\braket{\psi_i}{\psi_j}}^{2}\\
        & = \tr\rbra*{ \rbra*{ \KBp{i}\KBp{j}\KBp{i} }^{s-1} \KBp{i} \rbra*{ \KBp{i}\KBp{j}\KBp{i} }^{s-1} } \cdot \abs*{\braket{\psi_i}{\psi_j}}^{6}\\
        & = \abs*{\braket{\psi_i}{\psi_j}}^{4s+2} .
    \end{align}
    Therefore,
    \begin{align}
        \sum_{i=0}^{K-1}\sum_{j=0}^{K-1}\mu_i\mu_j\tr\rbra*{M_{i,j,s}^{\textup{odd}}}
        = \Ex_{i,j \sim \mu}\sbra*{\abs*{\braket{\psi_i}{\psi_j}}^{4s+2}}
        = \Ex_{i,j \sim \mu}\sbra*{\abs*{\braket{\psi_i}{\psi_j}}^{2t}}
        = \calF_t\rbra{\calE}.
    \end{align}
\end{proof}

Combining the subroutines \cref{alg:frame_potential_even,alg:frame_potential_odd} in \cref{sec:algorithm}, we show our algorithm for estimating the state frame potential in \cref{alg:state_frame_potential}.

\begin{algorithm}[htbp]
    \caption{$\mathsf{StateFramePotentialEstimator}\rbra{\calO_\mu, \calO_S,t,\eps}$}
    \label{alg:state_frame_potential}
    \begin{algorithmic}[1]
        \REQUIRE Input oracles $\calO_\mu$, $\calO_S$, order $t \in \N$ and accuracy $\eps \in \interval[open]{0}{1}$.
        \ENSURE An estimate of $\calF_t\rbra{\calE}$ to within additive error $\eps$. 

        \IF {$t$ is even}
            \STATE \textbf{return} $\mathsf{TraceEvenOrderPivotMatrix}\rbra{\calO_\mu, \calO_S, t, \eps}$
            \Comment{\cref{alg:frame_potential_even}}
        \ELSE
            \STATE \textbf{return} $\mathsf{TraceOddOrderPivotMatrix}\rbra{\calO_\mu, \calO_S, t, \eps}$
            \Comment{\cref{alg:frame_potential_odd}}
        \ENDIF 
    \end{algorithmic}
\end{algorithm}

\begin{theorem}[Estimate the state frame potential]\label{thm:potential_upperbound}
    Fix an integer $t \geq 1$. Given a computationally efficient description of an ensemble of quantum states $\calE$, \cref{alg:state_frame_potential} estimates $\calF_t\rbra{\calE}$ to within additive error $\eps$ using $O\rbra{\sqrt{t\log\rbra{1/\eps}}/\eps}$ queries to $\calO_S$.
\end{theorem}

\begin{proof}
    Correctness follows from~\cref{prop:even_order_pure,prop:odd_order_pure}, and the complexity bound follows from~\cref{thm:main_estimate_even,thm:main_estimate_odd}.
\end{proof}

\subsection{Lower bound}\label{sec:query_lower}

We show a matching lower bound that the algorithm of~\cref{sec:query_upper} is optimal up to this logarithmic factor.
The lower bound rests on a distinguishability argument: a hypothetical algorithm using $o\rbra{\sqrt{t}/\eps}$ queries would distinguish two carefully chosen ensembles whose frame potentials differ by $\Theta\rbra{\eps}$, contradicting the Hellinger-distance-based query lower bound of \cite{Bel19}.

We first present a slightly weaker result about ensembles of constant cardinality.
The construction of hard instances will inform the more general result that follows.

\begin{theorem} \label{thm:query_lower_constant}
    Any quantum algorithm for estimating to within additive error $\varepsilon$ the state frame potential $\calF_t\rbra*{\calE}$ for any ensemble $\calE$ with cardinality $2$ requires $\Omega\rbra{\sqrt{t}/\varepsilon}$ queries to the state-preparation oracle $\calO_S$.
\end{theorem}

To prove~\cref{thm:query_lower_constant}, we use the following result of~\cite{Bel19} on the quantum query complexity of distinguishing probability distributions.

\begin{lemma}[{\cite[Theorem 4]{Bel19}}] \label[lemma]{lemma:bel19}
    Let $p$ and $q$ be two probability distributions on $K$ elements, with $\calO_p$ and $\calO_q$ as their sampling oracles.
    For an unknown unitary oracle $\calO$, any quantum algorithm that determines whether $\calO = \calO_p$ or $\calO = \calO_q$, given the promise that one of these two cases holds, requires $\Omega\rbra{1/d_{\textup{H}}\rbra{p, q}}$ queries to $\calO$, where
    \begin{align}
        d_{\textup{H}}\rbra*{p, q} = \sqrt{\frac{1}{2}\sum_{i=0}^{K-1} \rbra*{\sqrt{p_i} - \sqrt{q_i}}^2}
    \end{align}
    is the Hellinger distance. 
\end{lemma}

Now we are ready to prove \cref{thm:query_lower_constant}.

\begin{proof}[Proof of \cref{thm:query_lower_constant}]
    We use the hard instance taken from \cite{CWYZ26} and consider the problem in \cite{Wan25} that allows us to derive quantum query lower bounds.
    The problem is to distinguish the following two probability distributions $P^\pm$:
    \begin{align}
        p_0^\pm = 1 - \frac{1}{t} \pm \frac{\varepsilon}{t}, \qquad p_1^\pm = \frac{1}{t} \mp \frac{\varepsilon}{t},
    \end{align}
    where $\varepsilon \in \interval[open]{0}{1}$ is a parameter to be determined.
    The sampling oracles for $P^{\pm}$ are $U_{\pm}$ defined by
    \begin{align}\label{def:psi-pm}
        U_{\pm} \ket{0} = \ket{\psi_\pm} \coloneqq \sqrt{p_0^\pm} \ket{0} + \sqrt{p_1^\pm} \ket{1}.
    \end{align}
    It follows from~\cite[Equation (51)]{Wan25} that
    \begin{align}
        d_{\textup{H}}\rbra*{P^+, P^-} \leq \frac{\varepsilon}{\sqrt{t-1}}. 
    \end{align}
    By \cref{lemma:bel19}, any quantum algorithm that distinguishes $P^\pm$ by their sampling oracles requires query complexity $\Omega\rbra{1/d_{\textup{H}}\rbra{P^+, P^-}} = \Omega\rbra{\sqrt{t}/\varepsilon}$.
    
    On the other hand, suppose that there is a quantum algorithm for estimating the state frame potential $\calF_t\rbra{\calE}$ to within additive error $\varepsilon$, using $Q\rbra{t, \varepsilon}$ queries to the state-preparation oracle. 
    Now we consider the two ensembles $\mathcal{E}_{\pm}$:
    \begin{align} \label{eq:def-ensemble}
        \mathcal{E}_{\pm} = \cbra*{\rbra*{\frac{1}{2}, \ket{0}}, \rbra*{\frac{1}{2}, \ket{\psi_{\pm}}}}.
    \end{align}
    Then, for $t \geq 3$, we have
    \begin{align} \label{eq:potential-diff}
        \calF_t\rbra*{\mathcal{E}_+} - \calF_t\rbra*{\mathcal{E}_-} \geq \frac{1}{9} \varepsilon. 
    \end{align}
    To see this, since
    \begin{align}
        \calF_{t}\rbra*{\mathcal{E}_{\pm}}
        & = \frac{1}{4}\rbra*{\abs*{\braket{0}{0}}^{2t} + \abs*{\braket{0}{\psi_{\pm}}}^{2t} + \abs*{\braket{\psi_{\pm}}{0}}^{2t} + \abs*{\braket{\psi_{\pm}}{\psi_{\pm}}}^{2t}} \\
        & = \frac{1}{2}\rbra*{1 + \abs*{\braket{0}{\psi_{\pm}}}^{2t}} \\
        & = \frac{1}{2} + \frac{1}{2}\rbra*{p_0^{\pm}}^t,
    \end{align}
    we have
    \begin{align}\label{eq:fep-fem}       
        \calF_t\rbra*{\mathcal{E}_+} - \calF_t\rbra*{\mathcal{E}_-}
        & = \frac{1}{2}\rbra*{\rbra*{p_0^{+}}^t - \rbra*{p_0^{-}}^t} \\
        & = \frac{1}{2}\rbra*{\rbra*{\frac{p_0^{+}}{p_0^-}}^t - 1}\rbra*{p_0^{-}}^t \\
        & = \frac{1}{2}\rbra*{\rbra*{1+\frac{p_0^{+}-p_0^-}{p_0^-}}^t - 1}\rbra*{p_0^{-}}^t \\
        & \geq \frac{1}{2}\rbra*{1+\frac{p_0^{+}-p_0^-}{p_0^-}\cdot t - 1}\rbra*{p_0^{-}}^t \\
        & = \frac{1}{2}\rbra*{p_0^{+}-p_0^-}\cdot t \cdot \rbra*{p_0^{-}}^{t-1} \\
        & = \varepsilon \cdot \rbra*{1-\frac{1}{t} - \frac{\varepsilon}{t}}^{t-1} \\
        & \geq \varepsilon \cdot \rbra*{1-\frac{2}{t}}^{t-1} \\
        & \geq \frac{1}{9}\cdot \varepsilon,
    \end{align}
    where the first inequality holds because of $p_0^+ > p_0^-$ and $p_0^- > 0$ for $t \geq 3$.
    To estimate the state frame potential of each of the ensembles, the state-preparation oracle is defined by
    \begin{align}
        \mathcal{O}_S^{\pm} = \ketbra{0}{0} \otimes I + \ketbra{1}{1} \otimes U_{\pm},
    \end{align}
    which can be implemented using $1$ query to $U_\pm$, and the sampling oracle for both cases is
    \begin{align}
        \calO_\mu \ket{0} = \frac{1}{\sqrt{2}} \rbra*{ \ket{0} + \ket{1} }.
    \end{align}
    With this, the problem of distinguishing $U_+$ and $U_-$ is reduced to the problem of distinguishing $\mathcal{O}_S^{+}$ and $\mathcal{O}_S^{-}$. 
    By the quantum algorithm for estimating the state frame potential, we can estimate $\calF_{t}\rbra*{\mathcal{E}}$ to within additive error $\frac{1}{20}\varepsilon$ with $Q\rbra{t, \frac{1}{20}\varepsilon}$ queries to $\mathcal{O}_S$, and thus we can determine whether $\mathcal{O}_S = \mathcal{O}_S^+$ or $\mathcal{O}_S = \mathcal{O}_S^-$ according to the estimation results of $\calF_{t}\rbra{\mathcal{E}}$. 
    Therefore,
    \begin{align}
        Q\rbra*{t, \frac{1}{20}\varepsilon}
        \geq \Omega\rbra*{\frac{1}{d_{\textup{H}}\rbra*{P^+, P^-}}}
        = \Omega\rbra*{\frac{\sqrt{t}}{\varepsilon}}.
    \end{align}
    In other words, any quantum algorithm that estimates the state frame potential $\calF_t\rbra{\mathcal{E}}$ to within additive error $\Theta\rbra{\varepsilon}$ requires query complexity $\Omega\rbra{\sqrt{t}/\varepsilon}$. 
\end{proof}

\begin{theorem}\label{thm:query_lower_arbitrary}
For any integer $K \geq 2$, any quantum algorithm for estimating to within additive error $\varepsilon$ the state frame potential $\calF_t\rbra*{\calE}$ for any ensemble $\calE$ with cardinality $K$ requires $\Omega\rbra{\sqrt{t}/\varepsilon}$ queries to the state-preparation oracle $\calO_S$.
\end{theorem}

\begin{proof}
The proof is analogous to that of~\cref{thm:query_lower_constant}.
For a positive integer $t$ and a parameter $\eps \in \interval[open]{0}{1}$, we construct two ensembles,
\begin{align}\label{eq:genralized_ensemble}
    \mathcal{E}_{\pm} = \cbra*{ \rbra*{\tfrac{1}{2K}, \tfrac{t\ket{0}+\ket{2}}{\sqrt{t^2+1}}},
    \ldots, \rbra*{\tfrac{1}{2K}, \tfrac{t\ket{0}+\ket{K+1}}{\sqrt{t^2+1}}},
    \rbra*{\tfrac{1}{2K}, \tfrac{t\ket{\psi_\pm}+\ket{K+2}}{\sqrt{t^2+1}}},
    \ldots, \rbra*{\tfrac{1}{2K}, \tfrac{t\ket{\psi_\pm}+\ket{2K+1}}{\sqrt{t^2+1}}} },
\end{align}
where 
\begin{align}
    \ket{\psi_{\pm}} = \sqrt{p_{0}^{\pm}}\ket{0} + \sqrt{p_1^{\pm}}\ket{1}
\end{align}
and
\begin{align}
    p_0^{\pm} = 1-\frac{1}{t} \pm \frac{\varepsilon}{t}, \qquad p_1^{\pm} = \frac{1}{t} \mp \frac{\varepsilon}{t}.
\end{align}    

We first calculate $\calF_t\rbra{\calE_{+}} - \calF_t\rbra{\calE_{-}}$.
Since
\begin{align}
    \calF_t\rbra*{\calE_{\pm}} = {} & \frac{1}{4K^2}\sum_{j,k=0}^{K-1}\biggl( \abs*{\rbra*{\frac{t\bra{0}+\bra{2+j}}{\sqrt{t^2+1}}}\rbra*{\frac{t\ket{0}+\ket{2+k}}{\sqrt{t^2+1}}}}^{2t}\\ 
    & + \abs*{\rbra*{\frac{t\bra{0}+\bra{2+j}}{\sqrt{t^2+1}}}\rbra*{\frac{t\ket{\psi_{\pm}}+\ket{K+2+k}}{\sqrt{t^2+1}}}}^{2t} \\
    & + \abs*{\rbra*{\frac{t\bra{\psi_{\pm}}+\bra{K+2+j}}{\sqrt{t^2+1}}}\rbra*{\frac{t\ket{0}+\ket{2+k}}{\sqrt{t^2+1}}}}^{2t} \\
    & + \abs*{\rbra*{\frac{t\bra{\psi_{\pm}}+\bra{K+2+j}}{\sqrt{t^2+1}}}\rbra*{\frac{t\ket{\psi_{\pm}}+\ket{K+2+k}}{\sqrt{t^2+1}}}}^{2t}\biggr) \\
        = {} & \frac{1}{4K^2}\sum_{j,k=0}^{K-1}\rbra*{\abs*{\frac{t^2+\delta_{j,k}}{t^2+1}}^{2t} + \abs*{\frac{t^2\braket{0}{\psi_{\pm}}}{t^2+1}}^{2t} + \abs*{\frac{t^2\braket{\psi_{\pm}}{0}}{t^2+1}}^{2t} + \abs*{\frac{t^2+\delta_{j,k}}{t^2+1}}^{2t}} \\
        = {} & \frac{1}{4K^2}\sum_{j,k=0}^{K-1}\rbra*{\abs*{\frac{t^2+\delta_{j,k}}{t^2+1}}^{2t} + 2\cdot \rbra*{\frac{t^4\cdot p_0^{\pm}}{\rbra*{t^2+1}^2}}^{t} + \abs*{\frac{t^2+\delta_{j,k}}{t^2+1}}^{2t}},
\end{align}
we have
\begin{align}
    \calF_t\rbra*{\calE_{+}} - \calF_t\rbra*{\calE_{-}}
    = {} & \frac{1}{4K^2}\sum_{j,k=0}^{K-1}\rbra*{2\cdot \rbra*{\frac{t^4\cdot p_0^{+}}{\rbra*{t^2+1}^2}}^{t} - 2\cdot \rbra*{\frac{t^4\cdot p_0^{-}}{\rbra*{t^2+1}^2}}^{t}}\\
    = {} & \frac{1}{2}\rbra*{\rbra*{\frac{t^4\cdot p_0^{+}}{\rbra*{t^2+1}^2}}^{t} - \rbra*{\frac{t^4\cdot p_0^{-}}{\rbra*{t^2+1}^2}}^{t}}.
\end{align}
The same reasoning as in~\cref{eq:fep-fem} yields, for $t \geq 3$,
\begin{align}\label{eq:fep-fem2}
    \calF_t\rbra*{\calE_{+}} - \calF_t\rbra*{\calE_{-}}
    & \geq \frac{1}{2}\rbra*{\frac{t^4\cdot p_0^{+}}{\rbra*{t^2+1}^2} - \frac{t^4\cdot p_0^{-}}{\rbra*{t^2+1}^2}}\cdot t \cdot \rbra*{\frac{t^4\cdot p_0^{-}}{\rbra*{t^2+1}^2}}^{t-1} \\
    & = \frac{1}{2}\rbra*{p_0^+- p_0^-}\cdot t \cdot \frac{t^4}{\rbra*{t^2+1}^2}\cdot \rbra*{\frac{t^4\cdot p_0^{-}}{\rbra*{t^2+1}^2}}^{t-1} \\
    & = \varepsilon \cdot \rbra*{1-\frac{1}{t^2+1}}^{2t}\cdot \rbra*{p_0^{-}}^{t-1} \\
    & = \varepsilon \cdot \rbra*{1-\frac{1}{t^2+1}}^{2t}\cdot \rbra*{1-\frac{1}{t}-\frac{\varepsilon}{t}}^{t-1} \\
    & \geq \varepsilon \cdot \rbra*{1-\frac{1}{t^2+1}}^{2t}\cdot \rbra*{1-\frac{2}{t}}^{t-1}\\
    & \geq \frac{9^5}{10^6}\cdot \varepsilon.
\end{align}
The state preparation oracle for $\calE_{\pm}$ is defined by
\begin{align}
    \calO_S^{\pm} = \sum_{k=0}^{K-1}\ketbra{k}{k}\otimes U_k + \sum_{k=K}^{2K-1}\ketbra{k}{k}\otimes \rbra*{U_{\pm}U_k},
\end{align}
where $U_k$ and $U_\pm$ satisfy
\begin{align}
        U_k\ket{0} = \frac{t}{\sqrt{t^2+1}}\ket{0}+ \frac{1}{\sqrt{t^2+1}}\ket{2+k}
\end{align}
and
\begin{align}
    U_{\pm}\ket{0} = \ket{\psi_{\pm}} = \sqrt{p_0^\pm} \ket{0} + \sqrt{p_1^\pm} \ket{1}.
\end{align}
The implementation of $\calO_S^{\pm}$ requires only $1$ query to $U_{\pm}$.
With the sampling oracle
\begin{align}
    \calO_{\mu}\ket{0} = \frac{1}{\sqrt{2K}}\sum_{j=0}^{2K-1}\ket{j}
\end{align}
for both cases, the problem of distinguishing $U_+$ and $U_-$ is reduced to the problem of determining whether $\calO_S=\calO_S^+$ or $\calO_S = \calO_S^-$.
Since $U_{\pm}$ fits the definition of \cref{def:psi-pm} in the proof of \cref{thm:query_lower_constant}, the desired lower bound follows from the same argument applied to \cref{eq:fep-fem2}. 
\end{proof}

\section{General sample model}\label{sec:general_sample_model}

In experimental settings where one cannot invoke the state-preparation circuit but can collect many identical copies of a state drawn from the ensemble---for example, through repeated runs of the same simulator initialization---the relevant resource is the number of copies consumed rather than the number of oracle calls.
The general sample model formalizes this setting: one may obtain multiple copies of a state $\ket{\psi}$ drawn according to $\mu$.
The frame potential is then directly accessible via a generalized SWAP test applied to two independently drawn $t$-fold samples. Our algorithm builds on this observation.
In this section, we prove the matching sample-complexity bounds $\Theta\rbra{t/\eps^2}$.
The upper and lower bounds are stated in \cref{sec:sample_upper_bound} and \cref{sec:sample_lower_bound} respectively.

\subsection{Upper bound}\label{sec:sample_upper_bound}

Our algorithm is based on the generalized SWAP test of~\cite{EAO+02,KLL+17}.
Define $S_k$ to be the cyclic-permutation operator on $k$ quantum registers, i.e.,
\begin{align}\label{eq:cyclic_permutation_gate}
    S_k = \sum_{j_1,\ldots,j_k=1}^{2^n} C_{j_1,\ldots,j_k},
\end{align}
where 
\begin{align}
    C_{j_1,\ldots,j_k}=\ketbra{j_k}{j_1} \otimes \ketbra{j_1}{j_2} \otimes \ketbra{j_2}{j_3} \otimes \cdots \otimes \ketbra{j_{k-1}}{j_k}.
\end{align}
Equivalently, $S_k$ acts as $S_k\ket{j_1,j_2,\ldots,j_k}=\ket{j_k,j_1,\ldots,j_{k-1}}$. The corresponding generalized SWAP-test circuit is shown in~\cref{fig:swap_test}.

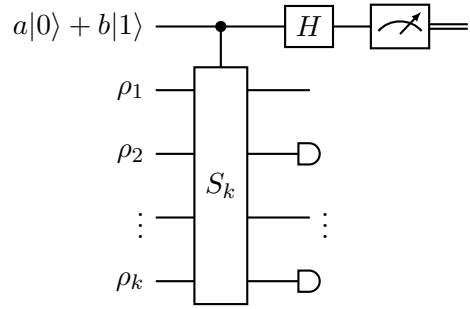
\begin{figure}[ht]
    \centering
    \begin{quantikz}[row sep = {24pt, between origins}]
        \lstick{$a\ket{0} + b\ket{1}$} & \ctrl{1} & \gate{H} & \meter{} & \setwiretype{c} \\
        \lstick{$\rho_1$} & \gate[4]{S_k} & \\
        \lstick{$\rho_2$} &               & \meterD{}\\
        \lstick{$\vdots$}   &               & \rstick{$\vdots$}\\
        \lstick{$\rho_k$} &               & \meterD{}\\
    \end{quantikz}
    \caption{The generalized SWAP-test circuit. Here, $S_k$ is the cyclic-permutation gate defined in~\cref{eq:cyclic_permutation_gate}. The last $k-1$ data registers are discarded (traced out), and the control qubit is measured in the computational basis.}
    \label{fig:swap_test}
\end{figure}

\begin{lemma}[Generalized SWAP test]\label{lem:swap_test}
    Given density operators $\rho_1,\ldots,\rho_k$ and a state $\ket{\phi}=a\ket{0}+b\ket{1}$, there is a quantum algorithm $\mathsf{SwapTest}\rbra*{\ket{\phi},\rho_1,\ldots,\rho_k}$ (see~\cref{fig:swap_test}) that returns the outcome $0$ with probability
    \begin{align}
        \frac{1}{2}\rbra*{1 + ab^\ast\tr\rbra*{\rho_1\rho_2\ldots\rho_k} + a^\ast b \tr\rbra*{\rho_k\rho_{k-1}\ldots\rho_1}},
    \end{align}
    using one copy of each input state $\rho_1,\ldots,\rho_k$.
\end{lemma}

\begin{proof}
    By~\cite[Equation (33)]{KLL+17}, after tracing out the last $k-1$ data registers, the diagonal blocks of the reduced state, with respect to the computational basis of the control qubit, are $\rho_+$ and $\rho_-$, where
    \begin{align}
        \rho_+ & {} = \frac{\abs{a}^2\rho_1 + \abs{b}^2\rho_k + ab^\ast\rho_1\rho_2\ldots\rho_k + a^\ast b \rho_k\rho_{k-1}\ldots\rho_1}{2} \\
        \rho_- & {} = \frac{\abs{a}^2\rho_1 + \abs{b}^2\rho_k - ab^\ast\rho_1\rho_2\ldots\rho_k - a^\ast b \rho_k\rho_{k-1}\ldots\rho_1}{2}.
    \end{align}
    Therefore, the probability that the measured outcome is $0$ is $\tr\rbra*{\rho_+}$, as claimed.
\end{proof}

Now we state our algorithm formally in~\cref{alg:state_frame_potential_sample}. 
\begin{algorithm}[htbp]
    \caption{$\mathsf{MultiCopyStateFramePotentialEstimator}\rbra{\calE,t,\eps}$}
    \label{alg:state_frame_potential_sample}
    \begin{algorithmic}[1]
        \REQUIRE An ensemble $\mathcal{E} = \cbra{\rbra{\mu_i, \ket{\psi_i}}}_{i=0}^{K-1}$, order $t \in \N$, and accuracy $\eps \in \interval[open]{0}{1}$.
        \ENSURE An estimate of $\calF_t\rbra{\calE}$ to within additive error $\eps$. 

        \STATE $M \gets \left\lceil \frac{8}{\eps^2}\right\rceil$.
        \FOR {$m=1,\ldots,M$}
            \STATE Sample $i \sim \mu$.
            \STATE Sample $j \sim \mu$.
            \STATE $\ket{\phi} \gets H\ket{0} = \frac{1}{\sqrt{2}}\ket{0}+\frac{1}{\sqrt{2}}\ket{1}$.
            \FOR {$k=0,\ldots,t-1$}
                \STATE $\rho_{2k+1} \gets \ketbra{\psi_i}{\psi_i}$.
                \STATE $\rho_{2k+2} \gets \ketbra{\psi_j}{\psi_j}$.
            \ENDFOR
            \STATE $x_m \gets \mathsf{SwapTest}\rbra{\ket{\phi},\rho_1,\ldots,\rho_{2t}}$. \Comment{$\Prob\sbra{x_m = 0} = \rbra{1 + \abs{\braket{\psi_i}{\psi_j}}^{2t}}/2$.}
        \ENDFOR
        \STATE \textbf{return} $X=1-\frac{2}{M} \sum_{m=1}^M x_m$. 
    \end{algorithmic}
\end{algorithm}

\begin{theorem}
    Fix an integer $t \geq 1$. Given sample access to an ensemble $\calE$ of quantum states, \cref{alg:state_frame_potential_sample} estimates $\calF_t\rbra{\calE}$ to within additive error $\eps$, with probability at least $3/4$, using $O\rbra{t/\varepsilon^2}$ samples from $\calE$.
\end{theorem}

\begin{proof}
    We first prove the correctness of our algorithm. 
    Condition on the sampled labels $I=i$ and $J=j$, where $\Prob\sbra{I=i}=\mu_i$ and $\Prob\sbra{J=j}=\mu_j$.
    Set $\ket{\phi}=\rbra{\ket{0}+\ket{1}}/\sqrt{2}$, $\rho_{2k+1} = \ketbra{\psi_i}{\psi_i}$, and $\rho_{2k+2} = \ketbra{\psi_j}{\psi_j}$ for $k = 0, \ldots, t-1$. By~\cref{lem:swap_test},
    \begin{align}
        \Prob\sbra*{x_m = 1 \vert I=i, J=j} = \frac{1 - \abs*{\braket{\psi_i}{\psi_j}}^{2t}}{2}.
    \end{align}
    By the conditional probability formula, we have
    \begin{align}
        2\Prob\sbra*{x_m = 1, I = i, J = j}
        = \Prob\sbra*{I = i, J = j} - \Prob\sbra*{I = i, J = j} \abs*{\braket{\psi_i}{\psi_j}}^{2t}.
    \end{align}
    Rearranging the terms,
    \begin{align}
        \mu_i\mu_j \abs*{\braket{\psi_i}{\psi_j}}^{2t} = \mu_i\mu_j - 2\Prob\sbra*{x_m = 1, I = i, J = j}.
    \end{align}
    Summing over the indices $i$ and $j$ gives
    \begin{align}
        \sum_{i,j} \mu_i\mu_j \abs*{\braket{\psi_i}{\psi_j}}^{2t}
        = \sum_{i,j}\mu_i\mu_j - 2\sum_{i,j}\Prob\sbra*{x_m = 1, I = i, J = j}.
    \end{align}
    By the definition of $\calF_t\rbra{\calE}$, we have
    \begin{align}
        \calF_t\rbra*{\calE} = 1 - 2\Prob\sbra{x_m=1}.
    \end{align}
    By Hoeffding's inequality~\cite{Hoe63}, we have
    \begin{align}
        \Prob\sbra*{\abs*{\calF_t\rbra*{\calE} - X} \geq \eps}
        = \Prob\sbra*{\abs*{\frac{1}{M} \sum_{m=1}^M x_m - \Prob\sbra*{x_m=1}} \geq \frac{\eps}{2}}
        \leq 2\exp\rbra*{-\frac{M\eps^2}{2}}
        = 2e^{-4}
        < \frac{1}{4}.
    \end{align}
    so $M = O\rbra{1/\eps^2}$ iterations suffice. Each iteration uses $t$ copies of each sampled state and hence $2t=O\rbra{t}$ samples in total. Therefore,~\cref{alg:state_frame_potential_sample} uses $O\rbra{t/\eps^2}$ samples from $\calE$.
\end{proof}

\subsection{Lower bound}\label{sec:sample_lower_bound}

We use the Helstrom--Holevo bound~\cite{Hel67,Hol73}.

\begin{lemma}[Helstrom-Holevo bound]\label{lem:helstrom-holevo}
    Let $\rho_0$ and $\rho_1$ be two quantum states, and suppose that an unknown state $\rho$ is chosen uniformly from $\cbra{\rho_0,\rho_1}$.
    For any POVM $M = \cbra{M_0,M_1}$, the success probability $p_{\mathrm{succ}}$ of distinguishing the two cases is bounded by
    \begin{align}
        p_{\mathrm{succ}}
        = \frac{1}{2}\tr\rbra*{M_0\rho_0} + \frac{1}{2}\tr\rbra*{M_1\rho_1}
        \leq \frac{1}{2}\rbra*{1+\frac{1}{2}\Abs*{\rho_0-\rho_1}_1}.
    \end{align}
\end{lemma}

We first prove the result for ensembles of cardinality two.

\begin{theorem}\label{thm:sample_lower_constant}
    Fix an integer $t$ and a parameter $\eps \in \interval[open]{0}{1}$.
    Any quantum algorithm for estimating to within additive error $\eps$ the state frame potential $\calF_t\rbra*{\calE}$ for any ensemble $\calE$ with cardinality $2$ requires $\Omega\rbra{t/\eps^2}$ samples from $\calE$. 
\end{theorem}

\begin{proof}
    We reuse the hard states $\ket{\psi_{\pm}}$ from the proof of~\cref{thm:query_lower_constant}, defined in~\cref{def:psi-pm}.
    Consider the problem of distinguishing $\ket{\psi_+}$ from $\ket{\psi_-}$.
    By~\cref{lem:helstrom-holevo}, distinguishing the two states requires sample complexity
    \begin{align}
        \Omega\rbra*{\frac{1}{1-\abs{\braket{\psi_+}{\psi_-}}}} = \Omega\rbra*{\frac{t}{\varepsilon^2}}. 
    \end{align}
    where $\rho_0 = \ketbra{\psi_+}{\psi_+}$ and $\rho_1 = \ketbra{\psi_-}{\psi_-}$.
    On the other hand, let $\mathcal{E}_{\pm}$ be the ensemble defined by \cref{eq:def-ensemble}.
    Note that \cref{eq:potential-diff} gives $\calF_t\rbra{\mathcal{E}_+} - \calF_t\rbra{\mathcal{E}_-} \geq \Omega\rbra{\varepsilon}$. 
    Consequently, an estimator with additive error $\Theta\rbra{\varepsilon}$ would distinguish $\ket{\psi_+}$ from $\ket{\psi_-}$. Hence, estimating the state frame potential requires $\Omega\rbra{t/\varepsilon^2}$ samples.
\end{proof}

We can also generalize our sample lower bound for any ensemble with an arbitrarily large cardinality.

\begin{theorem} \label{thm:sample_lower_arbitrary}
    Fix an integer $t$ and a parameter $\eps \in \interval[open]{0}{1}$.
    For any integer $K \geq 2$, any quantum algorithm for estimating to within additive error $\eps$ the state frame potential $\calF_t\rbra{\calE}$ of any ensemble $\calE$ with cardinality $K$ requires $\Omega\rbra{t/\eps^2}$ samples from $\calE$.
\end{theorem}

\begin{proof}
    The proof is analogous to that of~\cref{thm:sample_lower_constant}.
    Let $\calE_\pm$ be the ensemble defined by \cref{eq:genralized_ensemble} in \cref{thm:query_lower_arbitrary}.
    Note that $\calF_t\rbra{\calE_+} - \calF_t\rbra{\calE_-} \geq \Omega\rbra{\eps}$.
    Consequently, an estimator with additive error $\Theta\rbra{\eps}$ would distinguish $\ket{\psi_+}$ from $\ket{\psi_-}$. Hence, estimating the state frame potential requires $\Omega\rbra{t/\eps^2}$ samples.
\end{proof}

\section{Single-copy sample model} \label{S:single_copy_state_frame_potential}

Throughout this section, we make a slight abuse of notation.
In the main text, an ensemble is specified by weights and states, for example $\calE=\cbra{\rbra{\mu_i,\ket{\psi_i}}}$, and its frame potential is denoted by $\calF_t\rbra{\calE}$.
Here, for a fixed collection of states $\cbra{\ket{\psi_i}}$, we identify a collection of weights $\nu = \cbra{\nu_i}$ with the atomic probability measure on $\calP$ given by
\begin{equation}
    \nu=\sum_i \nu_i \delta_{\ket{\psi_i}},
\end{equation}
where $\delta_{\ket{\psi_i}}$ is the point mass at $\ket{\psi_i}$. We denote the corresponding frame potential by $\calF_t\rbra{\nu}$.
This convention, commonly used in the study of designs, changes only the notation: $\calF_t\rbra{\nu}$ equals the frame potential of the associated state ensemble, and none of the definitions or results are affected.
We retain the ensemble notation $\calF_t\rbra{\calE}$ in theorem and lemma statements and use the measure notation $\calF_t\rbra{\mu}$, $\calF_t\rbra{\hat{\nu}}$, and so forth within proofs.

\subsection{Upper bound}\label{sec:single_copy_sample_upper_bound}

In the single-copy sample model, one obtains only one state $\ket{\psi_i} \sim \mu$ per sample, together with the label $i$.
As a simple strategy for estimating the frame potential in this model, we introduce and analyze a \emph{store-and-estimate} algorithm.
The underlying idea is as follows:
\begin{enumerate}
    \item By repeatedly sampling from $\mu$, one stores sufficiently many copies of $\ket{\psi_i}$ for all $i$ such that $\mu_i \geq \muth$, where $\muth$ is a threshold probability that we choose. The subset of indices satisfying this threshold condition is called the \emph{heavy} set of $\mu$.
    \item Using the stored states in the heavy set, we perform generalized SWAP tests. If sufficiently many copies of every heavy state were stored in the first stage, these tests estimate the frame potential accurately.
\end{enumerate}

Based on this strategy, we design~\cref{alg:single_copy_state_frame_potential_sample}, where the threshold probability is set as
\begin{equation}
    \muth  = \min \cbra*{ 1, 2^{-\mathrm{H}_\alpha\rbra*{\mu}} \eta^{-1/\rbra*{\alpha-1}} },
\end{equation}
where $\eta$ is an algorithmic parameter and $\mathrm{H}_\alpha\rbra{\mu}$ is the R\'{e}nyi-$\alpha$ entropy of $\mu$. We assume that this entropy is known in advance.

The following theorem shows quantitatively that both the estimation error and the failure probability can be made arbitrarily small by choosing $N$ sufficiently large.

\begin{theorem}\label{Thm:one_shot_state_frame_potential}
    By setting the parameter $\eta$ of~\cref{alg:single_copy_state_frame_potential_sample} as $\eta = \eps/216$, it suffices to set the number $N$ of samples as
    \begin{align}\label{Eq:NumberOfSampling}
        N = O \rbra*{ \frac{t}{\eps^2}\log \frac{1}{\delta}  +\frac{2^{\mathrm{H}_{\alpha}\rbra*{\mu}}}{\eps^{2-\frac{1}{\alpha-1}}}
        \rbra*{ \mathrm{H}_\alpha\rbra*{\mu} + \frac{1}{\alpha-1}\log \frac{1}{\eps}  + \log \frac{1}{\delta} } },
    \end{align}
    to guarantee that the algorithm returns an estimate $\hat{\calF}$ satisfying $\abs{\hat{\calF} - \calF_t\rbra{\calE}} \leq \eps$ with probability at least $1-\delta$.
\end{theorem}

\cref{Thm:one_shot_state_frame_potential} follows by controlling three sources of failure or estimation error.

First, the algorithm may return $\FAIL$, which clearly increases the failure probability. The following lemma shows that this rarely occurs if $N$ is sufficiently large. 

\begin{algorithm}[htbp]
    \caption{$\mathsf{SingleCopyStateFramePotentialEstimator}\rbra{\calE,t,\eps,\delta,N,\mathrm{H}_\alpha\rbra{\mu}}$}
    \label{alg:single_copy_state_frame_potential_sample}
    \begin{algorithmic}[1]
        \REQUIRE An ensemble $\calE = \cbra{\rbra{\mu_i,\ket{\psi_i}}}_{i=0}^{K-1}$, order $t \in \N$, accuracy $\eps \in \interval[open]{0}{1}$, failure probability $\delta \in \interval[open]{0}{1}$, sample budget $N$, and the Rényi-$\alpha$ entropy $\mathrm{H}_\alpha\rbra{\mu}$ of $\mu$ for some $\alpha > 1$.
        \ENSURE An estimate $\hat{\calF}$ of $\calF_t\rbra{\calE}$ to within additive error $\varepsilon$.

        \STATE Sample $N$ i.i.d. states $\ket{\psi_{i_1}}, \ldots, \ket{\psi_{i_N}}$ from $\mu$ and store them.
        \STATE Store the ordered list $\Psi_N \coloneqq \rbra{\ket{\psi_{i_1}}, \ldots, \ket{\psi_{i_N}}}$.
        \FOR {$i=0, \ldots, K-1$}
            \STATE $\hat{N}_i \gets \abs{ \cbra{ k \in \cbra{ 1,\ldots, N } \colon i_k=i } }$.
        \ENDFOR
        \STATE $\eta \gets \frac{1}{216}\eps$.
        \STATE Set a threshold probability $\muth \gets \min\cbra{ 1, 2^{-\mathrm{H}_\alpha\rbra{\mu}} \eta^{-1/(\alpha-1)} }$.
        \STATE Define an empirical heavy set $\hat{\calI}_{\heavy} \gets \cbra{ i \colon \hat{N}_i \geq N \muth/2 }$.
        \FOR {$i \in \hat{\calI}_\heavy$}
            \STATE $\hat{a}_i \gets \lfloor \hat{N}_i/t\rfloor$.
            \STATE Form $\hat{a}_i$ blocks of $\ket{\psi_i}^{\otimes t}$ from $\Psi_N$.
        \ENDFOR
        \STATE $\hat{A}\rbra{ \hat{\calI}_\heavy } \gets \sum_{i \in \hat{\calI}_\heavy}\hat{a}_i$.
        \STATE Define a probability distribution $\hat{\nu}$ over $\hat{\calI}_\heavy$ by $\hat{\nu}_i \gets \hat{a}_i/\hat{A}\rbra{ \hat{\calI}_\heavy }$.

        \STATE Set the number of generalized SWAP tests to $M \gets \left\lceil \frac{8}{\eps^2}\ln \frac{8}{\delta}\right\rceil$.
        \STATE Sample i.i.d. labels $L_1, \ldots, L_{2M} \sim \hat{\nu}$.
        \IF{$\exists i\in \hat{\calI}_\heavy$ such that $\abs{ \cbra{ m \colon L_m =i } } > \hat{a}_i$}
            \STATE \textbf{return} $\FAIL$ due to a lack of blocks $\ket{\psi_i}^{\otimes t}$.
        \ELSE
            \FOR {$m=1,\ldots, M$}
                \STATE $\rho \gets \rbra{\ketbra{\psi_{L_{2m-1}}}{\psi_{L_{2m-1}}}}^{\otimes t}$ using an unused block.
                \STATE $\sigma \gets \rbra{\ketbra{\psi_{L_{2m}}}{\psi_{L_{2m}}}}^{\otimes t}$ using an unused block.
                \STATE $x_m \gets \mathsf{SwapTest}\rbra{ H\ket{0}, \rho, \sigma }$.
            \ENDFOR
        \STATE \textbf{return} $\hat{\calF} = 1-\frac{2}{M} \sum_{m=1}^M x_m$.
        \ENDIF
    \end{algorithmic}
\end{algorithm}

\begin{restatable}{lemma}{LemmaFAIL} \label[lemma]{Lemma:FAIL}
    If $N$ satisfies
    \begin{equation}
        N \geq \max \cbra*{ \frac{2t}{\muth} \rbra*{ 6\ln \frac{2}{\delta_\FAIL \muth} + 1 },  2tM + \frac{2t}{\muth} },
    \end{equation}
    then~\cref{alg:single_copy_state_frame_potential_sample} returns $\FAIL$ with probability at most $\delta_\FAIL$.
\end{restatable}

The value of $M$ in~\cref{Lemma:FAIL} need not equal $\eps^{-2}\log\delta^{-1}$, as in the algorithm. The lemma allows arbitrary $M$.

Second, we show that, when the algorithm returns an estimator $\hat{\calF}$, it is close to $\calF_t\rbra{\hat{\nu}}$, assuming that $N$ is accordingly large.
Here, $\hat{\nu}$ is determined in~\cref{alg:single_copy_state_frame_potential_sample} as an empirical probability distribution over the empirical heavy set $\hat{\calI}_\heavy$, which can be naturally extended to a probability distribution over $\cbra{\ket{\psi_i} }_{i=0}^{K-1}$. 

\begin{restatable}{lemma}{LemmaEstimate} \label[lemma]{Lemma:Estimate}
    If $M$ satisfies
    \begin{equation}
        M \geq \frac{2}{\epsest^2} \ln \frac{2}{\delest}, \label{Eq:ESTIMATE}
    \end{equation}
    then the probability that the algorithm returns an estimate $\hat{\calF}$ satisfying $\abs{ \hat{\calF} - \calF_t(\hat{\nu}) } > \epsest$ is at most $\delest$.
\end{restatable}

Finally, again assuming sufficiently large $N$, it holds that $\calF_t\rbra{\hat{\nu}} \approx \calF_t\rbra{\calE}$.

\begin{restatable}{lemma}{LemmaBias} \label[lemma]{Lemma:Bias}
    If the parameter $\eta$ is set to $\epsFP/108$ and if $N$ satisfies
    \begin{equation}\label{Eq:BIAS}
        N \geq \frac{1728}{\epsFP^2 \muth} \ln \frac{6}{\delFP \muth},
    \end{equation}
    then $\abs{ \calF_t\rbra{\calE} - \calF_t\rbra{\hat{\nu}} } \leq \epsFP$ with probability at least $1- \delFP$.
\end{restatable}

Proofs of these lemmas are given in~\cref{SS:LemmasFAILEstimateBias}. Together they imply~\cref{Thm:one_shot_state_frame_potential}.

\begin{proof}[Proof of~\cref{Thm:one_shot_state_frame_potential}]
    By setting $\delta_\FAIL = \delta/2$ in~\cref{Lemma:FAIL}, we obtain that, if $N$ satisfies
    \begin{equation}\label{Eq:FAIL_order}
        N = \Omega\rbra*{ \max\cbra*{ \frac{t}{\muth} \log \frac{1}{\delta \muth},  tM + \frac{t}{\muth} } },
    \end{equation}
    then~\cref{alg:single_copy_state_frame_potential_sample} does not return $\FAIL$ with probability at least $1-\delta/2$.
    
    By~\cref{Lemma:Estimate,Lemma:Bias} and the union bound, if $M$ and $N$ satisfy
    \begin{align}
        & M = \Omega\rbra*{ \frac{1}{\eps^2} \log \frac{1}{\delta} }, \label{Eq:estimate_order}\\
        & N = \Omega\rbra*{ \max\cbra*{ \frac{1}{\eps^2 \muth} \log \frac{1}{\delta \muth},\ \frac{1}{\eps} \log \frac{1}{\delta} } }, \label{Eq:bias_order}
    \end{align}
    then
    \begin{align}
        \abs*{ \hat{\calF} -  \calF_t\rbra*{\calE} }
        \leq \abs*{ \calF_t\rbra*{\calE} - \calF_t\rbra*{\hat{\nu}} }  + \abs*{ \hat{\calF} - \calF_t\rbra*{\hat{\nu}} }
        \leq \eps,
    \end{align}
    whenever an estimate is returned, except with probability at most $\delta/2$. Here, we have set $\epsest = \epsFP = \eps/2$ and $\delest = \delFP = \delta/4$. Note that we have used the fact that $\eta = \eps/216$.

    Altogether, if~\cref{Eq:FAIL_order,Eq:estimate_order,Eq:bias_order} simultaneously hold, then the probability that~\cref{alg:single_copy_state_frame_potential_sample} returns $\hat{\calF}$ and $\hat{\calF}$ satisfies $\abs{\hat{\calF} - \calF_t\rbra{\calE} } \leq \eps$ is at least $1 - \delta$. The dominant conditions on $N$ can be summarized as
    \begin{equation}
            N = \Omega\rbra*{ \max\cbra*{\frac{1}{\eps^2 \muth} \log \frac{1}{\delta \muth},\ \frac{t}{\eps^2}\log \frac{1}{\delta} } }.
    \end{equation}
    Note that $\muth \leq 1$. Substituting $\muth^{-1} = \Theta\rbra{ 2^{\mathrm{H}_\alpha\rbra{\mu}} \eps^{1/\rbra{\alpha-1}} }$, we obtain that it suffices to choose $N$ as
    \begin{align}
        N = O \rbra*{ \frac{t}{\eps^2}\log \frac{1}{\delta}  + \frac{2^{\mathrm{H}_{\alpha}\rbra*{\mu}}}{\eps^{2-\frac{1}{\alpha-1}}}
        \rbra*{ \mathrm{H}_\alpha\rbra*{\mu} + \frac{1}{\alpha-1}\log \frac{1}{\eps}  + \log \frac{1}{\delta} } },
    \end{align}
    which is nothing but~\cref{Eq:NumberOfSampling}.
\end{proof}

\subsection{\texorpdfstring{Proof of~\cref{Lemma:FAIL,Lemma:Estimate,Lemma:Bias}}{Proof of Lemmas}}\label{SS:LemmasFAILEstimateBias}

To prove~\cref{Lemma:FAIL,Lemma:Estimate,Lemma:Bias}, we first establish several properties of the heavy set of $\mu$.

\subsubsection{\texorpdfstring{Properties of the heavy set of $\mu$}{Properties of the heavy set of mu}}\label{SSS:PropertiesHeavySet}

In what follows, we denote by $\calI_\heavy$ the true heavy set of $\mu$.
That is, 
\begin{align}
    &\calI_\heavy \coloneqq \cbra*{i \in \cbra*{0, \ldots, K-1} \colon \mu_i \geq \muth }.
\end{align}
Note that this is in general different from the empirical heavy set $\hat{\calI}_\heavy$.
We also use the weight on $\calI_\heavy$,
\begin{align}
    \mu(\calI_\heavy) \coloneqq \sum_{i \in \calI_\heavy} \mu_i,
\end{align}
and the numbers of samples that belong to $\hat{\calI}_\heavy$ and $\calI_\heavy$: 
\begin{align}
    \hat{N}(\hat{\calI}_\heavy) \coloneqq \sum_{i \in \hat{\calI}_\heavy} \hat{N}_i, \text{\ \ and\ \ }
    \hat{N}(\calI_\heavy) \coloneqq \sum_{i \in \calI_\heavy} \hat{N}_i.
\end{align}

An important property of $\calI_\heavy$ is that it contains most of the probability mass of $\mu$, with the omitted mass controlled by the algorithmic parameter $\eta$.

\begin{lemma}\label[lemma]{Lemma:p_heavy_large}
    The probability $\mu\rbra{\calI_\heavy}$ of the heavy set of $\mu$ satisfies $\mu\rbra{\calI_\heavy} \geq 1 - \eta$.
\end{lemma}

\begin{proof}[Proof of~\cref{Lemma:p_heavy_large}]
    It suffices to show that $\mu_\tail \coloneqq 1- \mu\rbra{\calI_\heavy} < \eta$. As $\mu_i \geq \muth$ for any $i \in \calI_\heavy$, $\mu_i < \muth$ for any $i \notin \calI_\heavy$. Using this, it follows that, for all $i \notin \calI_\heavy$,
    \begin{equation}
        \mu_i = \mu_i^{\alpha} \mu_i^{1 - \alpha} < \mu_i^{\alpha} \muth^{1-\alpha}.
    \end{equation}
    By taking the summation over all $i \notin \calI_\heavy$, we have
    \begin{equation}
        \mu_\tail < \muth^{1-\alpha} \sum_{i \notin \calI_\heavy} \mu_i^{\alpha} \leq \muth^{1-\alpha} \sum_{i =0}^{K-1} \mu_i^{\alpha}.
    \end{equation}
    As $\sum_i \mu_i^{\alpha} = 2^{\rbra{1-\alpha} \mathrm{H}_{\alpha}\rbra{\mu}}$ and $\muth = \min\cbra{ 1, 2^{-\mathrm{H}_\alpha\rbra{\mu}} \eta^{-1/\rbra{\alpha-1}} }$, we obtain $\mu_\tail < \eta$, completing the proof.   
\end{proof}

We can also bound the cardinalities of $\calI_\heavy$ and $\hat{\calI}_\heavy$ in terms of $\muth$.

\begin{lemma} \label[lemma]{Lemma:Size_calI}
    The cardinalities of $\calI_\heavy$ and $\hat{\calI}_\heavy$ satisfy
    \begin{align}
        \abs*{ \calI_\heavy } \leq 1/\muth, \text{\ \ and\ \ } \abs*{ \hat{\calI}_\heavy } \leq 2/\muth.
    \end{align}
\end{lemma}

\begin{proof}[Proof of~\cref{Lemma:Size_calI}]
    For any $i \in \calI_\heavy$, $\mu_i \geq \muth$. We hence obtain
    \begin{equation}
        1 = \sum_{i=0}^{K-1}  \mu_i \geq \sum_{i \in \calI_\heavy} \mu_i \geq \muth \abs*{ \calI_\heavy },
    \end{equation}
    which is the first equation.
    
    The second equation is similarly obtained using the fact that $\hat{N}_i \geq N \muth/2$ for $i \in \hat{\calI}_\heavy$. By taking the summation over all $i \in \hat{\calI}_\heavy$, we have
    \begin{equation}
        \hat{N}\rbra*{ \hat{\calI}_\heavy } \geq  \frac{N\muth}{2} \abs*{ \hat{\calI}_\heavy } \geq \frac{\hat{N}\rbra*{ \hat{\calI}_\heavy } \muth}{2} \abs*{ \hat{\calI}_\heavy }.
    \end{equation}
    This implies the second equation.
\end{proof}

From \cref{Lemma:Size_calI}, one may expect that $\hat{\calI}_\heavy$ is a larger set than $\calI_\heavy$. This is indeed the case with high probability if $N$ is sufficiently large. 

\begin{lemma} \label[lemma]{Lemma:calI_inclusion}
    Let $\delta_\heavy \in \interval[open]{0}{1}$. If $N \geq \frac{8}{\muth} \ln \frac{1}{\delta_\heavy \muth}$, then $\calI_\heavy \subseteq \hat{\calI}_\heavy$ with probability at least $1 - \delta_\heavy$.
\end{lemma}

\begin{proof}[Proof of~\cref{Lemma:calI_inclusion}]
    Clearly, $\Ex\sbra{\hat{N}_i} = \mu_i N$. As $\mu_i \geq \muth$ for $i \in \calI_\heavy$, it holds that $\Ex\sbra{\hat{N}_i} \geq N \muth$. Using the multiplicative Chernoff bound, we obtain that, for $i \in \calI_\heavy$,
    \begin{align}
        \Prob\sbra*{ \hat{N}_i < \frac{N\muth}{2} } &\leq \Prob\sbra*{ \hat{N}_i < \frac{\Ex\sbra*{\hat{N}_i}}{2} }\\
        &\leq e^{-\Ex\sbra*{\hat{N}_i}/8}\\
        &\leq e^{-N \muth/8}.
    \end{align}
    By the union bound, we obtain
    \begin{align}
        \Prob\sbra*{ \exists i \in \calI_\heavy, \hat{N}_i < \frac{N\muth}{2} }
        & \leq \abs*{ \calI_\heavy } e^{-N \muth/8}\\
        & \leq \frac{1}{\muth} e^{-N \muth/8},\\
        & \leq \delta_\heavy,
    \end{align}
    where we have used~\cref{Lemma:Size_calI} in the second line, and the assumption on $N$ in the last line.
    
    As $\hat{\calI}_\heavy$ is a set of $i$ such that $\hat{N}_i \geq \frac{N\muth}{2}$, this implies that $\calI_\heavy \subseteq \hat{\calI}_\heavy$ with probability at least $1 - \delta_\heavy$, completing the proof.
\end{proof}

Note that~\cref{alg:single_copy_state_frame_potential_sample} defines $\hat{\calI}_\heavy$ using the threshold $\hat{N}_i \geq N \muth/2$, rather than $\hat{N}_i \geq N \muth$. This factor of $1/2$ ensures that~\cref{Lemma:calI_inclusion} holds.

\subsubsection{\texorpdfstring{Proof of~\cref{Lemma:FAIL}}{Proof of Lemma}} \label{SSS:ProofOfLemmaFAIL}

We now show~\cref{Lemma:FAIL}, which we restate below for clarity.

\LemmaFAIL*

We begin with the fact that the desired statement holds with high probability if both $\hat{A}(\hat{\calI}_\heavy)$ and $N$ are sufficiently large.

\begin{lemma}\label[lemma]{Lemma:preFAIL}
    Let $\delta_1 \in \interval[open]{0}{1}$. If $\hat{A}\rbra{\hat{\calI}_\heavy} \geq 4M$ and 
    \begin{equation}
        N \geq \frac{2t}{\muth} \rbra*{ 6\ln \frac{2}{\delta_1 \muth} + 1 },
    \end{equation} 
    then the probability that~\cref{alg:single_copy_state_frame_potential_sample} returns $\FAIL$ is at most $\delta_1$.
\end{lemma}

\begin{proof}[Proof of~\cref{Lemma:preFAIL}]
    Let $\hat{\boldsymbol{L}}=\rbra{L_1, \ldots, L_{2M}}$ be the sequence of i.i.d. labels sampled from $\hat{\nu}$ by the algorithm.
    Let $\hat{n}_i$ denote the number of occurrences of label $i$ in this sequence. The multiplicative Chernoff bound gives, for any $i \in \hat{\calI}_\heavy$,
    \begin{align}
        \Prob\sbra*{ \hat{n}_i > \hat{a}_i }
        & = \Prob\sbra*{ \hat{n}_i > \frac{\hat{a}_i}{\Ex\sbra*{\hat{n}_i}} \Ex\sbra*{\hat{n}_i} } \\
        & \leq \exp\rbra*{ - \frac{\hat{a}_i/\Ex\sbra*{\hat{n}_i} - 1}{3} \Ex\sbra*{\hat{n}_i} } \\
        & \leq \exp\rbra*{ - \frac{\hat{a}_i - \Ex\sbra*{\hat{n}_i}}{3} }.
    \end{align}
    
    It also holds that 
    \begin{equation}
        \Ex\sbra*{\hat{n}_i}
        = 2M \hat{\nu}_i = 2M \frac{\hat{a}_i}{\hat{A}\rbra*{\hat{\calI}_\heavy}}
        \leq \frac{\hat{a}_i}{2},
    \end{equation}
    where we have used the assumption that $\hat{A}\rbra{\hat{\calI}_\heavy} \geq 4M$. Using this, we obtain
    \begin{align}
        \Prob\sbra*{ \hat{n}_i > \hat{a}_i } 
        \leq  e^{- \hat{a}_i/6} 
        \leq  e^{- \hat{a}_{\min}/6},
    \end{align}
    where $\hat{a}_{\min} \coloneqq \min_{i \in \hat{\calI}_\heavy} \hat{a}_i$.
    Since the algorithm returns $\FAIL$ if there exists $i \in \hat{\calI}_\heavy$ such that $\hat{n}_i > \hat{a}_i$, the union bound leads to
    \begin{align}
        \Prob\sbra*{ \text{\cref{alg:single_copy_state_frame_potential_sample} returns $\FAIL$} }
        & \leq \abs*{ \hat{\calI}_\heavy } e^{- \hat{a}_{\min}/6}\\
        & \leq \frac{2}{\muth} e^{- \hat{a}_{\min}/6},
    \end{align}
    where we used~\cref{Lemma:Size_calI}.
    
    Finally, since $\hat{a}_i = \floor{ \hat{N}_i/t } \geq \floor{ N \muth/\rbra{2t} }$ for every $i \in \hat{\calI}_\heavy$,
    \begin{equation}
        \hat{a}_{\min} \geq \floor*{ \frac{N \muth}{2t} } \geq \frac{N \muth}{2t}-1,
    \end{equation}
    leading to
    \begin{align}
        \Prob\sbra*{ \text{\cref{alg:single_copy_state_frame_potential_sample} returns $\FAIL$} }
        \leq \frac{2}{\muth} e^{- \rbra*{ \frac{N \muth}{2t}-1 }/6}.
    \end{align}
    This is at most $\delta_1$ due to the assumption on $N$.
\end{proof}

We next show that the remaining assumption of~\cref{Lemma:preFAIL}, $\hat{A}\rbra{\hat{\calI}_\heavy }\geq 4M$, holds with high probability when $N$ is sufficiently large.

\begin{lemma} \label[lemma]{Lemma:FAIL_supp}
    Let $\delta_2 \in \interval[open]{0}{1}$. If $N$ satisfies
    \begin{equation}\label{Eq:FAIL_supp_conditions}
        N \geq 8 \max\cbra*{ \frac{1}{\muth} \ln \frac{2}{\delta_2\muth},\ 2 \ln \frac{2}{\delta_2}, \ 2tM + \frac{2t}{\muth} }, 
    \end{equation}
    then $\hat{A}\rbra{ \hat{\calI}_\heavy } \geq 4M$ with probability at least $1 - \delta_2$.
\end{lemma}

\begin{proof}[Proof of~\cref{Lemma:FAIL_supp}]
    We first show that if $N \geq 16\ln \frac{2}{\delta_2}$, which is one of the conditions in~\cref{Eq:FAIL_supp_conditions}, then $\hat{N}\rbra{ \calI_\heavy } \geq N/4$ with probability at least $1 - \delta_2/2$.
    This simply follows from the multiplicative Chernoff bound. As $\Ex\sbra{\hat{N}\rbra{\calI_\heavy}} = N \mu\rbra{\calI_\heavy}$, we have
    \begin{align}
        \Prob \sbra*{ \hat{N}\rbra*{ \calI_\heavy } < \frac{N}{4} }
        & \leq \exp\rbra*{ - \frac{1}{2}\rbra*{ 1 - \frac{1}{4 \mu\rbra*{ \calI_\heavy }} }^2 \Ex\sbra*{ \hat{N}\rbra*{ \calI_\heavy } } }\\
        & = \exp\sbra*{ - \frac{1}{2}\rbra*{ 1 - \frac{1}{4 \mu\rbra*{ \calI_\heavy }} }^2 \mu\rbra*{ \calI_\heavy } N }.
    \end{align}
    By~\cref{Lemma:p_heavy_large}, $\mu\rbra{ \calI_\heavy } \geq 1 - \eta$. Since $\eta \leq 1/2$, we have $\mu\rbra{\calI_\heavy} \geq 1/2$.
    Hence, it holds that
    \begin{equation}
        \Prob\sbra*{ \hat{N}\rbra*{ \calI_\heavy } < \frac{N}{4} } \leq e^{-N/16}.
    \end{equation} 
    Hence, if $N \geq 16\ln \frac{2}{\delta_2}$, then $\hat{N}\rbra{ \calI_\heavy } \geq N/4$ with probability at least $1 - \delta_2/2$.
        
    Moreover, using the condition in~\cref{Eq:FAIL_supp_conditions} that
    \begin{equation}
        N \geq \frac{8}{\muth} \ln \frac{2}{\delta_2\muth},
    \end{equation}
    \cref{Lemma:calI_inclusion} implies that $\calI_\heavy \subseteq \hat{\calI}_\heavy$ with probability at least $1 - \delta_2/2$.
    
    Combining these two events with the union bound, we conclude that $\calI_\heavy \subseteq \hat{\calI}_\heavy$ and $\hat{N}\rbra{\calI_\heavy} \geq N/4$ simultaneously with probability at least $1 - \delta_2$.
    When this is the case, $\hat{N}\rbra{ \hat{\calI}_\heavy } \geq \hat{N}\rbra{ \calI_\heavy } \geq N/4$, and we obtain 
    \begin{align}
        \hat{A}\rbra*{ \hat{\calI}_\heavy } &= \sum_{i \in \hat{\calI}_\heavy} \floor*{ \frac{\hat{N}_i}{t} } \\
        & \geq \sum_{i \in \hat{\calI}_\heavy} \rbra*{ \frac{\hat{N}_i}{t} -1 }\\
        & = \frac{\hat{N}(\hat{\calI}_\heavy)}{t} - \abs*{ \hat{\calI}_\heavy }\\
        & \geq \frac{N}{4t} - \frac{2}{\muth},
    \end{align}
    where we have used~\cref{Lemma:Size_calI} in the last line. Using the last assumption that $N \geq 16tM+8t/\muth$, the desired statement is obtained.
\end{proof}

\cref{Lemma:FAIL} follows from~\cref{Lemma:preFAIL,Lemma:FAIL_supp}.

\begin{proof}[Proof of~\cref{Lemma:FAIL}]
    By setting both $\delta_1$ in~\cref{Lemma:preFAIL} and $\delta_2$ in~\cref{Lemma:FAIL_supp} to $\delta_\FAIL/2$, the following statement follows. If
    \begin{equation}
         N \geq \frac{2t}{\muth} \rbra*{ 6\ln \frac{4}{\delta \muth} + 1 },
    \end{equation}
    and
    \begin{equation}
        N \geq 8 \max\cbra*{ \frac{1}{\muth} \ln \frac{4}{\delta\muth},\ 2 \ln \frac{4}{\delta},\ 2tM + \frac{2t}{\muth} },
    \end{equation}
    then the probability that~\cref{alg:single_copy_state_frame_potential_sample} returns $\FAIL$ is at most $\delta_\FAIL$.
    It is straightforward to observe that these conditions on $N$ are simplified to
    \begin{equation}
         N \geq \max\cbra*{ \frac{2t}{\muth} \rbra*{ 6\ln \frac{4}{\delta \muth} + 1 },\  2tM + \frac{2t}{\muth} },
    \end{equation}
    which holds by assumption.
\end{proof}

\subsubsection{\texorpdfstring{Proof of~\cref{Lemma:Estimate}}{Proof of Lemma}}

\cref{Lemma:Estimate} can be shown in a standard manner.

\LemmaEstimate*

\begin{proof}[Proof of~\cref{Lemma:Estimate}]
    Consider the ideal experiment that performs all $M$ generalized SWAP tests for the i.i.d. labels drawn from $\hat\nu$, irrespective of whether the stored blocks are available. Its estimator satisfies $\Ex\sbra{\hat{\calF}}=\calF_t\rbra{\hat\nu}$. Hoeffding's inequality gives
    \begin{equation}
        \Prob\sbra*{ \abs*{ \hat{\calF} - \calF_t\rbra*{ \hat{\nu} } } > \epsest } \leq 2 e^{-M \epsest^2/2}.
    \end{equation}
    This is at most $\delest$ under~\cref{Eq:ESTIMATE}. Whenever the actual algorithm returns an estimate, it agrees with the ideal estimator on the same labels and measurement outcomes. Therefore, the event that the actual algorithm returns a bad estimate is contained in the ideal experiment's bad-estimate event, proving the claim.
\end{proof}

\subsubsection{\texorpdfstring{Proof of~\cref{Lemma:Bias}}{Proof of Lemma}}

We finally consider~\cref{Lemma:Bias}. The following is a restatement for clarity.

\LemmaBias*

The proof is based on the fact that the difference of the frame potentials can be bounded from above by the $\ell_1$-distance of the probability distributions.

\begin{lemma}\label[lemma]{Lemma:Frame_Potential_Trace_Norm}
    Let $\xi$ and $\zeta$ be discrete probability measures on $\calP$. Then,
    \begin{equation}
        \abs*{ \calF_t\rbra*{\xi} - \calF_t\rbra*{\zeta} }  \leq 2 \Abs*{ \xi-\zeta }_1,
    \end{equation}
    where $\Abs*{ \xi-\zeta }_1$ denotes the $\ell_1$-distance between the two probability distributions.
\end{lemma}

\begin{proof}[Proof of~\cref{Lemma:Frame_Potential_Trace_Norm}]
    Write $\xi = \cbra{\rbra{\xi_i, \ket{\psi_i}}}_{i=0}^{K-1}$ and $\zeta = \cbra{\rbra{\zeta_i, \ket{\psi_i}}}_{i=0}^{K-1}$, allowing zero weights when their supports differ. Then
    \begin{align}
        \abs*{ \calF_t\rbra*{ \xi } - \calF_t \rbra*{ \zeta } } 
        & = \abs*{ \sum_{i,j=0}^{K-1} \rbra*{ \xi_i \xi_j - \zeta_i \zeta_j } \abs*{ \braket{\psi_i}{\psi_j} }^{2t} }\\
        & \le \sum_{i,j=0}^{K-1} \abs*{ \xi_i \xi_j - \zeta_i \zeta_j }\\
        & \le \sum_{i,j=0}^{K-1} \abs*{ \xi_i  \abs*{ \xi_j - \zeta_j } + \zeta_j  \abs*{ \xi_i - \zeta_i } }\\
        & \le 2 \Abs*{ \xi - \zeta }_1,
    \end{align}
    where we have used the triangle inequality in the second and the third lines, and the fact that $\abs{ \braket{\psi_i}{\psi_j} } \leq 1$ in the second line.
    This completes the proof.
\end{proof}

Below, we bound the $\ell_1$-distance between $\mu$ and $\hat{\nu}$.
To this end, we introduce a couple of distributions.
    By restricting $\mu$ to the heavy set, define the unnormalized measure $\restrictedmu$ by
\begin{align}
        \rbra{\restrictedmu}_i \coloneqq
    \begin{cases}
        \mu_i & \text{if $i \in \calI_\heavy$,}\\
        0 & \text{otherwise,}
    \end{cases}
\end{align}
and its normalized probability distribution $\barrestrictedmu$ by
\begin{align}
    &\barrestrictedmu \coloneqq \frac{\restrictedmu}{\mu\rbra*{ \calI_\heavy }}.
\end{align}

    Similarly, define the normalized restriction $\barrestrictednu$ of $\hat{\nu}$ to $\calI_\heavy$ by
\begin{equation}
    \barrestrictednu \coloneqq
    \begin{cases}
        \hat{a}_i/\hat{A}\rbra*{ \calI_\heavy } & \text{if $i \in \calI_\heavy$,}\\
        0 & \text{otherwise,}
    \end{cases}
\end{equation}
where $\hat{A}\rbra{\calI_\heavy} \coloneqq \sum_{i \in \calI_\heavy}\hat{a}_i$.
As the support of $\hat{\nu}$ is $\hat{\calI}_\heavy$, this restriction makes sense if $\calI_\heavy \subseteq \hat{\calI}_\heavy$.
This is guaranteed by~\cref{Lemma:calI_inclusion} when $N$ is sufficiently large, which is satisfied if the assumption of~\cref{Lemma:Bias} is met.

Finally, we introduce a probability distribution $\tilde{\mu}$ by
\begin{equation}
    \tilde{\mu} = \cbra*{ \tilde{\mu}_i \coloneqq \frac{\hat{N}_i}{\hat{N}\rbra*{ \calI_\heavy }} \colon i \in \calI_\heavy }. 
\end{equation}

For these distributions, the following lemmas hold. The proofs are given later.

\begin{restatable}{lemma}{EllDistanceA}\label[lemma]{Lemma:Ell1_Distance_1}
    Let $\eps_1, \delta_1 \in \interval[open]{0}{1}$. If $N$ satisfies
    \begin{equation}\label{Eq:Cond_N_Ell1_Distance_1}
        N \geq \frac{48}{\eps_1^2 \muth} \ln \frac{2}{\delta_1 \muth},
    \end{equation}
    then $\Abs{ \barrestrictedmu - \tilde{\mu} }_1 \leq \eps_1$ with probability at least $1-\delta_1$.
\end{restatable}

\begin{restatable}{lemma}{EllDistanceB}\label[lemma]{Lemma:Ell1_Distance_2}
    Let $\eps_2, \delta_2 \in \interval[open]{0}{1}$. If $N$ satisfies     
    \begin{equation}
        N \geq \frac{8}{1-\eta}\max\cbra*{ \ln \frac{1}{\delta_2}, \ \frac{t}{\eps_2 \muth} },
    \end{equation}
    then $\Abs{ \tilde{\mu} - \barrestrictednu }_1 \leq \eps_2$ with probability at least $1-\delta_2$.
\end{restatable}

\begin{restatable}{lemma}{EllDistanceC}\label[lemma]{Lemma:Ell1_Distance_3}
    Let $\delta_3 \in (0,1)$. If $\calI_\heavy \subseteq \hat{\calI}_\heavy$ and if $N$ satisfies
    \begin{equation}
        N \geq \max\cbra*{ \frac{3}{\eta}, \ \frac{8}{1-\eta} } \ln \frac{2}{\delta_3},
        \quad \text{and} \quad
        N \geq \frac{8t}{\rbra*{ 1-\eta }\muth},
    \end{equation}
    then $\Abs{ \barrestrictednu - \hat{\nu} }_1 \leq \frac{32\eta}{3\rbra{1-\eta}}$ with probability at least $1-\delta_3$.
\end{restatable}

We now prove~\cref{Lemma:Bias} based on these lemmas. 

\begin{proof}[Proof of~\cref{Lemma:Bias}]
    Due to~\cref{Lemma:Frame_Potential_Trace_Norm}, we have
    \begin{equation}\label{Eq:Frame_Pot_Ell1_Distance_Eval}
        \abs*{ \calF_t\rbra*{ \mu } - \calF_t\rbra*{ \hat{\nu} } } \leq 2 \Abs*{ \mu - \hat{\nu} }_1.
    \end{equation}
    Using the triangle inequality, it follows that
    \begin{equation}\label{Eq:Ell1_triangle}
        \Abs*{ \mu - \hat{\nu} }_1 \leq \Abs*{ \mu - \barrestrictedmu }_1 + \Abs*{ \barrestrictedmu - \hat{\nu} }_1.
    \end{equation}
    
    The first term on the right-hand side of~\cref{Eq:Ell1_triangle} is bounded above by $2\eta$:
    \begin{align}
        \Abs*{ \mu - \barrestrictedmu }_1
        & \leq \Abs*{ \mu - \restrictedmu }_1 + \Abs*{ \restrictedmu - \barrestrictedmu }_1 \\
        & = 1 - \mu\rbra*{ \calI_\heavy }  + \rbra*{ 1 - \mu\rbra*{ \calI_\heavy } } \Abs*{ \barrestrictedmu }_1\\
        & \leq 2 \rbra*{ 1 - \mu\rbra*{ \calI_\heavy } }\\
        & \leq 2 \eta, \label{Eq:Ell1_Distance_Eval_1}
    \end{align}
    The equality uses $\restrictedmu = \mu\rbra{\calI_\heavy}\barrestrictedmu$, and the final inequality uses~\cref{Lemma:p_heavy_large}.
    
    To evaluate the second term on the right-hand side of~\cref{Eq:Ell1_triangle}, we again use the triangle inequality, leading to
    \begin{align}
        \Abs*{ \barrestrictedmu - \hat{\nu} }_1
        \leq \Abs*{ \barrestrictedmu - \tilde{\mu} }_1 + \Abs*{ \tilde{\mu} - \barrestrictednu }_1 + \Abs*{ \barrestrictednu - \hat{\nu} }_1.
    \end{align}
    Each term on the right-hand side is bounded from above using~\cref{Lemma:Ell1_Distance_1,Lemma:Ell1_Distance_2,Lemma:Ell1_Distance_3}.
    Setting the parameters in the lemmas to $\eps_1=\eps_2=\epsFP/6$ and $\delta_1=\delta_2=\delta_3=\delFP/3$, and applying the union bound, gives
    \begin{equation}\label{Eq:Ell1_Distance_Eval_2}
        \Abs*{ \barrestrictedmu - \hat{\nu} }_1 \leq \frac{\epsFP}{3} + \frac{32\eta}{3\rbra*{1-\eta}},
    \end{equation}
    with probability at least $1 - \delFP$, if $N$ satisfies all of the following conditions:
    \begin{align}
        N & \geq \frac{1728}{\epsFP^2 \muth} \ln \frac{6}{\delFP \muth},\label{Eq:N_dominant_Trace_Eval}\\
        N & \geq \frac{8}{1-\eta}\max\cbra*{ \ln \frac{3}{\delFP}, \ \frac{6t}{\epsFP\muth} }, \label{Eq:N_rounding_bias}\\
        N & \geq \max\cbra*{ \frac{3}{\eta}, \ \frac{8}{1-\eta} } \ln \frac{6}{\delFP},\label{Eq:N_tail_bias}
    \end{align}
    and if $\calI_\heavy \subseteq \hat{\calI}_\heavy$.
        
    By assumption,~\cref{Eq:N_dominant_Trace_Eval} holds.
    As~\cref{Eq:N_dominant_Trace_Eval} is stronger than the assumption in~\cref{Lemma:calI_inclusion}, $\calI_\heavy \subseteq \hat{\calI}_\heavy$ also holds.  
    Moreover, since the algorithm sets $\eta = \epsFP/108 \leq 1/108$, the condition in~\cref{Eq:N_dominant_Trace_Eval} further implies~\cref{Eq:N_rounding_bias,Eq:N_tail_bias}.
    Hence, all the above assumptions hold.
    In addition, since $\epsFP \leq 1$, the second branch of~\cref{Eq:N_rounding_bias} gives $N \geq \frac{48t}{\rbra{1-\eta}\epsFP\muth} \geq \frac{8t}{\rbra{1-\eta}\muth}$, so the additional condition of~\cref{Lemma:Ell1_Distance_3} is also met.
    
    Substituting~\cref{Eq:Ell1_Distance_Eval_1,Eq:Ell1_Distance_Eval_2} into~\cref{Eq:Ell1_triangle}, we obtain, with probability at least $1-\delFP$,
    \begin{equation}
        \Abs*{ \mu - \hat{\nu} }_1 \leq \frac{\epsFP}{3} + 2\eta + \frac{32\eta}{3\rbra*{ 1 - \eta }}.
    \end{equation}
    This, together with~\cref{Eq:Frame_Pot_Ell1_Distance_Eval}, leads to
    \begin{equation}
        \abs*{ \calF_t\rbra*{ \mu } - \calF_t\rbra*{ \hat{\nu} } } \leq \frac{2}{3}\epsFP + 4\eta + \frac{64\eta}{3\rbra*{ 1-\eta }}.
    \end{equation}
    As $\eta = \epsFP/108$ and $\epsFP \leq 1$, we have $1-\eta \geq 107/108$, so $4\eta + \frac{64\eta}{3\rbra{1-\eta}} \leq \frac{\epsFP}{27} + \frac{64\epsFP}{321} \leq \frac{\epsFP}{3}$, which implies $\abs{ \calF_t\rbra{\mu} - \calF_t\rbra{\hat{\nu}} } \leq \epsFP$ as desired.
\end{proof}

\subsubsection{Proof of \texorpdfstring{\cref{Lemma:Ell1_Distance_1}}{Lemma}}

\EllDistanceA*

\begin{proof}[Proof of~\cref{Lemma:Ell1_Distance_1}]
    For any $i \in \calI_\heavy$, $\Ex\sbra{\hat{N}_i} = \mu_i N$. Using the multiplicative Chernoff bound, we obtain that, for any $i \in \calI_\heavy$,
    \begin{align}
        \Prob\sbra*{ \abs*{ \hat{N}_i - N\mu_i } \geq \frac{\eps_1}{4} N \mu_i }
        & \leq 2 \exp\rbra*{ - \frac{\eps_1^2}{48}N\mu_i } \\
        & \leq 2 \exp\rbra*{ - \frac{\eps_1^2}{48}N \muth },
    \end{align}
    where we used the fact that $\mu_i \geq \muth$ for $i \in \calI_\heavy$.
    By the union bound and $| \calI_\heavy| \leq \muth^{-1}$, we have
    \begin{align}
        \Prob\sbra*{ \text{$\exists i \in \calI_\heavy$ such that} \abs*{ \hat{N}_i - N\mu_i } \geq \frac{\eps_1}{4} N \mu_i }
        & \le 2 \abs*{ \calI_\heavy } \exp\rbra*{ - \frac{\eps_1^2}{48}N \muth }\\
        & \le \frac{2}{\muth} \exp\rbra*{ - \frac{\eps_1^2}{48}N \muth }.
        \end{align}
    Due to the condition on $N$, given by~\cref{Eq:Cond_N_Ell1_Distance_1}, it follows that $\abs{ \hat{N}_i - N\mu_i } \leq \eps_1 N \mu_i/4$ for all $i \in \calI_\heavy$ with probability at least $1-\delta_1$.    
        
    We next bound $\Abs{ \barrestrictedmu - \tilde{\mu} }_1$ on the event that $\abs{ \hat{N}_i - N\mu_i } \leq \eps_1 N \mu_i/4$ for all $i \in \calI_\heavy$.
    Write $\hat{N}_i = \rbra{ 1 + \gamma_i } \mu_i N$, where $\abs{ \gamma_i } \leq \eps_1/4$. Then
    \begin{equation}
        \hat{N}\rbra{ \calI_\heavy } = \sum_{i \in \calI_\heavy} \hat{N}_i = \rbra{ 1+ \bar{\gamma} } \mu\rbra{\calI_\heavy}N,
    \end{equation}
    where $\bar{\gamma} \coloneqq \sum_{i \in \calI_\heavy}\gamma_i\mu_i/\mu\rbra{\calI_\heavy}$.
    As $\cbra{ \mu_i/\mu\rbra{\calI_\heavy} }_{i \in \calI_\heavy}$ is a probability distribution, $\abs{ \bar{\gamma} } \leq \eps_1/4$.    
    Using this notation, we obtain
    \begin{equation}
        \frac{\hat{N}_i}{\hat{N}\rbra*{ \calI_\heavy }} = \frac{\mu_i}{\mu\rbra*{ \calI_\heavy }}\frac{1+\gamma_i}{1+\bar{\gamma}}.
    \end{equation}
    
    We can now directly evaluate the $\ell_1$ distance between $\barrestrictedmu$ and $\tilde{\mu}$:
    \begin{align}
        \Abs*{ \barrestrictedmu - \tilde{\mu} }_1 &= \sum_{i \in \calI_\heavy} \abs*{ \frac{\mu_i}{\mu\rbra*{\calI_\heavy}} - \frac{\hat{N}_i}{\hat{N}\rbra*{\calI_\heavy}} }\\
        & = \sum_{i \in \calI_\heavy} \frac{\mu_i}{\mu\rbra*{ \calI_\heavy }} \abs*{ 1 - \frac{1+\gamma_i}{1+\bar{\gamma}} }\\
        & \leq \sum_{i \in \calI_\heavy} \frac{\mu_i}{\mu\rbra*{ \calI_\heavy }} \frac{\abs*{ \bar{\gamma} - \gamma_i }}{1+\bar{\gamma}}\\
        & \leq \sum_{i \in \calI_\heavy} \frac{\mu_i}{\mu\rbra*{ \calI_\heavy }} \frac{\abs*{ \bar{\gamma} } + \abs*{ \gamma_i }}{1+\bar{\gamma}}\\
        & \leq \frac{2\eps_1}{4-\eps_1}\\
        & \leq \eps_1,
    \end{align}
    where we have used $\abs{ \bar{\gamma} } \leq \eps_1/4$ in the penultimate line.
\end{proof}

\subsubsection{\texorpdfstring{Preliminaries for~\cref{Lemma:Ell1_Distance_2,Lemma:Ell1_Distance_3}}{Preliminaries for Lemmas}}

Before proving~\cref{Lemma:Ell1_Distance_2,Lemma:Ell1_Distance_3}, we establish a concentration bound for $\hat{N}\rbra{ \calI_\heavy }$.

\begin{lemma}\label[lemma]{Lemma:N_heavy}
    Let $\delta' \in \interval[open]{0}{1}$. If $N$ satisfies
    \begin{equation}
        N \geq \frac{8}{1 - \eta} \ln \frac{1}{\delta'},
    \end{equation}
    then $\hat{N}(\calI_\heavy) \geq  \rbra{ 1-\eta }N/2$ with probability at least $1-\delta'$.
\end{lemma}

\begin{proof}[Proof of~\cref{Lemma:N_heavy}]
    As $\Ex\sbra{ \hat{N}\rbra{ \calI_\heavy } } = \mu\rbra{ \calI_\heavy } N$, we obtain by the multiplicative Chernoff bound that
    \begin{equation}
        \Prob\sbra*{ \hat{N}\rbra*{ \calI_\heavy } < \frac{\mu\rbra*{ \calI_\heavy }}{2}N }
        \leq \exp\rbra*{ -\frac{\mu\rbra*{ \calI_\heavy } N}{8} }.
    \end{equation}    
    Moreover, as $\mu\rbra{ \calI_\heavy } \geq 1-\eta$ from~\cref{Lemma:p_heavy_large}, it follows that
    \begin{align}
        \Prob\sbra*{ \hat{N}\rbra*{ \calI_\heavy } < \frac{1-\eta}{2}N }
        & \leq \Prob\sbra*{ \hat{N}\rbra*{ \calI_\heavy } < \frac{\mu\rbra*{ \calI_\heavy }}{2}N }\\
        & \leq \exp\rbra*{ - \frac{\mu(\calI_\heavy) N}{8} }\\
        & \leq \exp\rbra*{ - \frac{1-\eta}{8}N }\\
        & \leq \delta',
    \end{align}
    where the last line follows by assumption.
\end{proof}

\subsubsection{Proof of \texorpdfstring{\cref{Lemma:Ell1_Distance_2}}{Lemma}}

\EllDistanceB*

\begin{proof}[Proof of~\cref{Lemma:Ell1_Distance_2}]
    Write $\hat{N}_i = t \hat{a}_i + r_i$, where $r_i \in \cbra{0, \ldots, t-1}$.
    By taking the summation over $i \in \calI_\heavy$, we obtain $\hat{N}\rbra{\calI_\heavy} = t\hat{A}\rbra{\calI_\heavy} + R$, where $R = \sum_{i \in \calI_\heavy} r_i$.
    This leads to
    \begin{equation}
        \hat{A}\rbra{\calI_\heavy} = \frac{\hat{N}\rbra{\calI_\heavy} - R}{t}.
    \end{equation}
    Note that 
    \begin{equation}\label{Eq:Domain_R}
        0 \leq R \leq \rbra*{ t-1 } \abs*{ \calI_\heavy }, \text{\ and\ } R \leq \hat{N}\rbra*{ \calI_\heavy }.
    \end{equation}
    Using this expression of $\hat{A}\rbra{\calI_\heavy}$, the distribution $\barrestrictednu = \cbra{\hat{a}_i/\hat{A}\rbra{\calI_\heavy} }_{i \in \calI_\heavy}$ is given by
    \begin{equation}
        \frac{\hat{a}_i}{\hat{A}\rbra{\calI_\heavy}} =  \frac{\hat{N}_i - r_i}{\hat{N}\rbra{\calI_\heavy}-R}.
    \end{equation}
    
    Using the triangle inequality, we have
    \begin{align}
        \Abs*{ \tilde{\mu} - \barrestrictednu }_1
        & = \sum_{i \in \calI_\heavy} \abs*{ \frac{\hat{N}_i}{\hat{N}\rbra*{\calI_\heavy}} - \frac{\hat{a}_i}{\hat{A}\rbra*{\calI_\heavy}} }\\
        & = \sum_{i \in \calI_\heavy} \abs*{ \frac{\hat{N}_iR}{\hat{N}\rbra*{\calI_\heavy}\rbra*{\hat{N}\rbra*{\calI_\heavy}-R}} - \frac{r_i}{\hat{N}\rbra*{\calI_\heavy}-R} }\\
        & \leq \sum_{i \in \calI_\heavy} \frac{\hat{N}_iR}{\hat{N}\rbra*{ \calI_\heavy }\rbra*{ \hat{N}\rbra*{\calI_\heavy}-R }} +\sum_{i \in \calI_\heavy} \frac{r_i}{\hat{N}\rbra*{\calI_\heavy}-R}\\
        & \leq \frac{2R}{\hat{N}\rbra*{ \calI_\heavy } - R}\\
        & \leq \frac{2\rbra*{t-1}\abs*{ \calI_\heavy }}{\hat{N}\rbra*{\calI_\heavy} - \rbra*{t-1} \abs*{ \calI_\heavy }}, \label{Eq:rounding_distance_bound}
    \end{align}
    where we have used~\cref{Eq:Domain_R} in the third and the second-to-last lines. 
        
    As $N \geq \frac{8}{1-\eta} \ln \frac{1}{\delta_2}$ by assumption, it follows from~\cref{Lemma:N_heavy} that, by setting $\delta'$ to $\delta_2$, $\hat{N}\rbra{\calI_\heavy} \geq \rbra{1-\eta}N/2$ with probability at least $1-\delta_2$.
    Furthermore, $\rbra{1-\eta}N/2 \geq 4t/\rbra{\eps_2 \muth}$ by assumption. Thus, under the conditions in the lemma, it holds that
    \begin{equation}
        \hat{N}\rbra{\calI_\heavy} \geq \frac{4t}{\eps_2 \muth},    
    \end{equation} 
    with probability at least $1-\delta_2$. When this is the case, it also holds that
    \begin{equation}\label{Eq:Nheavy_muth}
        \hat{N}\rbra*{ \calI_\heavy } \geq \frac{4t \abs*{ \calI_\heavy } }{\eps_2},
    \end{equation} 
    as $\muth^{-1} \geq \abs{ \calI_\heavy }$ due to~\cref{Lemma:Size_calI}.
    
    The desired statement is obtained by substituting~\cref{Eq:Nheavy_muth} into~\cref{Eq:rounding_distance_bound}:
    \begin{equation}
        \Abs*{ \tilde{\mu} - \barrestrictednu }_1 \leq \frac{2\rbra*{t-1}\abs*{ \calI_\heavy }}{4t\abs*{ \calI_\heavy }/\eps_2 - \rbra*{t-1} \abs*{ \calI_\heavy }} \leq \eps_2.
    \end{equation}   
\end{proof}

\subsubsection{Proof of \texorpdfstring{\cref{Lemma:Ell1_Distance_3}}{Lemma}}

\EllDistanceC*

\begin{proof}[Proof of~\cref{Lemma:Ell1_Distance_3}]
    As $\calI_\heavy \subseteq \hat{\calI}_\heavy$, it follows that
    \begin{align}
        \hat{\nu} - \barrestrictednu
        = \cbra*{ \frac{\hat{a}_i}{\hat{A}} - \frac{\hat{a}_i}{\hat{A}\rbra*{\calI_\heavy}} }_{i \in \calI_\heavy} \bigcup
        \cbra*{ \frac{\hat{a}_i}{\hat{A}} }_{i \in \hat{\calI}_\heavy \setminus \calI_\heavy}
    \end{align}
    Every component in the first collection is non-positive because $\hat{A}\rbra{\calI_\heavy}\leq\hat{A}$.
    Since the signed measure has total mass zero, its $\ell_1$-norm is twice the sum of its positive components. Hence,
    \begin{equation}\label{Eq:A_prime_hat_A}
        \Abs*{ \barrestrictednu - \hat{\nu} }_1 = 2 \frac{A'}{\hat{A}},
    \end{equation}
    where $A' \coloneqq \sum_{i \in \hat{\calI}_\heavy \setminus \calI_\heavy}\hat{a}_i$. Below, we provide an upper bound on $A'$ and a lower bound on $\hat{A}$.
    
    Since $\hat{a}_i = \floor{ \hat{N}_i/t } \leq \hat{N}_i/t$, it follows that
    \begin{equation}
        A' \leq \sum_{i \in \hat{\calI}_\heavy \setminus \calI_\heavy} \frac{\hat{N}_i}{t} \leq \sum_{i \notin \calI_\heavy} \frac{\hat{N}_i}{t} = \frac{N-\hat{N}(\calI_\heavy)}{t}.
    \end{equation}
    Also, we have $\Ex\sbra{ \hat{N}\rbra{ \calI_\heavy } } = \mu\rbra{ \calI_\heavy }N$.
    Using the multiplicative Chernoff bound, it holds for $\beta \geq 1$ that
    \begin{align}
        \Prob\sbra*{ N-\hat{N}\rbra*{ \calI_\heavy } \geq \rbra*{ 1+\beta }\rbra*{ 1-\mu\rbra*{ \calI_\heavy } } N }\\
        \leq \exp\rbra*{ -\frac{\beta}{3} \rbra*{ 1-\mu\rbra*{ \calI_\heavy } } N }.
    \end{align}
    Setting $\beta$ to $\frac{2N\eta}{ \rbra{ 1-\mu\rbra{ \calI_\heavy } } N} -1$ and using $1 - \mu\rbra{ \calI_\heavy } \leq \eta$ from~\cref{Lemma:p_heavy_large}, we obtain
    \begin{equation}
        \Prob\sbra*{ N-\hat{N}\rbra*{ \calI_\heavy } \geq  2N\eta } \leq \exp\rbra*{ -\frac{N\eta}{3} }.
    \end{equation}
    As $N\eta \geq 3 \ln \frac{2}{\delta_3}$ by assumption, $N-\hat{N}\rbra{\calI_\heavy} < 2N\eta$ with probability at least $1 - \delta_3/2$.
    When this is the case, 
    \begin{equation}\label{Eq:Upper_Bound_A_prime}
        A' \leq \frac{2N\eta}{t}.
    \end{equation}
    
    We next consider $\hat{A}$. As $\calI_\heavy \subseteq \hat{\calI}_\heavy$, we have
    \begin{align}
        \hat{A} 
        & = \sum_{i \in \hat{\calI}_\heavy} \hat{a}_i\\
        & \geq \sum_{i \in \calI_\heavy} \hat{a}_i\\
        & \geq \sum_{i \in \calI_\heavy} \floor*{ \frac{\hat{N}_i}{t} }\\
        & \geq \sum_{i \in \calI_\heavy} \rbra*{ \frac{\hat{N}_i}{t} -1 }\\
        & \geq \frac{\hat{N}\rbra*{\calI_\heavy}}{t} - \abs*{ \calI_\heavy }.
    \end{align}
    As $N \geq \frac{8}{1-\eta} \ln \frac{2}{\delta_3}$ by assumption, it follows from~\cref{Lemma:N_heavy} that $\hat{N}\rbra{ \calI_\heavy } \geq \rbra{ 1-\eta }N/2$ with probability at least $1-\delta_3/2$.
    Moreover, as $N \geq \frac{8t}{\rbra{1-\eta}\muth}$ by assumption and $\abs{ \calI_\heavy } \leq \muth^{-1}$ by~\cref{Lemma:Size_calI}, we have $\abs{ \calI_\heavy } \leq \frac{\rbra{1-\eta}N}{8t}$. When $\hat{N}\rbra{\calI_\heavy} \geq \rbra{1-\eta}N/2$, this gives $\abs{ \calI_\heavy } \leq \frac{\hat{N}\rbra{ \calI_\heavy }}{4t}$.
    Combining these, it holds that
    \begin{equation}\label{Eq:Lower_Bound_hat_A}
        \hat{A} \geq \frac{\hat{N}\rbra{ \calI_\heavy }}{t} - \frac{\hat{N}\rbra{ \calI_\heavy }}{4t} = \frac{3\hat{N}(\calI_\heavy)}{4t} \geq \frac{3(1-\eta)}{8} \frac{N}{t}.
    \end{equation}
    
    Substituting~\cref{Eq:Upper_Bound_A_prime,Eq:Lower_Bound_hat_A} into~\cref{Eq:A_prime_hat_A} and using the union bound gives
    \begin{equation}
        \Abs*{ \barrestrictednu - \hat{\nu} }_1 = \frac{2A'}{\hat{A}} \leq \frac{2 \cdot 2N\eta/t}{3\rbra*{1-\eta}N/\rbra*{8t}} = \frac{32\eta}{3\rbra*{1-\eta}},
    \end{equation}
    with probability at least $1-\delta_3$.
\end{proof}

\section*{Acknowledgments}

J.~B.\ was supported by the Quantum Advantage Pathfinder (QAP) project of UKRI Engineering and Physical Sciences Research Council under grant No.~EP/X026167/1.
W.~F.\ was supported by the Engineering and Physical Sciences Research Council grant EP/X025551/1.
Y.~N.\ was supported by JST PRESTO Grant Number JPMJPR2456 and JST CREST Grant Number JPMJCR23I3.
Q.~W.\ was supported by startup funding from Shanghai Jiao Tong University. 

\addcontentsline{toc}{section}{References}

\bibliographystyle{alphaurl}
\bibliography{main}

@article{RS09,
    author = {Roy, Aidan and Scott, A. J.},
    title = {Unitary designs and codes},
    doi = {10.1007/s10623-009-9290-2},
    journal = {Designs, Codes and Cryptography},
    volume = {53},
    pages = {13--31},
    year = {2009},
}

@article{RS07,
  author        = {Roy, Aidan and Scott, A. J.},
  title         = {Weighted Complex Projective 2-Designs from Bases: Optimal State Determination by Orthogonal Measurements},
  journal       = {Journal of Mathematical Physics},
  volume        = {48},
  number        = {7},
  pages         = {072110},
  year          = {2007},
  doi           = {10.1063/1.2748617},
}

@article{NHMW17,
    author = {Nakata, Yoshifumi and Hirche, Christoph and Morgan, Ciara and Winter, Andreas},
    title = {Decoupling with random diagonal unitaries},
    doi = {10.22331/q-2017-07-21-18},
    journal = {Quantum},
    volume = {1},
    pages = {18},
    year = {2017}
}

@inproceedings{AE07,
  author        = {Ambainis, Andris and Emerson, Joseph},
  title         = {Quantum {$t$}-Designs: {$t$}-Wise Independence in the Quantum World},
  booktitle     = {Proceedings of the 22nd Annual IEEE Conference on Computational Complexity},
  pages         = {129--140},
  year          = {2007},
  doi           = {10.1109/CCC.2007.26},
}

@article{HC22,
  author        = {Ho, Wen Wei and Choi, Soonwon},
  title         = {Exact emergent quantum state designs from quantum chaotic dynamics},
  journal       = {Physical Review Letters},
  volume        = {128},
  pages         = {060601},
  year          = {2022},
  doi           = {10.1103/PhysRevLett.128.060601},
}

@article{CL22,
  author        = {Claeys, Pieter W. and Lamacraft, Austen},
  title         = {Emergent quantum state designs and biunitarity in dual-unitary circuit dynamics},
  journal       = {Quantum},
  volume        = {6},
  pages         = {738},
  year          = {2022},
  doi           = {10.22331/q-2022-06-15-738},
}

@article{Varikuti:2024xeq,
  author        = {Varikuti, Naga Dileep and Bandyopadhyay, Soumik},
  title         = {Unraveling the emergence of quantum state designs in systems with symmetry},
  journal       = {Quantum},
  volume        = {8},
  pages         = {1456},
  year          = {2024},
  doi           = {10.22331/q-2024-08-29-1456},
}

@article{MCR26,
  author        = {Mandal, Saptarshi and Claeys, Pieter W. and Roy, Sthitadhi},
  title         = {Partial projected ensembles and spatiotemporal structure of information scrambling},
  journal       = {Physical Review B},
  volume        = {113},
  number        = {2},
  pages         = {024303},
  year          = {2026},
  doi           = {10.1103/h2q2-yfqs},
}

@article{MSE+24,
  author        = {Mark, Daniel K. and Surace, Federica and Elben, Andreas and Shaw, Adam L. and Choi, Joonhee and Refael, Gil and Endres, Manuel and Choi, Soonwon},
  title         = {Maximum Entropy Principle in Deep Thermalization and in Hilbert-Space Ergodicity},
  journal       = {Physical Review X},
  volume        = {14},
  number        = {4},
  pages         = {041051},
  year          = {2024},
  doi           = {10.1103/PhysRevX.14.041051},
}

@article{CMH+23,
  author        = {Cotler, Jordan S. and Mark, Daniel K. and Huang, Hsin-Yuan and Hern{\'a}ndez, Felipe and Choi, Joonhee and Shaw, Adam L. and Endres, Manuel and Choi, Soonwon},
  title         = {Emergent Quantum State Designs from Individual Many-Body Wave Functions},
  journal       = {PRX Quantum},
  volume        = {4},
  number        = {1},
  pages         = {010311},
  year          = {2023},
  doi           = {10.1103/PRXQuantum.4.010311},
}

@article{Safranek:2019nwk,
  author        = {{\v{S}}afr{\'a}nek, Dominik and Deutsch, J. M. and Aguirre, Anthony},
  title         = {Quantum coarse-grained entropy and thermodynamics},
  journal       = {Physical Review A},
  volume        = {99},
  pages         = {010101},
  year          = {2019},
  doi           = {10.1103/PhysRevA.99.010101},
}

@article{Safranek:2019wml,
  author        = {{\v{S}}afr{\'a}nek, Dominik and Deutsch, J. M. and Aguirre, Anthony},
  title         = {Quantum coarse-grained entropy and thermalization in closed systems},
  journal       = {Physical Review A},
  volume        = {99},
  pages         = {012103},
  year          = {2019},
  doi           = {10.1103/PhysRevA.99.012103},
}

@article{Safranek:2020tgg,
  author        = {{\v{S}}afr{\'a}nek, Dominik and Aguirre, Anthony and Schindler, Joseph and Deutsch, J. M.},
  title         = {A Brief Introduction to Observational Entropy},
  journal       = {Foundations of Physics},
  volume        = {51},
  number        = {5},
  pages         = {101},
  year          = {2021},
  doi           = {10.1007/s10701-021-00498-x},
}

@article{Sinha:2023rwr,
  author        = {Sinha, Shivam and Majumdar, Nripendra and Aravinda, S.},
  title         = {Generalized $\alpha $-observational entropy},
  journal       = {Quantum Information Processing},
  volume        = {24},
  number        = {6},
  pages         = {189},
  year          = {2025},
  doi           = {10.1007/s11128-025-04785-8},
}

@article{CSM+23,
  author  = {Choi, Joonhee and Shaw, Adam L. and Madjarov, Ivaylo S. and Xie, Xin and Finkelstein, Ran and Covey, Jacob P. and Cotler, Jordan S. and Mark, Daniel K. and Huang, Hsin-Yuan and Kale, Anant and Pichler, Hannes and Brand\~{a}o, Fernando G.S.L. and Choi, Soonwon and Endres, Manuel},
  title   = {Preparing random states and benchmarking with many-body quantum chaos},
  journal = {Nature},
  volume  = {613},
  number  = {7944},
  pages   = {468--473},
  year    = {2023},
  doi     = {10.1038/s41586-022-05442-1},
}

@article{SCC+24,
  author  = {Shaw, Adam L. and Chen, Zhuo and Choi, Joonhee and Mark, Daniel K. and Scholl, Pascal and Finkelstein, Ran and Elben, Andreas and Choi, Soonwon and Endres, Manuel},
  title   = {Benchmarking highly entangled states on a 60-atom analogue quantum simulator},
  journal = {Nature},
  volume  = {628},
  pages   = {71--77},
  year    = {2024},
  doi     = {10.1038/s41586-024-07173-x},
}

@article{MHS+25,
  author    = {Mok, Wai-Keong and Haug, Tobias and Shaw, Adam L. and Endres, Manuel and Preskill, John},
  title     = {Optimal Conversion from Classical to Quantum Randomness via Quantum Chaos},
  journal   = {Physical Review Letters},
  volume    = {134},
  pages     = {180403},
  year      = {2025},
  doi       = {10.1103/PhysRevLett.134.180403},
  issue     = {18},
  numpages  = {7},
}

@article{SJA19,
  author  = {Sim, Sukin and Johnson, Peter D. and Aspuru-Guzik, Al{\'a}n},
  title   = {Expressibility and Entangling Capability of Parameterized Quantum Circuits for Hybrid Quantum-Classical Algorithms},
  journal = {Advanced Quantum Technologies},
  volume  = {2},
  number  = {12},
  pages   = {1900070},
  year    = {2019},
  doi     = {10.1002/qute.201900070},
}

@article{HSCC22,
  author    = {Holmes, Zo{\"e} and Sharma, Kunal and Cerezo, Marco and Coles, Patrick J.},
  title     = {Connecting Ansatz Expressibility to Gradient Magnitudes and Barren Plateaus},
  journal   = {PRX Quantum},
  volume    = {3},
  pages     = {010313},
  year      = {2022},
  publisher = {American Physical Society},
  doi       = {10.1103/PRXQuantum.3.010313},
}

@misc{GP22,
  author       = {Gily{\'e}n, Andr{\'a}s and Poremba, Alexander},
  title        = {Improved quantum algorithms for fidelity estimation},
  year         = {2022},
  eprint       = {2203.15993},
  howpublished = {arXiv preprint},
}

@inproceedings{GSLW19,
  author        = {Gily{\'e}n, Andr{\'a}s and Su, Yuan and Low, Guang Hao and Wiebe, Nathan},
  title         = {Quantum singular value transformation and beyond: Exponential improvements for quantum matrix arithmetics},
  booktitle     = {Proceedings of the 51st Annual ACM SIGACT Symposium on Theory of Computing},
  pages         = {193--204},
  year          = {2019},
  doi           = {10.1145/3313276.3316366},
}

@article{RASW23,
  author   = {Rethinasamy, Soorya and Agarwal, Rochisha and Sharma, Kunal and Wilde, Mark M.},
  title    = {Estimating distinguishability measures on quantum computers},
  journal  = {Physical Review A},
  volume   = {108},
  number   = {1},
  pages    = {012409},
  year     = {2023},
  doi      = {10.1103/PhysRevA.108.012409},
  numpages = {36},
}

@article{WZ24,
  author    = {Wang, Qisheng and Zhang, Zhicheng},
  title     = {Fast quantum algorithms for trace distance estimation},
  journal   = {IEEE Transactions on Information Theory},
  volume    = {70},
  number    = {4},
  pages     = {2720--2733},
  year      = {2024},
  publisher = {IEEE},
  doi       = {10.1109/TIT.2023.3321121},
}

@incollection{BHMT02,
  author    = {Brassard, Gilles and H{\o}yer, Peter and Mosca, Michele and Tapp, Alain},
  title     = {Quantum amplitude amplification and estimation},
  booktitle = {Quantum Computation and Information},
  series    = {Contemporary Mathematics},
  volume    = {305},
  number    = {},
  pages     = {53--74},
  year      = {2002},
  publisher = {AMS},
  doi       = {10.1090/conm/305/05215},
}

@inproceedings{Wat02,
  author    = {Watrous, John},
  title     = {Limits on the power of quantum statistical zero-knowledge},
  booktitle = {Proceedings of the 43rd Annual IEEE Symposium on Foundations of Computer Science},
  pages     = {459--468},
  year      = {2002},
  doi       = {10.1109/SFCS.2002.1181970},
}

@article{Wan24,
  author  = {Wang, Qisheng},
  title   = {Optimal trace distance and fidelity estimations for pure quantum states},
  journal = {IEEE Transactions on Information Theory},
  volume  = {70},
  number  = {12},
  pages   = {8791--8805},
  year    = {2024},
  doi     = {10.1109/TIT.2024.3447915},
}

@article{WGL+24,
  author  = {Wang, Qisheng and Guan, Ji and Liu, Junyi and Zhang, Zhicheng and Ying, Mingsheng},
  title   = {New Quantum Algorithms for Computing Quantum Entropies and Distances},
  journal = {IEEE Transactions on Information Theory},
  volume  = {70},
  number  = {8},
  pages   = {5653--5680},
  year    = {2024},
  doi     = {10.1109/TIT.2024.3399014},
}

@article{WZC+23,
  author  = {Wang, Qisheng and Zhang, Zhicheng and Chen, Kean and Guan, Ji and Fang, Wang and Liu, Junyi and Ying, Mingsheng},
  title   = {Quantum Algorithm for Fidelity Estimation},
  journal = {IEEE Transactions on Information Theory},
  volume  = {69},
  number  = {1},
  pages   = {273--282},
  year    = {2023},
  doi     = {10.1109/TIT.2022.3203985},
}

@inproceedings{FW25,
  author    = {Fang, Wang and Wang, Qisheng},
  title     = {Optimal quantum algorithm for estimating fidelity to a pure state},
  booktitle = {Proceedings of the 33rd Annual European Symposium on Algorithms},
  pages     = {4:1--4:12},
  year      = {2025},
  doi       = {10.4230/LIPIcs.ESA.2025.4},
}

@article{BASTS10,
  author  = {Ben-Aroya, Avraham and Schwartz, Oded and Ta-Shma, Amnon},
  title   = {Quantum expanders: Motivation and construction},
  journal = {Theory of Computing},
  volume  = {6},
  pages   = {47--79},
  year    = {2010},
  doi     = {10.4086/toc.2010.v006a003},
}

@inproceedings{GL20,
  author    = {Gily{\'e}n, Andr{\'a}s and Li, Tongyang},
  title     = {Distributional property testing in a quantum world},
  booktitle = {Proceedings of the 11th Innovations in Theoretical Computer Science Conference},
  pages     = {25:1--25:19},
  year      = {2020},
  doi       = {10.4230/LIPIcs.ITCS.2020.25},
}

@article{KLL+17,
  author  = {Kimmel, Shelby and Lin, Cedric Yen-Yu and Low, Guang Hao and Ozols, Maris and Yoder, Theodore J.},
  title   = {Hamiltonian simulation with optimal sample complexity},
  journal = {npj Quantum Information},
  volume  = {3},
  number  = {1},
  pages   = {13},
  year    = {2017},
  doi     = {10.1038/s41534-017-0013-7},
}

@misc{GHS21,
  author       = {Gur, Tom and Hsieh, Min-Hsiu and Subramanian, Sathyawageeswar},
  title        = {Sublinear quantum algorithms for estimating {von Neumann} entropy},
  year         = {2021},
  eprint       = {2111.11139},
  howpublished = {arXiv preprint},
}

@article{SH21,
  author  = {Subramanian, Sathyawageeswar and Hsieh, Min-Hsiu},
  title   = {Quantum algorithm for estimating $\alpha$-{Renyi} entropies of quantum states},
  journal = {Physical Review A},
  volume  = {104},
  number  = {2},
  pages   = {022428},
  year    = {2021},
  doi     = {10.1103/PhysRevA.104.022428},
}

@misc{Wan25,
  author       = {Wang, Qisheng},
  title        = {Information-theoretic lower bounds for approximating monomials via optimal quantum {Tsallis} entropy estimation},
  year         = {2025},
  eprint       = {2509.03496},
  howpublished = {arXiv preprint},
}

@article{EAO+02,
  author  = {Ekert, Artur K. and Alves, Carolina Moura and Oi, Daniel K. L. and Horodecki, Micha{\l} and Horodecki, Pawe{\l} and Kwek, L. C.},
  title   = {Direct estimations of linear and nonlinear functionals of a quantum state},
  journal = {Physical Review Letters},
  volume  = {88},
  number  = {21},
  pages   = {217901},
  year    = {2002},
  doi     = {10.1103/PhysRevLett.88.217901},
}

@misc{CWYZ26,
  author={Chen, Kean and Wang, Qisheng and Yu, Zhan and Zhang, Zhicheng},
  journal={IEEE Transactions on Information Theory}, 
  title={Simultaneous Estimation of Nonlinear Functionals of a Quantum State}, 
  year={2026},
  volume={},
  number={},
  pages={1-1},
  doi={10.1109/TIT.2026.3699531}
}

@article{Liu25,
  author  = {Liu, Yupan},
  title   = {Quantum state testing beyond the polarizing regime and quantum triangular discrimination},
  journal = {Computational Complexity},
  volume  = {34},
  number  = {},
  pages   = {11},
  year    = {2025},
  doi     = {10.1007/s00037-025-00273-8},
}

@article{Hoe63,
  author  = {Hoeffding, Wassily},
  title   = {Probability inequalities for sums of bounded random variables},
  journal = {Journal of the American Statistical Association},
  volume  = {58},
  number  = {301},
  pages   = {13--30},
  year    = {1963},
  doi     = {10.1080/01621459.1963.10500830},
}

@article{NTKD25,
  author    = {Nakata, Yoshifumi and Takeuchi, Yuki and Kliesch, Martin and Darmawan, Andrew},
  title     = {Computational Complexity of Unitary and State Design Properties},
  journal   = {PRX Quantum},
  volume    = {6},
  number    = {3},
  pages     = {030345},
  year      = {2025},
  month     = {Sep},
  publisher = {American Physical Society},
  doi       = {10.1103/21vm-bz3t},
  numpages  = {36},
}

@article{Mel24,
  author    = {Mele, Antonio Anna},
  title     = {Introduction to {H}aar Measure Tools in Quantum Information: {A} Beginner's Tutorial},
  journal   = {Quantum},
  volume    = {8},
  pages     = {1340},
  year      = {2024},
  doi       = {10.22331/q-2024-05-08-1340},
}

@article{SV14,
  author  = {Sushant Sachdeva and Nisheeth K. Vishnoi},
  title   = {Faster Algorithms via Approximation Theory},
  journal = {Foundations and Trends in Theoretical Computer Science},
  volume  = {9},
  number  = {2},
  pages   = {125-210},
  year    = {2014},
  doi     = {10.1561/0400000065},
}

@inproceedings{Bel19,
  author    = {Belovs, Aleksandrs},
  title     = {Quantum algorithms for classical probability distributions},
  booktitle = {Proceedings of the 27th Annual European Symposium on Algorithms},
  pages     = {16:1--16:11},
  year      = {2019},
  doi       = {10.4230/LIPIcs.ESA.2019.16},
}

@article{Zhu22,
  author  = {Zhu, Huangjun},
  title   = {Quantum measurements in the light of quantum state estimation},
  journal = {PRX Quantum},
  volume  = {3},
  number  = {3},
  pages   = {030306},
  year    = {2022},
  doi     = {10.1103/PRXQuantum.3.030306},
}

@article{ATS07,
  author  = {Aharonov, Dorit and Ta-Shma, Amnon},
  title   = {Adiabatic quantum state generation},
  journal = {SIAM Journal on Computing},
  volume  = {37},
  number  = {1},
  pages   = {47--82},
  year    = {2007},
  doi     = {10.1137/060648829},
}

@article{AS05,
  author  = {Atici, Alp and Servedio, Rocco A.},
  title   = {Improved bounds on quantum learning algorithms},
  journal = {Quantum Information Processing},
  volume  = {4},
  number  = {5},
  pages   = {355--386},
  year    = {2005},
  doi     = {10.1007/s11128-005-0001-2},
}

@article{Hel67,
  author  = {Helstrom, Carl W.},
  title   = {Detection theory and quantum mechanics},
  journal = {Information and Control},
  volume  = {10},
  number  = {3},
  pages   = {254--291},
  year    = {1967},
  doi     = {10.1016/S0019-9958(67)90302-6},
}

@article{Hol73,
  author  = {Holevo, Alexander S.},
  title   = {Statistical decision theory for quantum systems},
  journal = {Journal of Multivariate Analysis},
  volume  = {3},
  number  = {4},
  pages   = {337--394},
  year    = {1973},
  doi     = {10.1016/0047-259X(73)90028-6},
}

@inproceedings{vACGN23,
  author    = {van Apeldoorn, Joran and Cornelissen, Arjan and Gily{\'{e}}n, Andr{\'{a}}s and Nannicini, Giacomo},
  title     = {Quantum tomography using state-preparation unitaries},
  booktitle = {Proceedings of the 2023 Annual ACM-SIAM Symposium on Discrete Algorithms},
  pages     = {1265--1318},
  year      = {2023},
  doi       = {10.1137/1.9781611977554.ch47},
}

@article{BJ98,
  author  = {Bshouty, Nader H. and Jackson, Jeffrey C.},
  title   = {Learning {DNF} over the Uniform Distribution Using a Quantum Example Oracle},
  journal = {SIAM Journal on Computing},
  volume  = {28},
  number  = {3},
  pages   = {1136--1153},
  year    = {1998},
  doi     = {10.1137/S0097539795293123},
}

@misc{CWZ25,
  author       = {Chen, Kean and Wang, Qisheng and Zhang, Zhicheng},
  title        = {A list of complexity bounds for property testing by quantum sample-to-query lifting},
  year         = {2025},
  eprint       = {2512.01971},
  howpublished = {arXiv preprint},
}

@article{IH23,
  author  = {Ippoliti, Matteo and Ho, Wen Wei},
  title   = {Dynamical Purification and the Emergence of Quantum State Designs from the Projected Ensemble},
  journal = {PRX Quantum},
  volume  = {4},
  pages   = {030322},
  year    = {2023},
  doi     = {10.1103/PRXQuantum.4.030322},
  issue   = {3},
}

@inproceedings{JLS18,
  author    = {Ji, Zhengfeng and Liu, Yi-Kai and Song, Fang},
  title     = {Pseudorandom quantum states},
  booktitle = {Proceedings of the Advances in Cryptology -- CRYPTO 2018},
  volume    = {10993},
  pages     = {126--152},
  year      = {2018},
  doi       = {10.1007/978-3-319-96878-0_5},
}

@article{HLSW04,
  author  = {Hayden, Patrick and Leung, Debbie and Shor, Peter W. and Winter, Andreas},
  title   = {Randomizing quantum states: Constructions and applications},
  journal = {Communications in Mathematical Physics},
  volume  = {250},
  number  = {2},
  pages   = {371--391},
  year    = {2004},
  doi     = {10.1007/s00220-004-1087-6},
}

@inproceedings{AS04,
  author    = {Ambainis, Andris and Smith, Adam},
  title     = {Small pseudo-random families of matrices: {Derandomizing} approximate quantum encryption},
  booktitle = {Proceedings of the 8th International Workshop on Approximation, Randomization, and Combinatorial Optimization},
  volume    = {3122},
  pages     = {249--260},
  year      = {2004},
  doi       = {10.1007/978-3-540-27821-4_23},
}

@inproceedings{Kre21,
  author    = {Kretschmer, William},
  title     = {Quantum pseudorandomness and classical complexity},
  booktitle = {Proceedings of the 16th Conference on the Theory of Quantum Computation, Communication and Cryptography},
  volume    = {197},
  pages     = {2:1--2:20},
  year      = {2021},
  doi       = {10.4230/LIPIcs.TQC.2021.2},
}

@inproceedings{MY22,
  author    = {Morimae, Tomoyuki and Yamakawa, Takashi},
  title     = {Quantum commitments and signatures without one-way functions},
  booktitle = {Proceedings of the Advances in Cryptology -- CRYPTO 2022},
  volume    = {13507},
  pages     = {269--295},
  year      = {2022},
  doi       = {10.1007/978-3-031-15802-5_10},
}

@inproceedings{AQY22,
  author    = {Ananth, Prabhanjan and Qian, Luowen and Yuen, Henry},
  title     = {Cryptography from pseudorandom quantum states},
  booktitle = {Advances in Cryptology -- CRYPTO 2022},
  volume    = {13507},
  pages     = {208--236},
  year      = {2022},
  doi       = {10.1007/978-3-031-15802-5_8},
}

@inproceedings{MH25,
  author    = {Ma, Fermi and Huang, Hsin-Yuan},
  title     = {How to construct random unitaries},
  booktitle = {Proceedings of the 57th Annual ACM Symposium on Theory of Computing},
  pages     = {806--809},
  year      = {2025},
  doi       = {10.1145/3717823.3718254},
}

@article{SHH25,
  author  = {Schuster, Thomas and Haferkamp, Jonas and Huang, Hsin-Yuan},
  title   = {Random unitaries in extremely low depth},
  journal = {Science},
  volume  = {389},
  number  = {6755},
  pages   = {92--96},
  year    = {2025},
  doi     = {10.1126/science.adv8590},
}

@inproceedings{Sen06,
  author    = {Sen, Pranab},
  title     = {Random measurement bases, quantum state distinction and applications to the hidden subgroup problem},
  booktitle = {Proceedings of the 21st Annual IEEE Conference on Computational Complexity},
  pages     = {274--287},
  year      = {2006},
  doi       = {10.1109/CCC.2006.37},
}

@article{BH13,
  author  = {Brand{\~a}o, Fernando G. S. L. and Horodecki, Micha{\l}},
  title   = {Exponential quantum speed-ups are generic},
  journal = {Quantum Information and Computation},
  volume  = {13},
  number  = {11-12},
  pages   = {901--924},
  year    = {2013},
  doi     = {10.26421/QIC13.11-12-1},
}

@article{BIS+18,
  author  = {Boixo, Sergio and Isakov, Sergei V. and Smelyanskiy, Vadim N. and Babbush, Ryan and Ding, Nan and Jiang, Zhang and Bremner, Michael J. and Martinis, John M. and Neven, Hartmut},
  title   = {Characterizing quantum supremacy in near-term devices},
  journal = {Nature Physics},
  volume  = {14},
  number  = {6},
  pages   = {595--600},
  year    = {2018},
  doi     = {10.1038/s41567-018-0124-x},
}

@article{AAB+19,
  author  = {Arute, Frank and Arya, Kunal and Babbush, Ryan and Bacon, Dave and Bardin, Joseph C. and Barends, Rami and Biswas, Rupak and Boixo, Sergio and Brandao, Fernando G. S. L. and Buell, David A. and Burkett, Brian and Chen, Yu and Chen, Zijun and Chiaro, Ben and Collins, Roberto and Courtney, William and Dunsworth, Andrew and Farhi, Edward and Foxen, Brooks and Fowler, Austin and Gidney, Craig and Giustina, Marissa and Graff, Rob and Guerin, Keith and Habegger, Steve and Harrigan, Matthew P. and Hartmann, Michael J. and Ho, Alan and Hoffmann, Markus and Huang, Trent and Humble, Travis S. and Isakov, Sergei V. and Jeffrey, Evan and Jiang, Zhang and Kafri, Dvir and Kechedzhi, Kostyantyn and Kelly, Julian and Klimov, Paul V. and Knysh, Sergey and Korotkov, Alexander and Kostritsa, Fedor and Landhuis, David and Lindmark, Mike and Lucero, Erik and Lyakh, Dmitry and Mandrà, Salvatore and McClean, Jarrod R. and McEwen, Matthew and Megrant, Anthony and Mi, Xiao and Michielsen, Kristel and Mohseni, Masoud and Mutus, Josh and Naaman, Ofer and Neeley, Matthew and Neill, Charles and Niu, Murphy Yuezhen and Ostby, Eric and Petukhov, Andre and Platt, John C. and Quintana, Chris and Rieffel, Eleanor G. and Roushan, Pedram and Rubin, Nicholas C. and Sank, Daniel and Satzinger, Kevin J. and Smelyanskiy, Vadim and Sung, Kevin J. and Trevithick, Matthew D. and Vainsencher, Amit and Villalonga, Benjamin and White, Theodore and Yao, Z. Jamie and Yeh, Ping and Zalcman, Adam and Neven, Hartmut and Martinis, John M},
  title   = {Quantum supremacy using a programmable superconducting processor},
  journal = {Nature},
  volume  = {574},
  number  = {7779},
  pages   = {505--510},
  year    = {2019},
  doi     = {10.1038/s41586-019-1666-5},
}

@article{BFNV19,
  author  = {Bouland, Adam and Fefferman, Bill and Nirkhe, Chinmay and Vazirani, Umesh},
  title   = {On the complexity and verification of quantum random circuit sampling},
  journal = {Nature Physics},
  volume  = {15},
  number  = {2},
  pages   = {159--163},
  year    = {2019},
  doi     = {10.1038/s41567-018-0318-2},
}

@article{KRT17,
  author  = {Kueng, Richard and Rauhut, Holger and Terstiege, Ulrich},
  title   = {Low rank matrix recovery from rank one measurements},
  journal = {Applied and Computational Harmonic Analysis},
  volume  = {42},
  number  = {1},
  pages   = {88--116},
  year    = {2017},
  doi     = {10.1016/j.acha.2015.07.007},
}

@inproceedings{KL17,
  author    = {Kimmel, Shelby and Liu, Yi-Kai},
  title     = {Phase retrieval using unitary 2-designs},
  booktitle = {Proceedings of the 2017 International Conference on Sampling Theory and Applications},
  pages     = {345--349},
  year      = {2017},
  doi       = {10.1109/SAMPTA.2017.8024414},
}

@misc{KZG16,
  author       = {Kueng, Richard and Zhu, Huangjun and Gross, David},
  title        = {Distinguishing quantum states using Clifford orbits},
  year         = {2016},
  eprint       = {1609.08595},
  howpublished = {arXiv preprint},
}

@article{OAG+16,
  author  = {Oszmaniec, Micha{\l} and Augusiak, Remigiusz and Gogolin, Christian and Ko{\l}ody{\'n}ski, Jan and Ac{\'\i}n, Antonio and Lewenstein, Maciej},
  title   = {Random bosonic states for robust quantum metrology},
  journal = {Physical Review X},
  volume  = {6},
  number  = {4},
  pages   = {041044},
  year    = {2016},
  doi     = {10.1103/PhysRevX.6.041044},
}

@article{Dev05,
  author  = {Devetak, Igor},
  title   = {The private classical capacity and quantum capacity of a quantum channel},
  journal = {IEEE Transactions on Information Theory},
  volume  = {51},
  number  = {1},
  pages   = {44--55},
  year    = {2005},
  doi     = {10.1109/TIT.2004.839515},
}

@article{DW04,
  author  = {Devetak, Igor and Winter, Andreas},
  title   = {Relating quantum privacy and quantum coherence: {An} operational approach},
  journal = {Physical Review Letters},
  volume  = {93},
  number  = {8},
  pages   = {080501},
  year    = {2004},
  doi     = {10.1103/PhysRevLett.93.080501},
}

@article{ADHW09,
  author  = {Abeyesinghe, Anura and Devetak, Igor and Hayden, Patrick and Winter, Andreas},
  title   = {The mother of all protocols: restructuring quantum information's family tree},
  journal = {Proceedings of the Royal Society A},
  volume  = {465},
  number  = {2108},
  pages   = {2537--2563},
  year    = {2009},
  doi     = {10.1098/rspa.2009.0202},
}

@article{DBWR14,
  author  = {Dupuis, Fr{\'e}d{\'e}ric and Berta, Mario and Wullschleger, J{\"u}rg and Renner, Renato},
  title   = {One-shot decoupling},
  journal = {Communications in Mathematical Physics},
  volume  = {328},
  number  = {1},
  pages   = {251--284},
  year    = {2014},
  doi     = {10.1007/s00220-014-1990-4},
}

@article{SDTR13,
  author  = {Szehr, Oleg and Dupuis, Fr{\'e}d{\'e}ric and Tomamichel, Marco and Renner, Renato},
  title   = {Decoupling with unitary approximate two-designs},
  journal = {New Journal of Physics},
  volume  = {15},
  number  = {5},
  pages   = {053022},
  year    = {2013},
  doi     = {10.1088/1367-2630/15/5/053022},
}

@article{HOW05,
  author  = {Horodecki, Micha{\l} and Oppenheim, Jonathan and Winter, Andreas},
  title   = {Partial quantum information},
  journal = {Nature},
  volume  = {436},
  number  = {7051},
  pages   = {673--676},
  year    = {2005},
  doi     = {10.1038/nature03909},
}

@article{HOW07,
  author  = {Horodecki, Micha{\l} and Oppenheim, Jonathan and Winter, Andreas},
  title   = {Quantum state merging and negative information},
  journal = {Communications in Mathematical Physics},
  volume  = {269},
  number  = {1},
  pages   = {107--136},
  year    = {2007},
  doi     = {10.1007/s00220-006-0118-x},
}

@article{NWY21,
  author  = {Nakata, Yoshifumi and Wakakuwa, Eyuri and Yamasaki, Hayata},
  title   = {One-shot quantum error correction of classical and quantum information},
  journal = {Physical Review A},
  volume  = {104},
  number  = {1},
  pages   = {012408},
  year    = {2021},
  doi     = {10.1103/PhysRevA.104.012408},
}

@article{WN23,
  author  = {Wakakuwa, Eyuri and Nakata, Yoshifumi},
  title   = {One-shot triple-resource trade-off in quantum channel coding},
  journal = {IEEE Transactions on Information Theory},
  volume  = {69},
  number  = {4},
  pages   = {2400--2426},
  year    = {2023},
  doi     = {10.1109/TIT.2022.3222775},
}

@article{PSW06,
  title={Entanglement and the foundations of statistical mechanics},
  author={Popescu, Sandu and Short, Anthony J. and Winter, Andreas},
  journal={Nature Physics},
  volume={2},
  number={11},
  pages={754--758},
  year={2006},
  doi={10.1038/nphys444},
}

@article{LPSW09,
  title={Quantum mechanical evolution towards thermal equilibrium},
  author={Linden, Noah and Popescu, Sandu and Short, Anthony J. and Winter, Andreas},
  journal={Physical Review E—Statistical, Nonlinear, and Soft Matter Physics},
  volume={79},
  number={6},
  pages={061103},
  year={2009},
  doi={10.1103/PhysRevE.79.061103},
}

@article{DRHRW16,
  title={Relative thermalization},
  author={Del Rio, L{\'\i}dia and Hutter, Adrian and Renner, Renato and Wehner, Stephanie},
  journal={Physical Review E},
  volume={94},
  number={2},
  pages={022104},
  year={2016},
  doi={10.1103/PhysRevE.94.022104},
}

@article{KYI20,
  title={Characterizing complexity of many-body quantum dynamics by higher-order eigenstate thermalization},
  author={Kaneko, Kazuya and Iyoda, Eiki and Sagawa, Takahiro},
  journal={Physical Review A},
  volume={101},
  number={4},
  pages={042126},
  year={2020},
  doi={10.1103/PhysRevA.101.042126},
}

@article{IH22,
  title={Solvable model of deep thermalization with distinct design times},
  author={Ippoliti, Matteo and Ho, Wen Wei},
  journal={Quantum},
  volume={6},
  pages={886},
  year={2022},
  doi={10.22331/q-2022-12-29-886},
}

@article{BDP23,
  title={Deep thermalization in constrained quantum systems},
  author={Bhore, Tanmay and Desaules, Jean-Yves and Papi{\'c}, Zlatko},
  journal={Physical Review B},
  volume={108},
  number={10},
  pages={104317},
  year={2023},
  doi={10.1103/PhysRevB.108.104317},
}

@article{HP07,
  title={Black holes as mirrors: {Q}uantum information in random subsystems},
  author={Hayden, Patrick and Preskill, John},
  journal={Journal of high energy physics},
  volume={2007},
  number={09},
  pages={120--120},
  year={2007},
  doi={10.1088/1126-6708/2007/09/120},
}

@article{NWK23,
  title={Black holes as clouded mirrors: the {Hayden-Preskill} protocol with symmetry},
  author={Nakata, Yoshifumi and Wakakuwa, Eyuri and Koashi, Masato},
  journal={Quantum},
  volume={7},
  pages={928},
  year={2023},
  doi={10.22331/q-2023-02-21-928},
}

@article{NMK25,
  title={Decoding general error correcting codes and the role of complementarity},
  author={Nakata, Yoshifumi and Matsuura, Takaya and Koashi, Masato},
  journal={npj Quantum Information},
  volume={11},
  number={1},
  pages={4},
  year={2025},
  doi={10.1038/s41534-024-00951-5},
}

@article{SS08,
  title={Fast scramblers},
  author={Sekino, Yasuhiro and Susskind, Leonard},
  journal={Journal of High Energy Physics},
  volume={2008},
  number={10},
  pages={065--065},
  year={2008},
  doi={10.1088/1126-6708/2008/10/065},
}

@article{LSH+13,
  title={Towards the fast scrambling conjecture},
  author={Lashkari, Nima and Stanford, Douglas and Hastings, Matthew and Osborne, Tobias and Hayden, Patrick},
  journal={Journal of High Energy Physics},
  volume={2013},
  number={4},
  pages={1--33},
  year={2013},
  publisher={Springer},
  doi={10.1007/JHEP04(2013)022},
}

@article{MSS16,
  title={A bound on chaos},
  author={Maldacena, Juan and Shenker, Stephen H and Stanford, Douglas},
  journal={Journal of High Energy Physics},
  volume={2016},
  number={8},
  pages={106},
  year={2016},
  doi={10.1007/JHEP08(2016)106},
}

@article{RY17,
  title={Chaos and complexity by design},
  author={Roberts, Daniel A. and Yoshida, Beni},
  journal={Journal of High Energy Physics},
  volume={2017},
  number={4},
  pages={1--64},
  year={2017},
  doi={10.1007/JHEP04(2017)121},
}

@article{PCMCH24,
  title={Hilbert-space ergodicity in driven quantum systems: Obstructions and designs},
  author={Pilatowsky-Cameo, Sa{\'u}l and Marvian, Iman and Choi, Soonwon and Ho, Wen Wei},
  journal={Physical Review X},
  volume={14},
  number={4},
  pages={041059},
  year={2024},
  publisher={APS},
  doi={10.1103/PhysRevX.14.041059},
}

@article{MCS+23,
  title={Benchmarking quantum simulators using ergodic quantum dynamics},
  author={Mark, Daniel K. and Choi, Joonhee and Shaw, Adam L. and Endres, Manuel and Choi, Soonwon},
  journal={Physical Review Letters},
  volume={131},
  number={11},
  pages={110601},
  year={2023},
  publisher={APS},
  doi={10.1103/PhysRevLett.131.110601},
}

@article{EAZ05,
  title={Scalable noise estimation with random unitary operators},
  author={Emerson, Joseph and Alicki, Robert and {\.Z}yczkowski, Karol},
  journal={Journal of Optics B: Quantum and Semiclassical Optics},
  volume={7},
  number={10},
  pages={S347--S352},
  year={2005},
  doi={10.1088/1464-4266/7/10/021},
}

@article{KLR+08,
  title={Randomized benchmarking of quantum gates},
  author={Knill, Emanuel and Leibfried, Dietrich and Reichle, Rolf and Britton, Joe and Blakestad, R. Brad and Jost, John D. and Langer, Chris and Ozeri, Roee and Seidelin, Signe and Wineland, David J.},
  journal={Physical Review A},
  volume={77},
  number={1},
  pages={012307},
  year={2008},
  doi={10.1103/PhysRevA.77.012307},
}

@article{MGE11,
  title={Scalable and robust randomized benchmarking of quantum processes},
  author={Magesan, Easwar and Gambetta, Jay M. and Emerson, Joseph},
  journal={Physical Review Letters},
  volume={106},
  number={18},
  pages={180504},
  year={2011},
  doi={10.1103/PhysRevLett.106.180504},
}

@article{MGE12,
  title={Characterizing quantum gates via randomized benchmarking},
  author={Magesan, Easwar and Gambetta, Jay M and Emerson, Joseph},
  journal={Physical Review A},
  volume={85},
  number={4},
  pages={042311},
  year={2012},
  doi={10.1103/PhysRevA.85.042311},
}

@article{KBC+14,
  title={Optimal quantum control using randomized benchmarking},
  author={Kelly, J. and Barends, R. and Campbell, B. and Chen, Y. and Chen, Z. and Chiaro, B. and Dunsworth, A. and Fowler, A. G. and Hoi, I.-C. and Jeffrey, E. and Megrant, A. and Mutus, J. and Neill, C. and O'Malley, P. J. J. and Quintana, C. and Roushan, P. and Sank, D. and Vainsencher, A. and Wenner, J. and White, T. C. and Cleland, A. N. and Martinis, John M.},
  journal={Physical Review Letters},
  volume={112},
  number={24},
  pages={240504},
  year={2014},
  doi={10.1103/PhysRevLett.112.240504},
}

@article{SBM+16,
  title={Characterizing errors on qubit operations via iterative randomized benchmarking},
  author={Sheldon, Sarah and Bishop, Lev S. and Magesan, Easwar and Filipp, Stefan and Chow, Jerry M. and Gambetta, Jay M.},
  journal={Physical Review A},
  volume={93},
  number={1},
  pages={012301},
  year={2016},
  doi={10.1103/PhysRevA.93.012301},
}

@article{GKL+21,
  title={Experimental implementation of non-Clifford interleaved randomized benchmarking with a controlled-{$S$} gate},
  author={Garion, Shelly and Kanazawa, Naoki and Landa, Haggai and McKay, David C. and Sheldon, Sarah and Cross, Andrew W. and Wood, Christopher J.},
  journal={Physical Review Research},
  volume={3},
  number={1},
  pages={013204},
  year={2021},
  doi={10.1103/PhysRevResearch.3.013204},
}

@article{OWE19,
  title={Randomized benchmarking for individual quantum gates},
  author={Onorati, Emilio and Werner, Albert H. and Eisert, Jens},
  journal={Physical Review Letters},
  volume={123},
  number={6},
  pages={060501},
  year={2019},
  doi={10.1103/PhysRevLett.123.060501},
}

@article{HRO+22,
  title={General framework for randomized benchmarking},
  author={Helsen, Jonas and Roth, Ingo and Onorati, Emilio and Werner, Albert H and Eisert, Jens},
  journal={PRX quantum},
  volume={3},
  number={2},
  pages={020357},
  year={2022},
  publisher={APS},
  doi={10.1103/PRXQuantum.3.020357},
}

@misc{HKR22,
  title={Randomized benchmarking with random quantum circuits},
  author={Heinrich, Markus and Kliesch, Martin and Roth, Ingo},
  journal={ArXiv preprints},
  eprint={2212.06181},
  year={2022},
}

@inproceedings{MPSY24,
  title={Simple constructions of linear-depth t-designs and pseudorandom unitaries},
  author={Metger, Tony and Poremba, Alexander and Sinha, Makrand and Yuen, Henry},
  booktitle={Proceedings of the 65th IEEE 65th Annual Symposium on Foundations of Computer Science},
  pages={485--492},
  year={2024},
  doi={10.1109/FOCS61266.2024.00038},
}

@misc{CSBH25,
  title={Unitary designs in nearly optimal depth},
  author={Cui, Laura and Schuster, Thomas and Brandao, Fernando and Huang, Hsin-Yuan},
  year={2025},
  eprint = {2507.06216},
  howpublished = {ArXiv preprints},
}

\end{document}